\documentclass[11pt]{article}
\usepackage{amsmath,amssymb,amsfonts,latexsym,graphicx,amsthm}
\usepackage{fullpage,color}
\usepackage{url,hyperref}
\usepackage{comment}
\usepackage[linesnumbered,boxed,ruled,vlined]{algorithm2e}
\usepackage{framed}
\usepackage[shortlabels]{enumitem}
\usepackage{cleveref}

\usepackage[normalem]{ulem}

\usepackage{tikz}
\usetikzlibrary{shapes.geometric, positioning}

\usepackage{thm-restate}
\usetikzlibrary{calc}
\usetikzlibrary{intersections}
\usepackage{tcolorbox}
\usepackage{subcaption}
\usepackage[margin=1in]{geometry}
\usepackage{ifthen}

\usepackage{xcolor}
\usetikzlibrary{positioning}
\usetikzlibrary{arrows.meta}

\newtheorem{theorem}{Theorem}[section]
\newtheorem{lemma}[theorem]{Lemma}
\newtheorem{corollary}[theorem]{Corollary}

\newtheorem{claim}[theorem]{Claim}
\newtheorem{observation}[theorem]{Observation}
\theoremstyle{definition}
\newtheorem{definition}[theorem]{Definition}

\newtheorem{question}[theorem]{Question}
\newtheorem{invariant}[theorem]{Invariant}
\newtheorem{hypothesis}[theorem]{Hypothesis}

        {\medskip}

\newcommand{\eps}{\epsilon}

\newcommand{\ceil}[1]{\lceil #1 \rceil}
\newcommand{\floor}[1]{\lfloor #1 \rfloor}

\newcommand{\brac}[1]{\left(#1\right)}
\newcommand{\opt}{\mathsf{opt}}

\newcommand{\lmst}{\mathsf{MST}}
\newcommand{\mst}{\mathsf{mst}}

\newcommand{\wts}{\mathbf{w}}

\newcommand{\dist}{\mathsf{dist}}
\newcommand{\real}{\mathbb{R}}

\definecolor{BrickRed}{rgb}{0.8, 0.25, 0.33}
\def\EMPH#1{\emph{\textcolor{BrickRed}{#1}}}

\newcommand{\poly}{\mathrm{poly}}
\newcommand{\nw}{\mathrm{new}}
\newcommand{\old}{\mathrm{old}}
\newcommand{\DP}{\mathsf{DP}}
\newcommand{\DPtree}{\mathcal{T}_{\mathrm{dp}}}
\newcommand{\interval}{\mathcal{I}}
\newcommand{\paths}{\mathcal{P}}

\newcommand{\gr}{\mathrm{gr}}
\newcommand{\spn}{\mathsf{span}}

\begin{document}

	\title{Approximate Light Spanners in Planar Graphs}
    \author{}
	\author{
        Hung Le\thanks{University of Massachusetts Amherst, \href{}{hungle@cs.umass.edu}}\and
		Shay Solomon\thanks{Tel Aviv University, \href{}{solo.shay@gmail.com}}\and	
		Cuong Than\thanks{University of Massachusetts Amherst, \href{}{cthan@umass.edu}}\and
		Csaba D. T\'oth\thanks{California State University Northridge and Tufts University, \href{}{csaba.toth@csun.edu}}\and
		Tianyi Zhang\thanks{Nanjing University, \href{}{tianyiz25@nju.edu.cn}}
	}
	
	\date{}

	\maketitle

\thispagestyle{empty}
\begin{abstract}
		In their seminal paper, Alth\"{o}fer   et al.\ (DCG~1993) introduced the \emph{greedy spanner} and showed that, for any weighted planar graph $G$, the weight of the greedy $(1+\eps)$-spanner is at most $(1+\frac{2}{\epsilon}) \cdot \wts(\lmst(G))$, where $\wts(\lmst(G))$ is the weight of a minimum spanning tree $\lmst(G)$ of $G$. This bound is optimal in an \EMPH{existential sense}: there exist planar graphs $G$ for which any $(1+\eps)$-spanner has a weight of at least $(1+\frac{2}{\epsilon}) \cdot \wts(\lmst(G))$. 
 
         However, as an \EMPH{approximation algorithm}, even for a \EMPH{bicriteria} approximation, the weight approximation factor of the greedy spanner is essentially as large as the existential bound:  
          There exist planar graphs $G$ for which the greedy $(1+x \eps)$-spanner (for any $1\leq x = O(\eps^{-1/2})$) has a weight of
          $\Omega\brac{\frac{1}{\epsilon \cdot x^2}}\cdot \wts(G_{\opt, \eps})$,
          where $G_{\opt, \eps}$ is a $(1+\epsilon)$-spanner of $G$ of minimum weight.

         Despite the flurry of works over the past three decades on approximation algorithms for spanners as well as on light(-weight) spanners, there is still no (possibly bicriteria) approximation algorithm for light spanners in weighted planar graphs that outperforms the existential bound. As our main contribution, we present a polynomial time algorithm for constructing, in any weighted planar graph $G$, a $\brac{1+\epsilon\cdot 2^{O\brac{\log^* 1/\epsilon}}}$-spanner for $G$ of total weight $O(1)\cdot \wts(G_{\opt, \epsilon})$.
        
         To achieve this result, we develop a new technique, which we refer to as \EMPH{iterative planar pruning}. It iteratively modifies a spanner; each iteration replaces a heavy set of edges by a light path, to substantially decrease the total weight of the spanner while only slightly increasing its stretch. We leverage planarity to prove a  \EMPH{laminar} structural property of the edge set to be removed, which enables us to optimize the path to be inserted via dynamic programming. 
         Our technique applies dynamic programming \EMPH{directly to the input planar graph}, which significantly deviates from previous techniques used for network design problems in planar graphs, and might be of independent interest.
	\end{abstract}

\clearpage

\thispagestyle{empty}
\setcounter{tocdepth}{3}
\tableofcontents
\clearpage

\clearpage
\setcounter{page}{1}

\section{Introduction}

A \emph{$t$-spanner} of an edge-weighted undirected graph $G = (V,E,\wts)$ is a subgraph $H$ of $G$ such that $\dist_H(u,v)\leq t\cdot \dist_G(u,v)$ for every  pair $u,v$ of points in $V$, where $\dist_G$ denotes the shortest path distance in $G$ w.r.t the weight function $\wts$. The study of graph spanners was pioneered in the late 1980s by influential results on spanners for general graphs~\cite{PelegU89a,PelegS89} and in low-dimensional spaces~\cite{Chew86,Clarkson87,Keil88}, and it has grown into a very active and vibrant research area; see the survey~\cite{AhmedBSHJKS20}. A central research direction within the area of spanners is the design of approximation algorithms for an \EMPH{optimal}
$t$-spanner of a given graph, where the optimality usually refers to the two most common ``compactness'' measures: the spanner {\em size} (number of edges) or {\em weight} (total edge weight).\footnote{There are also other compactness measures, such as the {\em maximum degree}, but the most common and well-studied measures are the size and weight of spanners and for brevity we will restrict attention to them.}   The \EMPH{greedy} algorithm by Alth\"{o}fer et al.~\cite{althofer1993sparse} gives a $t$-spanner 
of size $O(n^{1 + 2/(t + 1)})$ for odd $t$ and $O(n^{1 + 2/t})$ for even $t$.
Therefore, the approximation ratio of the greedy algorithm (for minimizing the number of edges) is naively bounded by $O(n^{2/(t + 1)})$ for odd $t$ and by $O(n^{2/t})$ for even $t$. The recent line of work on light spanners~\cite{ChandraDNS95,ChechikW18,FiltserN22,LeS22,LeS23} implies almost the same approximation ratio for minimizing the weight of the spanners. 

Beating the approximation ratio of the greedy algorithm has been extremely challenging despite years of effort. Better approximation factors for sparsity are known 
only for $t \in \{2, 3 ,4\}$; specifically $O(\log n)$~\cite{KP94} for $t=2$ (matching a known lower bound~\cite{Kor01}) and $\tilde{O}(n^{1/3})$ for $t=3$~\cite{BBM+11} and $t = 4$~\cite{DZ16}. The gap between the upper and lower bounds is fairly large: for any $t \geq 3$ and for any constant $\eps>0$, there is no polynomial-time algorithm approximating a $t$-spanner with ratio better than $2^{\log{n}^{1 - \epsilon}/t}$ assuming $NP \not \subseteq BPTIME(2^{\mathrm{polylog}(n)})$~\cite{DinitzKR16}. (For directed graphs,\footnote{In the literature survey that follows we restrict  attention to undirected graphs, which is the focus of this work.} the current best approximation ratios are even worse~\cite{BBM+11,BGJ+12,BRR10,DK11}.) Some of the aforementioned results extend to minimizing the weight of the spanner and produce the same (or sometimes worse) approximation ratio~\cite{BBM+11,GKL23}. Efforts have also been made to experimentally test LP-based algorithms to find the (exact) minimum weight spanner of general graphs~\cite{SigurdZ04, AhmedHJKSS19,BoklerCJW24}.
 
For general  graphs, perhaps the most interesting stretch regime is $t\geq 3$ since, in this regime, $t$-spanners have a subquadratic size for any input graph. Even in this regime, as mentioned above, poly-logarithmic approximation algorithms are impossible under a standard complexity assumption. In various real-life applications, even stretch $t=3$ is too large. It is thus only natural to focus on structural classes of graphs, where we can hope to achieve (i) stretch $t=1+\eps$ 
for any given $\eps>0$ 
and (ii) a constant-factor approximation (with a constant that does not depend on $\eps$). Perhaps the two most basic and well-studied structural classes of graphs are \EMPH{low-dimensional Euclidean spaces} and \EMPH{planar graphs}.

In low-dimensional spaces, we are given a set $P$ of $n$ points in the Euclidean space $\real^d$ of any constant dimension $d$; the metric graph $G_P= (P,E,\wts)$ induced by $P$ is the edge-weighted complete graph with vertex set $P$ and edge weights $\wts(u,v) = \|u-v\|_2$ for all $u,v\in P$. Euclidean spanners have been thoroughly investigated; the book by Smid and Narasimhan~\cite{NS07} covers dozens of techniques for constructing Euclidean $(1+\eps)$-spanners, and a plethora of additional techniques have been devised since then~\cite{Smid25}. In particular, one can construct $(1+\eps)$-spanners with $O(\eps^{1-d}n)$ size~\cite{RS91,le2022truly} and $O(\eps^{-d}\wts(\lmst(G_P)))$ weight~\cite{DNS95,le2022truly}, where $\lmst(G_P)$ is the weight of the Euclidean minimum spanning tree for $P$.  These results imply an $O(\eps^{1-d})$-approximation for the minimum size and $O(\eps^{-d})$ for the minimum weight---of the optimal $(1+\eps)$-spanner.
Finding a minimum weight $(1+\epsilon)$-spanner is still NP-hard in this setting~\cite{CarmiC13}.
A major open problem is to obtain an $O(1)$-approximation (w.r.t.\ the size or weight) for the optimal Euclidean $(1+\eps)$-spanner, where the \EMPH{approximation factor does not depend on $\eps$}. Despite a vast literature on Euclidean spanners, there has been no progress on this problem until the recent work of Le et al.~\cite{le2024towards}, which gives the first \EMPH{bi-criteria approximation algorithm}. Specifically, they constructed $(1+\eps\cdot 2^{O(\log^*(d/\eps))})$-spanners whose size and weight are within $O(1)$ of those of the optimal $(1+\eps)$-spanners for any $\epsilon>0$.  

For planar graphs, the literature on spanners is much sparser. Basic geometric techniques and concepts do not apply and cannot even be adapted to planar graphs. 
One exception is with the aforementioned {\em greedy} algorithm by Alth\"{o}fer et al.~\cite{althofer1993sparse}, which applies to any graph, as it is oblivious to the input structure. They showed that the greedy $(1+\eps)$-spanner of any edge-weighted planar graph $G$ has weight at most $(1+\frac{2}{\eps})\wts(\lmst(G))$. (Since planar graphs have only $O(n)$ edges, the focus has been mainly on optimizing the weight.) Klein~\cite{Klein05} gave a more relaxed variant of the greedy algorithm, with the same asymptotic weight bound, which can be implemented in $O(n)$ time. Both the greedy algorithm and Klein's variant imply an $O(1/\eps)$-approximation for the minimum weight $(1+\eps)$-spanner for any input planar graph. The following fundamental question remains open.
\begin{question}\label{ques:main} 
Can one get a polynomial-time algorithm that computes a $(1+\eps)$-spanner of weight $O(1)\cdot\wts(G_{\opt,\eps})$, where $G_{\opt,\eps}$ is a minimum-weight $(1+\eps)$-spanner of any input planar graph $G$?
\end{question}

A wide variety of techniques have been designed for approximation algorithms in planar graphs over the years~\cite{Baker94,Klein05,BKM09,BHM11,BDHM16,BCEHKM11,EKS19,FL22}, but none of them appears to be applicable to \Cref{ques:main}. Compared to many network design problems, such as {\em TSP} or {\em Steiner tree}, the key challenge in approximating the minimum weight $(1+\eps)$-spanner is that we have to achieve two guarantees: 
 (i) stretch of $1+\eps$ (i.e., preserving all pairwise distances up to a factor of $1+\eps$) and (ii) minimizing the weight (approximating the minimum weight up to a given factor). Baker's technique~\cite{Baker94} is perhaps the most basic one: compute a BFS tree of $G$, divide the input graph into subgraphs consisting of $O(1)$ layers of the BFS tree, each subgraph then has treewidth $O(1)$, and solve the problem in each subgraph, say $H$, by applying dynamic programming on bounded treewidth graphs. If one applies Baker's technique 
to spanners, the issue is that the distance between vertices in each subgraph $H$ (of bounded treewidth) is not the same as the distance in the input graph, and hence the optimal solution for $H$ could be much heavier than $E(G_{\opt,\eps})\cap E(H)$. (When the graph is unweighted and the stretch is constant, one can extend $H$ by $O(t)$ more layers to preserve the distances between vertices in $H$~\cite{DraganFG11b}; when the graph is weighted, such a trick does not apply.) Other variants of Baker's technique, such as the contraction decomposition~\cite{Klein05}, run into a similar issue. 

The two different guarantees in approximating spanners are reminiscent of other problems with distance constraints in planar graphs. One such problem is the $\rho$-dominating set problem: given an input (non-constant) parameter $\rho$, find a minimum set of vertices such that other vertices must be within distance at most $\rho$ from the set. For these problems, there is an effective technique for designing bicriteria approximation algorithms by embedding the graph into a small treewidth graph~\cite{EKS19,FL22,CCLMST23}. Specifically, given a planar graph $G$ with diameter $\Delta$, one can embed $G$ into a graph $H$ with treewdith $O(1/\eps^4)$, such that $\dist_G(u,v)\leq \dist_H(u,v) \leq \dist_G(u,v) + \eps\Delta$ for any $\eps \in (0,1)$~\cite{CCLMST23}. However, the additive distortion $+\eps \Delta$ could be very large compared to $\dist_G(u,v)$; for example, in unweighted graph we have $\dist_G(u,v)=1$ for every edge $(u,v)$, while the additive distortion guarantee is $\eps\Delta$, which could be as large as $\Omega(\eps n)$. Therefore, such an embedding technique also does not provide a good bicriteria approximation algorithm for $(1+\eps)$-spanners. 

Another natural attempt to obtain a bicriteria approximation algorithm is to apply the greedy algorithm with stretch $(1+x\cdot \eps)$, for some parameter $x$, and compare its weight to the  minimum weight of a $(1+\eps)$-spanner. In \Cref{sec:greedy}, we present a hard planar graph instance for which the greedy $(1+x\cdot \eps)$-spanner has weight as large as $\Omega\left(\frac{1}{x^2\eps}\right)\wts(G_{\opt, \epsilon})$; this instance is obtained by a careful adaptation of a hard Euclidean instance from \cite{le2024towards}.
Therefore, the greedy algorithm does not yield an $O(1)$-approximation for the weight even when $x$ is rather large, say $x \approx \eps^{-1/3}$. 

\subsection{Our Main Contribution: Approximate Light Spanners}

In this paper, we make significant progress toward the resolution of \Cref{ques:main} by designing a bicriteria approximation algorithm for $(1+\eps)$-spanners in planar graphs: The stretch is $1+\eps\cdot 2^{O(\log^*(1/\eps))}$ while the approximation ratio is $O(1)$; here $\log^*$ is iterated logarithm.

\begin{theorem}\label{thm:main}
	Given any edge-weighted planar graph $G = (V, E, \wts)$ with integral edge weights $\wts: E\rightarrow \mathbb{N}_+$ as well as a parameter $\epsilon > 0$, a $\poly(n, 1/\eps)$-time algorithm can construct a $\brac{1+\epsilon\cdot 2^{O\brac{\log^* 1/\epsilon}}}$-spanner $H\subseteq G$ of total weight $O(1)\cdot \wts(G_{\opt, \epsilon})$, where $G_{\opt, \eps} = (V, E_{\opt, \epsilon})$ denotes a $(1+\epsilon)$-spanner of $G$ of minimum total weight.
\end{theorem}

To obtain our result, we develop a new technique, which we refer to as \EMPH{iterative planar pruning}.
While our technique
is inspired by the recent biciteria approximation algorithm for Euclidean $(1+\eps)$-spanner~\cite{le2024towards}, 
it has to deviate significantly from it,
since the algorithm of \cite{le2024towards} crucially relies on several basic properties of Euclidean geometry that do not hold in planar graphs.
In \Cref{sec:tech-overview} we provide a detailed overview of our technique, and its comparison with previous work and the technique of \cite{le2024towards} in particular. 
At a very high-level, our technique will heavily exploit the \EMPH{planarity} of the input graph to establish a certain \EMPH{laminar structural property} of carefully chosen paths 
in the current spanner. This laminar structure property is key for obtaining an efficient \EMPH{dynamic programming} algorithm, which is, in turn, 
used for iteratively finding a light path to replace a set of much heavier edges. 

Perhaps our most significant departure from existing techniques for network design problems in planar graphs, including TSP~\cite{Klein05}, Steiner tree~\cite{BKM09}, Steiner forest~\cite{BHM11}, and their prize-collecting counterparts~\cite{BCEHKM11}, to name a few, is the usage of dynamic programming. Existing techniques could be seen as providing a reduction to bounded treewidth graphs, where dynamic programming naturally arises\footnote{These reductions basically apply Baker's technique to the dual planar graph. As we pointed out above, this technique can distort the distances significantly, and hence such a reduction does not seem applicable to our problem.}. 
In contrast to previous work, the dynamic programming in our work \EMPH{applies directly to the input planar graph}, whose treewidth may be as large as $\Theta(\sqrt{n})$.  Our technique could potentially be applicable to other network design problems where, in addition to minimizing the total weight, one seeks to impose distance constraints between nodes in the network. One class of such problems lies in the context of {\em hop-constrained network design}: find a network of minimum weight such that the hop-distance between two nodes is at most a given input parameter $h$. For several hop-constraint problems, strong inapproximability results~\cite{DinitzKR16} rule out polylogarithmic \EMPH{single-criteria} approximation algorithms. There is a long line of work on bicriteria-approximation algorithms for hop-constrained network design problems in general graphs, where one wishes to approximate both the weight of the network and the hop constraint. The state-of-the-art algorithms for general graphs achieve polylogarithmic approximations to both the weight and the hop-constraint; see \cite{HHZ21} and references therein. 
Remarkably, none of the hop-constrained network design problems are known to admit a bicriteria constant approximation
in planar graphs. We anticipate that exploring the applicability of our technique to these problems in planar graphs is a promising research avenue. 

\subsection{Hardness of Minimum Spanners}

For \EMPH{unweighted} planar graphs, the minimum $t$-spanner problem is fully understood. The regime of stretch $t < 2$ is equivalent to $t = 1$, for which the only 1-spanner is the entire graph. 
Dragan, Fomin and Golovach~\cite{DraganFG11b} developed an efficient polynomial-time approximation scheme (EPTAS) using Baker’s technique, constructing a $t$-spanner with at most $(1+\delta)|E(G_{\opt,t})|$ edges for any $t \geq 2$ and $\delta > 0$. For exact algorithms, Brandes and Handke~\cite{brandes1998np} proved NP-hardness for the regime of stretch $t \geq 5$, which was later improved by Kobayashi~\cite{Kobayashi18a} to the entire nontrivial regime of stretch $t \geq 2$. Kobayashi~\cite{Kobayashi18a} also showed that the problem is NP-hard even for planar graphs of maximum degree at most $\Delta= 9$ and $t=2$ (however, it can be solved in polynomial time for $\Delta\leq 4$ and $t=2$~\cite{CK94}); G\'{o}mez, Miyazawa, a Wakabayashi~\cite{GomezMW23} settled the dichotomy between NP-hardness and polynomial-time algorithms for almost all pairs $(\Delta,t)$ of maximum degree and stretch in unweighted planar graphs. 

In \EMPH{weighted} planar graphs, the goal is to find a \EMPH{minimum weight} $t$-spanner (generalizing the measure of size in the unweighted setting). The problem remains NP-hard in this setting~\cite{brandes1998np,Kobayashi18a}, and Dragan et al.~\cite{DraganFG11b} extended their result to planar graphs with positive integer weights bounded by a constant $W$, yielding an EPTAS with a running time of $\frac{1}{\delta} \cdot (t\cdot W)^{O((t\cdot W/\delta)^2)} \cdot n$, for any $t \geq 2$ and $\delta > 0$. 
Remarkably,  no hardness result has been established in the stretch regime $1<t<2$, which is arguably the most interesting. 
In \Cref{sec:hardness}, we prove that the problem of computing an exact minimum-weight $(1+\epsilon)$-spanner in planar graphs is NP-hard, thus closing the only gap left open by previous work. 

\begin{theorem}\label{thm:lower}
    For every $\epsilon>0$, the problem of computing a minimum-weight $(1+\epsilon)$-spanner 
    for a given edge-weighted planar graph $G = (V, E, \wts)$ with polynomially bounded integral edge weights $\wts: E\rightarrow \mathbb{N}_+$, is NP-hard.
\end{theorem}

\section{Technical Overview}\label{sec:tech-overview}

In this section, we highlight the key ideas behind the proof of \Cref{thm:main}.
\subsection{Hard instances} 

As a reminder of the greedy spanner algorithm, it goes over all edges in the input graph in the non-decreasing order in terms of edge weights, and in each iteration it adds the edge to the spanner if the distance stretch between the two endpoints is larger than $1+\eps$. 

We will first demonstrate that the \EMPH{greedy} $(1+\eps)$-spanner algorithm incurs a weight approximation of $\Theta(1/\eps)$ for some hard instances. With a concrete hard instance at hand, we then demonstrate how a simple modification to the greedy spanner can yield a much better spanner construction for this particular instance. In \Cref{sec:prune} we highlight some key insights that are needed for generalizing this simple modification to obtain a \EMPH{pruning procedure} for general planar instances.

Consider a planar graph $G' = (V', E', \wts')$ on $2n+2$ vertices $V'=\{u_0, u_1, \ldots, u_n\}\cup \{v_0, v_1, \ldots, v_n\}$ that includes an edge $(u_i, v_i)$ of weight $1$ for every $0\leq i\leq n$, 
and two edges $(u_0, u_j), (v_0, v_j)$ of weight $\eps/2$ each, for every $1\leq j\leq n$.  See \Cref{fig:ladder} for an illustration.
The minimum weight $(1+\eps)$-spanner of $G'$ consists of edges $\{(u_0, v_0)\}\cup\{(u_0, u_j), (v_0, v_j) : 1\leq j\leq n\}$ and has total weight $1+n\epsilon$. However, the greedy algorithm for stretch $1+\eps$, and in fact for any stretch less than $1+2\eps$, could possibly begin with a minimum spanning tree consisting of edges $\{(u_1, v_1)\}\cup \{(u_0, u_j), (v_0, v_j) : 1\leq j\leq n\}$, and then greedily add all edges $(u_i, v_i), 2\leq i\leq n$, incurring a total weight of $(1+\eps)n$. Thus, the approximation ratio of the greedy algorithm is $(n+\eps n)/(1+\eps n) \approx 1/\eps$ when $n$ grows, which is asymptotically as large as the naive existential bound. 
In \Cref{sec:greedy}, we present a stronger hard instance, for which the weight of the greedy $(1+x\eps)$-spanner, for any $1\leq x = O(\eps^{-1/2})$, exceeds the minimum weight of any $(1+\eps)$-spanner  by a factor of $\Omega\brac{\frac{\epsilon^{-1}}{x^2}}$. In other words, the greedy algorithm performs poorly even as a \EMPH{bicriteria approximation} algorithm.

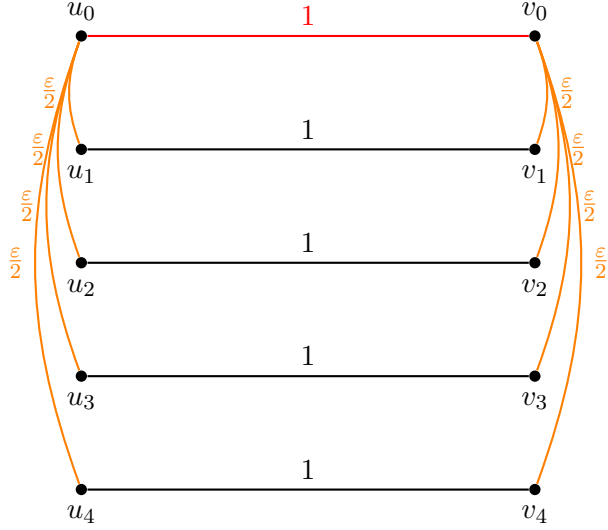
\begin{figure}
    \centering
    \begin{tikzpicture}[scale=1.5]
  \def\n{4} % Define n as 3 for example
  % Draw u nodes

  \node[circle, fill=black, inner sep=1.5pt, label=above:$u_0$] (u0) at (0, 0) {};
  \node[circle, fill=black, inner sep=1.5pt, label=above:$v_0$] (v0) at (4, 0) {};
  
  \foreach \i in {1,...,\n} {
    \node[circle, fill=black, inner sep=1.5pt, label=below:$u_\i$] (u\i) at (0, -\i) {};
  }
  % Draw v nodes
  \foreach \i in {1,...,\n} {
    \node[circle, fill=black, inner sep=1.5pt, label=below:$v_\i$] (v\i) at (4, -\i) {};
  }
  % Draw vertical edges with weight 1
  \draw[thick, red] (u0) -- node[above] {1} (v0);
  \foreach \i in {1,...,\n} {
    \draw[thick] (u\i) -- node[above] {1} (v\i);
  }
  % Draw edges from u0 to uj with weight ε/2 (bent upwards)
  \foreach \j in {1,...,\n} {
    \draw[thick, orange] (u0) to [bend right=20] node[left] {$\frac{\varepsilon}{2}$} (u\j);
  }
  % Draw edges from v0 to vj with weight ε/2 (bent downwards)
  \foreach \j in {1,...,\n} {
    \draw[thick, orange] (v0) to [bend left=20] node[right] {$\frac{\varepsilon}{2}$} (v\j);
  }
\end{tikzpicture}
    \caption{In this example, the optimal $(1+\eps)$-spanner contains all the red and orange edges, but a greedy $(1+\eps)$-spanner might contain all edges in the graph except for the red edge $(u_0, v_0)$.}
    \label{fig:ladder}
\end{figure}

Observe that the main reason that the greedy algorithm returns a heavy spanner is that it misses the {\em critical} edge $(u_0, v_0)$, which, 
together with the $2n$ edges in $\{(u_0, u_j), (v_0, v_j): 1\leq j\leq n\}$, ``serves'' (i.e., provides $(1+\eps)$-spanner paths between) all pairs $u_i,v_j$ of vertices in $G$. In general, the greedy $(1+\eps)$-spanner might miss some {\em critical paths} (instead of direct edges), each of which serves many pairs of vertices simultaneously in the unknown optimal $(1+\eps)$-spanner, while the greedy $(1+\eps)$-spanner has to connect each of these pairs of vertices separately, incurring a much higher total weight. Consequently, to outperform the $O(1/\eps)$-approximation of the greedy $(1+\eps)$-spanner algorithm, our high-level strategy is to identify critical paths that can replace as many existing edges in the greedy spanner as possible. For example, in the instance shown in \Cref{fig:ladder}, edge $(u_0, v_0)$ would be a critical edge, and adding $(u_0, v_0)$ to the greedy spanner allows us to remove the edges $(u_i, v_i)$, $1\leq i\leq n$, thereby achieving a much lower, and in fact the optimal, weight.

\subsection{A pruning framework in planar graphs} \label{sec:prune}

To reduce the approximation ratio to $O(1)$ for general instances, we develop a certain \EMPH{pruning procedure}, which is inspired by the pruning framework in \cite{le2024towards}, developed for Euclidean low-dimensional spaces, and is also reminiscent of the standard local search approach in the broader context of approximation algorithms. 

At a high level, following~\cite{le2024towards}, we start with a (greedy) $(1+\eps)$-spanner  $H$ for the input graph $G = (V,E,\wts)$ that has a (high) approximation ratio $O(1/\eps)$ (in the Euclidean setting the initial approximation is $O(1/\eps^2)$), find two sets of edges $F^{\nw}\subseteq E$ and a $F^\old\subseteq E(H)$, and exchange them: $H_1\leftarrow F^\nw\cup (H\setminus F^\old)$. The key guarantees of our construction are: (i) $H_1$ has stretch $(1+O(\eps))$ and (ii) $\wts(H_1)= O(\log1/\eps)\cdot \wts(G_{\opt,\eps})$. Thus, we pruned a set of heavy edges $F^\old$ to reduce the approximation ratio for the weight of the spanner $H$ {\em exponentially} at the expense of a slight increase in the stretch. We can then apply the same pruning procedure to $H_1$. By repeating the pruning procedure $\log^*{(1/\eps)}$ times, we obtain a spanner $H^*$ with stretch $(1+\eps\cdot  2^{O(\log^*{(1/\eps)})})$ and  weight
\begin{equation*}
    \wts(H^*) = O(1)\underbrace{\log \ldots \log (1/\eps)}_{\text{iterate $\log^*(1/\eps)$ times}} \wts(G_{\opt,\eps}) = O(1) \cdot \wts(G_{\opt,\eps}),
\end{equation*}
as claimed in \Cref{thm:main}. 

Our main technical contribution lies in efficiently finding the heavy(-weight) pruned set $F^\old$ and the light replacement set $F^{\nw}$ in each iteration of the pruning procedure. We note that the \EMPH{existence} of such sets is immediate: simply take $F^{\old} = E(H)$ (the edges of the current spanner) while $F^{\nw} = E(G_{\opt,\eps})$ (the edges of the optimal solution). However, 
$G_{\opt,\eps}$ is not known, 
and the \EMPH{algorithmic task of efficiently (in poly-time) computing} such sets
$F^\old$ and  $F^{\nw}$---on which our approximation algorithm crucially relies---is highly nontrivial, as discussed next. 

In Euclidean spaces, Le et al.~\cite{le2024towards} rely on Euclidean geometry in a crucial way to find such sets  $F^{\old}$ and $F^{\nw}$. For each edge $(u,v)$, let $\mathcal{E}_{uv}$ be the ellipsoid of width $O(\sqrt{\eps})$ that has $u$ and $v$ as the foci. A basic observation is that $\mathcal{E}_{uv}$ contains all points on any $(1+\eps)$-approximate path between $u$ and $v$. One of the key ideas  in~\cite{le2024towards} is the following. Consider any two edges 
$(s,t)$ and $(s',t')$
of almost the same length, say
$\Theta(\ell)$ for some $\ell > 0$, and suppose there are two points $z,w$ such that:  (i) both $z$ and $w$ lie in the intersection of the two ellipsoids $\mathcal{E}_{st}$ and    $\mathcal{E}_{s't'}$, and (ii) $\|z-w\|_2 = \Theta(\ell)/c$ for a sufficiently large constant $c$. Item (i) guarantees that, by taking $(z,w)$ to $F^{\nw}$ (to be added to the spanner) and both edges $\{(s,t),(s',t')\}$ to $F^{\old}$ (to be removed from the spanner), the stretch of the new spanner may grow only  slightly; item (ii) guarantees that  $F^{\old}$ is heavy while  $F^{\nw}$ is light. Edge $(z,w)$ is called a {\em helper} edge in~\cite{le2024towards}. Clearly, Euclidean geometry is central to the approach in \cite{le2024towards}. Furthermore, \cite{le2024towards} also used the  fact that in 
Euclidean spanners, the graph $G$ is the complete graph on $n$ input points in $\mathbb{R}^d$, and so any two vertices are connected by an edge, thus the helper edge can be added to $F^{\nw}$; in planar graphs, however, most pairs of vertices are non-adjacent. 

\paragraph{A laminar structural property.} In light of the above discussion, our primary objective is to efficiently compute critical paths in $G$ that could replace a set of edges with large total weight in any given spanner $H\subseteq G$ such that $\wts(H)\gg \wts(G_{\opt,\eps})$ (possibly slightly increasing the stretch of $H$). Before we can proceed to the {\em algorithmic task of efficiently computing} such critical paths, 
we must first prove their {\em existence}, which by itself is nontrivial.
In particular, removable edges could form far more intricate structures than the basic ladder-like graph shown in \Cref{fig:ladder}; for example, there could be multiple ladders hanging on a critical path (see \Cref{fig:manyladders} for an illustration).
The algorithmic task of  computing such paths poses further technical challenges; it turns out that   these challenges can be overcome by carefully employing dynamic programming, as discussed below. 
    
        \begin{figure}[htb]
        \centering
        \begin{tikzpicture}[scale=0.75]
  \def\n{4} % Define n as 3 for example
  % Draw u nodes

  \foreach \j in {0,...,2}{
    \foreach \i in {0,...,\n} {
     \node[circle, fill=black, inner sep=1.5pt] (u\j\i) at (0+6*\j, -\i) {};
     \node[circle, fill=black, inner sep=1.5pt] (v\j\i) at (3+6*\j, -\i) {};
    }

    \draw[thick, red] (u\j0) -- (v\j0);
    \foreach \i in {1,...,\n} {
      \draw[thick, cyan] (u\j\i) -- (v\j\i);
    }
  
    \foreach \i in {1,...,\n} {
      \draw[thick] (u\j0) to [bend right=20] node[left] {} (u\j\i);
      \draw[thick] (v\j0) to [bend left=20] node[right] {} (v\j\i);
    }
  }

  \node[circle, fill=black, inner sep=1.5pt, label=below:{$s$}] (s) at (-3, 0) {};
  \node[circle, fill=black, inner sep=1.5pt, label=below:{$t$}] (t) at (18, 0) {};
  \draw[thick, red] (s) -- (u00);
  \draw[thick, red] (v00) -- (u10);
  \draw[thick, red] (v10) -- (u20);
  \draw[thick, red] (v20) -- (t);
  
\end{tikzpicture}
        \caption{There could be multiple ladder structures that are hanging on a single critical path, colored red. The blue edges forming the ladders are removable edges of the current spanner $H$.}
        \label{fig:manyladders}
    \end{figure}
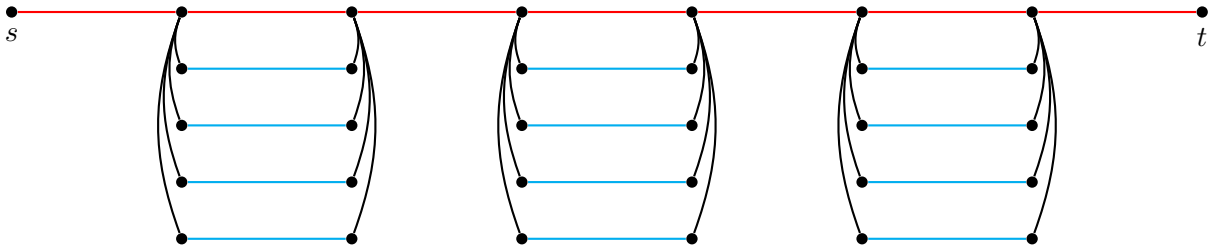

To be more precise, we shall consider critical paths with somewhat relaxed constraints. 
Specifically, the distances from removable edges to the critical path will not necessarily be as small as an $\eps$-fraction of the lengths of the removable edges, as in the hard instance shown in \Cref{fig:ladder}. For general instances, we can only guarantee a much weaker upper bound of a $(1 - \Omega(1))$-fraction instead of an $\eps$-fraction; this is the main reason why our pruning procedure increases the stretch. 
With the above relaxation, we can ultimately prove the existence of a path $\rho^*$ and an edge set $F^* = F^*(\rho^*) \subseteq E(H)$, such that $\frac{\wts(F^*)}{\wts(\rho^*)}\approx \frac{\wts(H)}{\wts(G_{\opt,\eps})}$, 
while the stretch of $H\cup E(\rho^*)\setminus F^*$ only slightly exceeds that of $H$. Hence, adding $\rho^*$ to $H$ and pruning $F^*$ from $H$ produces a significantly lighter spanner with almost the same stretch guarantee. We refer to such a pair, $\rho^*$ and $F^*$, as a \EMPH{pruning pair}.

We note that proving the existence of a pruning pair using a naive averaging argument is doomed.  
Specifically, let us associate every edge $(s, t)\in E(H)$ with a shortest path $\gamma_{s, t}$ in $G_{\opt,\eps}$ and map each such path $\gamma_{s, t}$ to 
a disjoint edge set $F_{s,t}\subseteq E(H)$, where the union of all paths $\gamma_{s, t}$ (respectively, edge sets $F_{s,t}$) is $G_{\opt,\eps}$ (resp., $H$). 
Consequently,  an average path $\gamma_{s, t}$ in $G_{\opt,\eps}$ is mapped to a set $F_{s,t}\subseteq E(H)$ of total weight $\wts(\gamma_{s, t})\cdot \frac{\wts(H)}{\wts(G_{\opt,\eps})}$; note that $\frac{\wts(F_{s,t})}{\wts(\gamma_{s,t})}\approx \frac{\wts(H)}{\wts(G_{\opt,\eps})}$.
It is now tempting to argue that $\gamma_{s, t}$ could replace this heavy edge set $F_{s,t} \subseteq E(H)$ (which would reduce the weight by the required amount) without blowing up the stretch. Alas,
the main issue with such an averaging argument is that it completely ignores cases where $\gamma_{s, t}$ intersects many different paths $\gamma_{u, v}$, but any single path $\gamma_{u, v}$ shares only a small proportion with $\gamma_{s, t}$, and so the path $\gamma_{s, t}$ itself could not replace any other edge $(u, v)\in E(H)$ without significantly blowing up the stretch.

In our proof, to establish the existence of a pruning pair $\rho^*, F^*$, where the single path $\rho^*$ in $G_{\opt,\eps}$ can remove the much heavier set of edges $F^*$ in $H$ without blowing up the stretch, we need to drill much deeper, by \EMPH{leveraging planarity} and exploring certain \EMPH{laminar structures of all paths $\{\gamma_{s, t} : (s, t)\in E(H)\}$}.
First, we introduce the following key definition of \EMPH{$\kappa$-hanging}, which formalizes the notion of removable edges (a canonical setting of  $\kappa$ is $\kappa = 2/3$):

\begin{definition}[$\kappa$-hang]\label{def:hang-intro}
	Consider any (not necessarily simple) path $\rho = \langle v_1, v_2, \ldots, v_k\rangle$  in $G$ and any edge $(a, b)\in E(H)$. We say that edge $(a, b)$ is \EMPH{$\kappa$-hanging} on the path $\rho$ if there exists a pair of vertices $v_i, v_j$ with $i<j$ such that:
	\begin{enumerate}[(1)]
		\item $\wts\brac{\rho[v_i, v_j]}\geq \kappa\cdot\wts(a, b)$, where $\rho[v_i,v_j]$ denotes the subpath of $\rho$ between $v_i$ and $v_j$; 
        
        intuitively, the lower bound on $\wts\brac{\rho[v_i, v_j]}$ will be helpful when we add the path $\rho[v_i, v_j]$ and remove the edge $(a, b)$ in spanner $H$ while preserving the stretch;
		\item $\dist_G(a, v_i) + \wts(\rho[v_i, v_j]) + \dist_G(v_j, b)\leq (1+\epsilon)\cdot \wts(a, b)$.
		\end{enumerate}
    We say that $(a, b)$ is \EMPH{$\kappa$-hanging} 
    at $(v_i, v_j)$ on $\rho$.
\end{definition}
  \begin{figure}[!htb]
        \centering
        \begin{tikzpicture}[scale=2]
  % Main curve ρ (rho)

  \draw[ultra thick] plot[smooth] coordinates {(0,0) (0.4, 0.7) (1,1)};

  \draw[ultra thick, color=orange] plot[smooth] coordinates {(1,1) (2, 1.1) (3,1)};

  \draw[ultra thick] plot[smooth] coordinates {(3,1) (3.6, 0.7) (4, 0)};

  \node at (2, 1.3) {$\rho[v_i, v_j]$};
  \node at (2, 0.6) {$\boldsymbol{\rho}$};
  
  % Points v_i and v_j on ρ
  \node[circle, fill=black, inner sep=1.5pt, label=below:$v_i$] (vi) at (1,1) {};
  \node[circle, fill=black, inner sep=1.5pt, label=below:$v_j$] (vj) at (3,1) {};
  
  % Points a and b
  \node[circle, fill=black, inner sep=1.5pt, label=above left:$a$] (a) at (0,2) {};
  \node[circle, fill=black, inner sep=1.5pt, label=above right:$b$] (b) at (4,2) {};
  
  % Straight edge (a,b)
  \draw[thick] (a) -- (b);
  
  % Curves connecting a to v_i and b to v_j
  \draw[thick, color=orange] (a) to[out=270, in=135] (vi);
  \draw[thick, color=orange] (b) to[out=270, in=45] (vj);

  %\draw [gray!50] plot [smooth cycle] coordinates {(-0.1,-0.2) (0.4, 0.5) (1,0.8) (2, 0.9) (3,0.8) (3.6, 0.5) (4.1, -0.2) (4.1, 0.2) (3.6, 0.9) (3,1.2) (2, 1.3) (1, 1.2) (0.4, 0.9) (-0.1, 0.2) };
\end{tikzpicture}
        \caption{An example for an edge $(a, b)$ that is $\kappa$-hanging at $(v_i, v_j)$ on path $\rho$. The orange path between $a$ and $b$ has length at most $(1+\eps)\cdot\wts(a, b)$; the orange sub-path of $\rho$ between $v_i$ and $v_j$, which is also a sub-path of the orange path between $a$ and $b$, has length at least $\kappa\cdot\wts(a, b)$.}
        \label{fig:hang-intro}
    \end{figure}
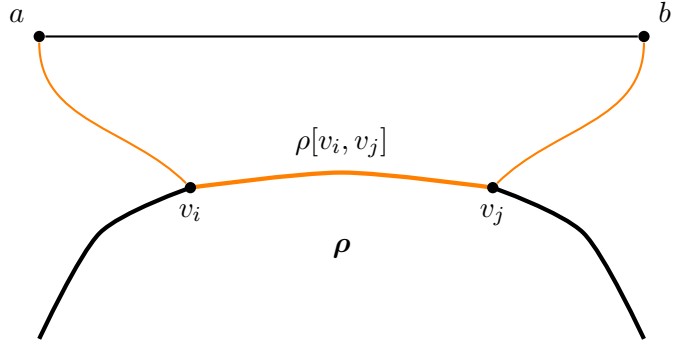

See 
\Cref{fig:hang-intro}
for an illustration. We argue that, since $\kappa = 2/3$, and due to items (1) and (2) in \Cref{def:hang-intro}, the spanner obtained from the current $(1+\eps)$-spanner,
by adding $\rho$ to
it and removing all the $\kappa$-hanging edges from it, has a stretch of at most $1 + O(\eps)$. The reason why the additive term $\eps$ has to increase by some constant factor to $O(\eps)$ is that our current spanner may not include the optimal sub-paths in $G$ from $a$ to $v_i$ or from $v_j$ to $b$. In fact, our current spanner may not even include $(1+\eps)$-spanner paths between these points, as we are iteratively pruning paths, hence our argument has to bypass several hurdles in order to prove the stretch bound. 

Proving the stretch bound alone is insufficient---we also need to guarantee that the weight reduces significantly, and thus it is crucial that the $\kappa$-hanging edges are heavy with respect to the path on which they are hanging.  
Furthermore, while proving the existence of a path $\rho^*$ with heavy $\kappa$-hanging edges is already nontrivial, the existence alone does not guarantee an efficient algorithm. Our key insight is that we can impose a \EMPH{laminar structure} of $\kappa$-hanging paths in a way that can be exploited for an efficient algorithm, formalized by the following structural lemma; when applying the lemma, we may assume that the weight of the current $(1+\eps)$-spanner $H$ exceeds that of the optimal solution $G_{\opt,\eps}$ by a factor of $\alpha$, for a large parameter $\alpha \gg 1$.

\begin{lemma}[structural property, simplified]\label{lm:struct}
Define $\alpha = \frac{w(H)}{\wts(G_{\opt,\eps})}$. There exists a pruning pair that consists of a shortest path $\rho^*$  in $G_{\opt, \epsilon}$ and an edge set $F^*\subseteq E(H)$, with the following  guarantees:
    \begin{enumerate}[(1)]
        \item Each edge $e\in F^*$ is $\frac{2}{3}$-hanging on $\rho^*$;
        \item The total weight $\wts\brac{F^*}$ of $F^*$ satisfies $\wts\brac{F^*}\geq \frac{\alpha}{6}\cdot \wts\brac{\rho^*}$;
        \item   For each edge $e\in F^*$, let $(a_e, b_e)$ be the hanging points of $e$ on $\rho^*$; then all the sub-path intervals $\{\rho^*[a_e, b_e]: e\in F^*\}$ form a \EMPH{laminar family}; recall that a laminar family is a family of sets where any two sets are either disjoint or related by containment.
    \end{enumerate}
\end{lemma}
By the first guarantee, the edges in $F^*$ are $2/3$-hanging; as discussed above, which means that we can remove them from $H$ without incurring much stretch. The second guarantee implies that the set $F^*$ of removable edges is heavy with respect to the weight of the edges in $\rho^*$, and hence removing the edges in $F^*$ from the spanner $H$ and adding the edges in $\rho$ in their place results in a significant reduction to the weight of $H$.
The third guarantee, namely the laminar family structure, is crucial for obtaining an efficient dynamic programming algorithm to compute an approximation of $\rho^*$. Here, we heavily exploit the \EMPH{planarity} of the input graph in establishing the laminar property.

\paragraph{Dynamic programming.} 
 The structural property, which asserts the {\em existence} of pruning pairs as stated in \Cref{lm:struct}, does not immediately lead to an efficient pruning procedure; indeed, in our existential proof, the structure of critical paths will depend on  $G_{\opt,\eps}$, which is of course unknown. So the algorithmic goal is to find an approximate shortest path $\rho$ between some vertex pair $s, t$ such that a large amount of edges in $H$ could be hanging on $\rho$ in some ladder-like manner, similarly to the illustration of \Cref{fig:manyladders}. For this task, we will adopt a dynamic programming approach which is fairly natural: we will maintain two tables (this is only an informal description, and some details in the main algorithm are omitted here)
$$\{\rho[s, t, L]: s, t\in V, L\leq (1+\eps)\,\dist_G(s, t)\}$$
$$\{P[s, t, L] : s, t\in V, L\leq (1+\eps)\,\dist_G(s, t)\},$$
where $\rho[s, t, L]$ will be an approximate shortest path (not necessarily a simple path) between $s$ and $t$, and $P[s, t, L]\subseteq E(H)$ will be a set of edges which can $\frac{1}{3(1+\eps)}$-hang on $\rho[s, t, L]$; note that the hanging parameter now degrades by roughly a factor of 2, from $\frac23$ to $\frac{1}{3(1+\eps)}$, which is needed for technical reasons. The transition rule of the dynamic programming table would be to find the best intermediate vertex $z$ and $1\leq L' < L$ so that the concatenated path $\rho[s, z, L']\circ \rho[z, t, L-L']$ collects the heaviest possible edge set.

The main technical issue here is that the two removable edge sets, $P[s, z, L']$ and $P[z, t, L-L']$, are usually not disjoint, and so $\wts(P[s, z, L']\cup P[z, t, L-L'])$ could be much smaller than $\wts(P[s, z, L']) + \wts(P[z, t, L-L'])$. If we take the intermediate vertex $z$ and $L'$ that maximize the total weight of the union $P[s, z, L']\cup P[z, t, L-L']$, then we could suffer a significant under-estimation of the total weight of edges removable by the best approximate shortest path between $s$ and $t$, especially when $P[s, z, L']$ and $P[z, t, L-L']$ have large intersections. Hence, ultimately, the ratio $\frac{\wts(P[s, t, L])}{\wts(\rho[s, t, L])}$ may be significantly smaller than $\frac{\wts(F^*)}{\wts(\rho^*)}$.

\begin{figure}[!htb]
    \centering
    \begin{tikzpicture}[scale=1.5]
    % Define the style for the nodes

    % Place the nodes
    \node[circle, fill=black, inner sep=1.5pt, label=below:$s$] (s) at (0,0) {};
    \node[circle, fill=black, inner sep=1.5pt, label=below:$u_1$] (c1) at (2,0) {};
    \node[circle, fill=black, inner sep=1.5pt, label=below:$v_1$] (d1) at (2.5,0) {};
    \node[circle, fill=black, inner sep=1.5pt, label=below:$u_2$] (c2) at (4,0) {};
    \node[circle, fill=black, inner sep=1.5pt, label=below:$v_2$] (d2) at (4.5,0) {};
    \node[circle, fill=black, inner sep=1.5pt, label=below:$u_9$] (c3) at (6,0) {};
    \node[circle, fill=black, inner sep=1.5pt, label=below:$v_9$] (d3) at (6.5,0) {};
    \node[circle, fill=black, inner sep=1.5pt, label=below:$t$] (t) at (8.5,0) {};

    \node[circle, fill=black, inner sep=1.5pt, label=above:$x$] (x) at (3.5,1) {};
    \node[circle, fill=black, inner sep=1.5pt, label=above:$y$] (y) at (5,1) {};

    \node[color=orange] at (4.25, 1.3) {$f$};

    % Connect the nodes with lines
    \draw[thick, red] (s) -- (c1) -- (d1) -- (c2) -- (d2);
    
    \draw[thick, red] (c3) -- (d3) -- (t);

    \draw[thick, red, dashed] (d2) -- (c3);

    \draw[thick, color=orange] (x) -- (y);
    
    \draw[thick] (x) to[out=200, in=45] (c1);
    \draw[thick] (y) to[out=220, in=30] (d1);
    \draw[thick] (x) to[out=-70, in=120] (c2);
    \draw[thick] (y) to[out=-110, in=60] (d2);
    \draw[thick] (x) to[out=-40, in=150] (c3);
    \draw[thick] (y) to[out=-20, in=135] (d3);
    
\end{tikzpicture}
    \caption{In this figure, the red path is $\rho[s,t,L]$, and $f = (x, y)$ appears many times in the multi-set $P[s, t, L]$. Then, we can show that $f$ can hang on the path $\rho[s, t, L]$ at multiple positions, say $(u_1, v_1), (u_2, v_2),\ldots ,(u_9, v_9)$. Then, we can replace the sub-path $\rho[u_1, v_9]$ with a shortcut $u_1\rightarrow (x, y)\rightarrow v_9$, which reduces the total weight of $\rho$ by at least $\wts(f)$.}
    \label{fig:overview-shortcut}
\end{figure}
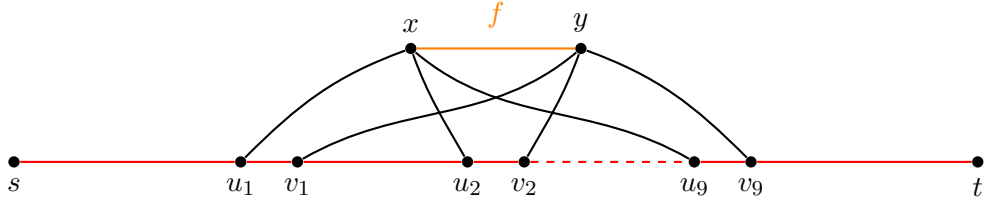

To resolve this issue, we will instead maximize the \EMPH{union with multiplicities}, $P[s, z, L']\uplus P[z, t, L-L']$, whose total weight is $\wts(P[s, z, L']) + \wts(P[z, t, L-L'])$. By using such an {\em over-estimation}, we will be able to \emph{lower bound} $\frac{\wts(P[s, t, L])}{\wts(\rho[s, t, L])}$ with $\frac{\wts(F^*)}{\wts(\rho^*)}$. 
It turns out that we can handle multi-set union and the consequent over-estimation much more effectively than an ordinary set union and the consequent under-estimation. 
The issue with an over-estimation is that when we replace the edge set $P[s, t, L]$ with the path $\rho[s, t, L]$, the total weight of $H$ usually does not decrease by the amount of $\wts(P[s, t, L])$, since $P[s, t, L]$ is a multi-set. To lower bound the total weight of edges that we can actually prune, we will upper bound the multiplicities of most edges 
in the multi-set $P[s, t, L]$ by a large constant (say $20$), crucially relying on the fact that $\rho[s, t, L]$ is an approximate shortest path. Roughly speaking, imagine that some edge $f\in P[s, t, L]$ has a high multiplicity, then we can show that $f$ must hang on the path $\rho[s, t, L]$ at many different positions. In this case, there exists a shortcut on $\rho[s, t, L]$ through the edge $f$ that can reduce the total weight 
of $\rho[s, t, L]$ by at least $\wts(f)$. Since we know that $\wts(\rho[s, t, L])$ is at most $(1+\eps)\dist_G(s, t)$ beforehand, we can upper bound the total weight 
improvements due to shortcuts by $\eps\cdot\dist_G(s, t)$, which, in turn, yields the required constant upper bound on the overall 
multiplicities of edges in $P[s, t, L]$. See \Cref{fig:overview-shortcut} for an illustration.

\section{Preliminaries}

For every integer $x\geq 1$, let $\floor{x}_2$ be the largest integer power of $2$ not exceeding $x$ (that is, $\floor{x}_2=2^{\floor{\log_2 x}}$). Let $G= (V, E, \wts)$ be an undirected weighted planar graph, where $\wts: E\rightarrow \{1, 2,\ldots, W\}$ is an integral edge weight function. Let $\epsilon > 0$ be the input stretch parameter. We assume that shortest paths are unique in both $G$ and $G_{\opt, \epsilon}$ by breaking ties lexicographically.

For any $s, t\in V$, let $\dist_{G}(s, t)$ denote the length of the shortest path between $s$ and $t$ in $G$. For any (not necessarily simple) path $\rho = \langle v_1, v_2, \ldots, v_k\rangle$ in $G$ and any pair of indices $1\leq i < j\leq k$, let $P[v_i, v_j]$ be the sub-path of $P$ between $v_i$ and $v_j$. For any two (not necessarily simple) paths $\rho_1$ and $\rho_2$, let $\rho_1\circ \rho_2$ be their concatenation if they share one endpoint. For any subgraph $H$ of $G$, let $\wts(H) = \sum_{e\in E(H)}\wts(e)$ be the total weight of edges in $H$. For any multi-set of edges $F$, let $\wts(F)$ be the total weight of edges in $F$. For two multi-sets $F_1$ and $F_2$, let $F_1\uplus F_2$ denote the multi-set union of $F_1$ and $F_2$ by summing the multiplicities of each element.

\section{Pruning Planar Light Spanners}
\label{sec:pruning}

Throughout this section, we assume that $W < n^2 / \eps$; in the end we will show how to deal with general cases. In order to prove \Cref{thm:main}, our main technical contribution is a pruning algorithm, as summarized in the following statement.
\begin{theorem}\label{prune}
	Take two parameters $\epsilon,\delta>0$ such that $\eps \leq \min\{10^{-2}, \delta\}$. Let $G = (V, E, \wts)$ be an undirected planar graph with positive integral edge weights $\wts: E\rightarrow \mathbb{N}_+$, 
    and let $H$ be a $(1+\delta)$-spanner of $G$ such that $\theta := \wts(H) / \wts(G_{\opt, \epsilon}) \geq \Omega(1)$.   Then one can compute, in polynomial time, two sets of edges, $F^\nw\subseteq E$ and $F^\old\subseteq E(H)$, such that the following holds:
	\begin{enumerate}[(1)]
         \item the stretch of graph $H_1 = F^\nw\cup (H\setminus F^\old)$ is at most $1+O(1)\cdot\delta$, and
         \item $\wts(H_1)\leq O(\log\theta)\cdot \wts(G_{\opt,\eps})$.
	
	\end{enumerate}
\end{theorem}

\noindent Our main theorem follows immediately by a successive application of \Cref{prune}. 
\begin{proof}[Proof of \Cref{thm:main}]
	Starting with $H$ being the greedy light $(1+\epsilon)$-spanner from \cite{althofer1993sparse} of weight $O(1/\epsilon)\cdot\wts(G_\mst)\leq O(1/\epsilon)\cdot\wts(G_{\opt, \epsilon})$, we successively apply \Cref{prune} and update $H\leftarrow H_1$ for $O(\log^* (1/\epsilon))$ iterations. At the end, $H$ is a $\brac{1+\epsilon\cdot 2^{O\brac{\log^* (1/\epsilon)}}}$-spanner of weight $O(1)\cdot \wts(G_{\opt, \epsilon})$.
\end{proof}
The rest of this section is dedicated to the proof of \Cref{prune}. In~\Cref{ssec:alg}, we present a pruning algorithm that finds the edge sets $F^\nw\subseteq E$ and $F^\old\subseteq E(H)$, followed by the weight and stretch analyses establishing properties (1) and (2) in~\Cref{ssec:analysis1}. 

\subsection{Description of the Pruning Algorithm and Runtime Analysis}
\label{ssec:alg}

\begin{definition}\label{hang}
	Consider any (not necessarily simple) path $\rho = \langle v_1, v_2, \ldots, v_k\rangle$ in $G$ and any edge $(a, b)\in E(H)$. We say that edge $(a, b)$ can $\kappa$-\emph{hang} on the path $\rho$ if there exists a pair of vertices $v_i, v_j$ with $i<j$ such that
	\begin{enumerate}[(1)]
		\item $\wts\brac{\rho[v_i, v_j]}\geq \kappa\cdot\wts(a, b)$;
		\item $\dist_G(a, v_i) + \wts(\rho[v_i, v_j]) + \dist_G(v_j, b)\leq (1+\epsilon)\cdot \wts(a, b)$.
	\end{enumerate}
	The vertex pair $(v_i, v_j)$ will be called the $\kappa$\emph{-hanging} points of $(a, b)$ on $\rho$; or equivalently, we say that $(a, b)$ is \emph{$\kappa$-hanging} at $(v_i, v_j)$ on $\rho$.
\end{definition}

Initially, set $F^\nw\leftarrow \emptyset$ and $F^\old \leftarrow \emptyset$. We will repeatedly find an approximate shortest path $\rho$ on which a large number of edges in $P\subseteq E(H)\setminus (F^\nw\cup F^\old)$ can $\frac{1}{3(1+\eps)}$-hang:  
In each round, we prune these $\frac{1}{3(1+\eps)}$-hanging edges in $P$ by setting $F^\old\leftarrow F^\old \cup P$, and then add the path $\rho$ to the spanner $H_1$ by setting
$F^\nw\leftarrow F^\nw \cup E(\rho)$. We expect that adding $E(\rho)$ and removing $P$ would significantly reduce the total weight while approximately preserving distance stretch.

To find such a good approximate shortest path $\rho$ in each round, we will rely on a dynamic programming approach. The dynamic programming procedure maintains three tables, described in the following.
\begin{framed}
	\noindent A dynamic programming table for triples $(s, t, L)$, where $s,t\in V(G)$ and $L\leq (1+\epsilon)\cdot\dist_G(s, t)$ is a positive integer.
	\begin{itemize}[leftmargin=*]
		\item A table of paths $\rho[s, t, L]$, for all $(s, t, L)\in V\times V\times [2nW]$ such that $L\leq (1+\epsilon)\cdot\dist_G(s, t)$., where $\rho[s, t, L]$ is a (not necessarily simple) path between $s$ and $t$ in $G$ such that $\wts\brac{\rho[s, t, L]} = L$. 

		\item A table of edge multi-sets $P[s, t, L]\subseteq E(H)\setminus (F^\nw\cup F^\old)$, for all $(s, t, L)\in V\times V\times [2nW]$ such that $L\leq (1+\epsilon)\cdot\dist_G(s, t)$, and each edge $e\in P[s, t, L]$ is $\frac{1}{3(1+\eps)}$-hanging on the path $\rho[s,t,L]$.
		\item A table of edge weight sums $\DP[s, t, L]\in \mathbb{N}^+$, for all $(s, t, L)\in V\times V\times [2nW]$ such that $L\leq (1+\epsilon)\cdot \dist_G(s, t)$ and $\DP[s, t, L] = \wts\brac{P[s, t, L]}$.
	\end{itemize}
\end{framed}

Intuitively, our goal is to find an approximate shortest path $\rho[s, t, L]$ on which an edge set $P[s, t, L]\subseteq E(H)\setminus (F^\nw\cup F^\old)$ of large total weight can be $\frac{1}{3(1+\eps)}$-hanging. Due to technical reasons, we approximate the charging sets by multi-sets (rather than ordinary sets).

We compute the entries of the three tables as follows.
For every pair of vertices $s, t\in V$, let $\pi_{s, t}$ be the shortest path between $s$ and $t$ in $G$. First, compute the (non-multi) set $B[s, t]$ of all edges $e\in E(H)\setminus (F^\nw\cup F^\old)$ that is $\frac{1}{3(1+\eps)}$-hanging at $(s, t)$ on $\pi_{s, t}$; this computation can be done in $O(n^3)$ time. 

Initialize $\rho[s,t,\dist_G(s, t)]\leftarrow \pi_{s,t}$ (i.e., a shortest $st$-path in $G$), $P[s, t, L] \leftarrow B[s,t]$, and $\DP[s, t, L] \leftarrow \wts\brac{B[s, t]}$ for any $L\leq (1+\eps)\dist_G(s, t)$. Next, go over every length parameter $L = 1, 2, \ldots, nW$. For every pair $(s, t)$ such that $L\leq (1+\epsilon)\cdot \dist_G(s, t)$, we choose a via point $z\in V$ and $0\leq L' < L$ by setting 
$$(z, L') = \arg\max_{(z, L')} \left\{\DP[s, z, L'] + \DP[z, t, L-L'] + \mathbf{1}\left[\max\{L', L-L'\} < \floor{L}_2\right]\cdot \wts\brac{B[s, t]}
\right\},$$
and then assign 
\begin{align*}
\DP[s, t, L] &\leftarrow \DP[s, z, L'] + \DP[z, t, L-L'] + \mathbf{1}\left[\max\{L', L-L'\} < \floor{L}_2\right]\cdot \wts\brac{B[s, t]},\\
P[s, t, L] & \leftarrow 
\begin{cases}
	P[s, z, L'] \uplus P[z, t, L-L']	& \mbox{\rm if }	\max\{L', L-L'\} \geq \floor{L}_2\\
	P[s, z, L'] \uplus P[z, t, L-L'] \uplus B[s, t]	&	\mbox{\rm if }\max\{L', L-L'\} < \floor{L}_2,
\end{cases} \\
\rho[s, t, L] & \leftarrow \rho[s, z, L']\circ \rho[z, t, L-L'].
\end{align*}
Intuitively, to find the best path between $s, t$, we identify the best via point $z'$ along with a length parameter $L'$ and build the path $\rho[s, t, L]$ as the concatenation of paths $\rho[s, z, L']$ and $\rho[z, t, L-L']$, and entries $P[s, t, L], \DP[s, t, L]$ keep track of the  pruned sets and weights associated with path $\rho[s, t, L]$.

After all entries in the dynamic programming table have been computed, which takes $O(n^5W^2) = n^{O(1)}$ time, find the triple $(s^*, t^*, L^*)$ which maximizes the ratio 
$$\beta = \frac{\DP[s^*, t^*, L^*] }{\wts\brac{\rho[s^*, t^*, L^*]}}.$$ 
If $\beta \geq 1$, then update $F^\nw\leftarrow F^\nw\cup E(\rho[s^*, t^*, L^*])$ and $F^\old\leftarrow F^\old\cup P[s^*, t^*, L^*]$, and start a new round (recomputing the dynamic programming tables); otherwise our pruning algorithm terminates and returns $$H_1 = F^\nw\cup (H\setminus F^\old)$$
The whole algorithm is summarized below as \Cref{alg-prune}.

\begin{algorithm}
    \caption{Prune a $(1+\delta)$-spanner $H\subseteq G$}\label{alg-prune}
    initialize $F^\nw, F^\old \leftarrow \emptyset$\;
    \While{true}{
        compute $B[s,t]\subseteq E(H)\setminus (F^\nw\cup F^\old)$ which are edges that is $\frac{1}{3(1+\eps)}$-hanging at $(s, t)$ on $\pi_{s, t}$\;
        initialize $\rho[s, t, \dist_G(s, t)]\leftarrow \pi_{s, t}, P[s, t, L] \leftarrow B[s,t], \DP[s, t, L] \leftarrow \wts\brac{B[s, t]}, \forall L\leq (1+\eps)\dist_G(s, t)$\;
        \For{$L = 1, 2, \ldots, nW$}{
            \For{$(s, t)$ such that $L\leq (1+\epsilon)\cdot \dist_G(s, t)$}{
                choose a vertex $z\in V$ and $0\leq L' < L$ such that:
                $(z, L') \leftarrow \arg\max_{(z, L')} \left\{\DP[s, z, L'] + \DP[z, t, L-L'] + \mathbf{1}\left[\max\{L', L-L'\} < \floor{L}_2\right]\cdot \wts\brac{B[s, t]}
\right\}$\;
                $\DP[s, t, L] \leftarrow \DP[s, z, L'] + \DP[z, t, L-L'] + \mathbf{1}\left[\max\{L', L-L'\} < \floor{L}_2\right]\cdot \wts\brac{B[s, t]}$\;
                $P[s, t, L] \leftarrow 
\begin{cases}
	P[s, z, L'] \uplus P[z, t, L-L']	& \mbox{\rm if }	\max\{L', L-L'\} \geq \floor{L}_2\\
	P[s, z, L'] \uplus P[z, t, L-L'] \uplus B[s, t]	&	\mbox{\rm if }\max\{L', L-L'\} < \floor{L}_2,
\end{cases}$ \;
                $\rho[s, t, L]  \leftarrow \rho[s, z, L']\circ \rho[z, t, L-L']$\;
            }
        }
        find the triple $(s^*, t^*, L^*)$ which maximizes the ratio $\beta = \frac{\DP[s^*, t^*, L^*] }{\wts\brac{\rho[s^*, t^*, L^*]}}$\;
        \If{$\beta\geq 1$}{
            update $F^\nw\leftarrow F^\nw\cup E(\rho[s^*, t^*, L^*])$, $F^\old\leftarrow F^\old\cup P[s^*, t^*, L^*]$\;
        }
        \Else{
            break\;
        }
    }
    \Return $H_1 \leftarrow F^\nw\cup (H\setminus F^\old)$\;
\end{algorithm}

\subsection{Weight Analysis}
\label{ssec:analysis1}

Let us begin with some basic properties of the dynamic programming scheme.

\begin{lemma}\label{dyn-basic}
    Throughout the course of the dynamic programming algorithm, for any triple $(s, t, L)$, $L\leq (1+\eps)\dist_G(s, t)$, the value of $\DP[s, t, L]$ is non-decreasing, and $\DP[s, t, L]\geq \wts(B[s, t])$.
\end{lemma}
\begin{proof}
    At the beginning, we have $\DP[s, t, L] = \wts(B[s, t])$. In each update, if we take $L' = 0$ and $z = s$, then the value of $\DP[s, z, L'] + \DP[z, t, L-L'] + \mathbf{1}\left[\max\{L', L-L'\} < \floor{L}_2\right]\cdot \wts\brac{B[s, t]}$ is exactly the current value of $\DP[s, t, L]$. Since $(z, L')$ is the maximizer, we know that $\DP[s, t, L]$ never decreases.
\end{proof}

\subsubsection{Structural Properties}
Let us begin with a high-level outline before focusing on technical details. When $H$ is much heavier than the unknown optimal spanner $G_{\opt,\eps}$, our goal is to add to $H$ a small set of new edges to help remove a large set of old edges from $H$. The purpose of this subsection, roughly speaking, is to identify a key structural property that there always exists an approximate shortest path $\rho^*$ in $G$ as well as an edge subset $F^*\subseteq E(H)$, such that $\frac{\wts(\rho^*)}{\wts(F^*)}\approx \frac{\wts(G_{\opt,\eps})}{\wts(H)}$, and $E(\rho^*)\cup (H\setminus F^*)$ has stretch $1+O(\delta)$. As we shall see shortly, finding the ideal choice of $\rho^*, F^*$ requires knowledge of the optimal $(1+\eps)$-spanner $G_{\opt,\eps}$, but we will discuss how to find $\rho^*, F^*$ algorithmically by allowing approximations later on.

Fix a plane embedding of $G$ (the embedding is used only in the analysis). For each edge $(a, b)\in E(H)$, let $\gamma_{a, b}$ be the shortest $ab$-path in $G_{\opt, \epsilon}$. Since $G$ is an embedded planar graph, we can define $R_{a, b}\subseteq \mathbb{R}^2$ as the bounded region enclosed by the edge $(a, b)$ and the path $\gamma_{a, b}$.
\begin{lemma}\label{lem:laminar}
	For any two distinct edges $e_1, e_2\in E(H)$, the regions $R_{e_1}$ and $R_{e_2}$ are either interior-disjoint, or one contains the other.
\end{lemma}
\begin{proof}
	The proof is evident because shortest paths and edges do not cross each other in any plane embedding of $G$. See \Cref{region} for an illustration.
    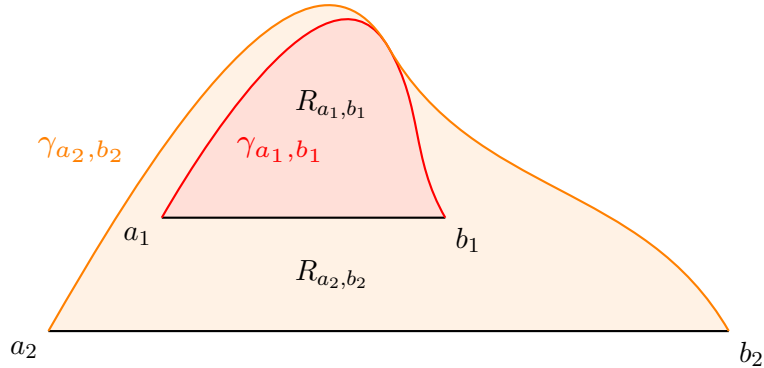
\begin{figure}[!htb]
        \centering
        \begin{tikzpicture}[scale=1.5]

\pgfdeclarelayer{background}
\pgfsetlayers{background,main}

% Define points
\coordinate (a1) at (0, 0);
\coordinate (b1) at (2.5, 0);
\coordinate (a2) at (-1, -1);
\coordinate (b2) at (5, -1);

% Draw the direct edges
\draw[thick] (a1) -- (b1); % Edge (a1, b1)
\draw[thick] (a2) -- (b2); % Edge (a2, b2)

% Draw the convex curves
\draw[thick, red] (a1) to[out=60, in=120] (2, 1.5) to[out=-60, in=120] (b1) node[midway, above, yshift=10pt, xshift=30pt, font=\scriptsize] {$\gamma_{a_1, b_1}$};
\draw[thick, orange] (a2) to[out=60, in=120] (2, 1.5) to[out=-60, in=120] (b2) node[midway, above, yshift=10pt, xshift=-20pt, font=\scriptsize] {$\gamma_{a_2, b_2}$};

% Highlight the regions
\begin{pgfonlayer}{background}
\fill[orange!20, opacity=0.5] (a2) to[out=60, in=120] (2, 1.5) to[out=-60, in=120] (b2) -- cycle;
\fill[red!20, opacity=0.5] (a1) to[out=60, in=120] (2, 1.5) to[out=-60, in=120] (b1) -- cycle;
\end{pgfonlayer}

% Label the regions
\node at (1.5, 1) {$R_{a_1, b_1}$};
\node at (1.5, -0.5) {$R_{a_2, b_2}$};

% Label the points
\node[below left] at (a1) {$a_1$};
\node[below right] at (b1) {$b_1$};
\node[below left] at (a2) {$a_2$};
\node[below right] at (b2) {$b_2$};

% Indicate the shared sub-curve

\end{tikzpicture}
        \caption{In this example, $e_1 = (a_1, b_1), e_2 = (a_2, b_2)$ and region $R_{e_1}$ is contained within region $R_{e_2}$.}
        \label{region}
    \end{figure}
\end{proof}

By Lemma~\ref{lem:laminar}, the regions $R_{e}$, $e\in E(H)$, naturally form a laminar family. Create a rooted tree $\mathcal{T}$ where each node corresponds to a region $R_{e}$, and a region $R_{e_1}$ is an ancestor of another region $R_{e_2}$ if and only if $R_{e_1}\supseteq R_{e_2}$.

Consider any iteration of the while-loop of \Cref{alg-prune}. Define a weight ratio 
\begin{equation} \label{eq:wtratio}
\alpha = \frac{\wts\brac{H\setminus (F^\nw\cup F^\old)} }{ \wts\brac{G_{\opt, \eps}} }.
\end{equation}
The following structural property will be important to our analysis.

\begin{lemma}[structural property]\label{struct}
	There exists a shortest path $\rho^*$ in $G_{\opt, \epsilon}$ together with an edge set $F^*\subseteq E(H)\setminus (F^\nw\cup F^\old)$ such that:
    \begin{enumerate}[(1)]
        \item Each edge $e\in F^*$ is $\frac{2}{3}$-hanging on $\rho^*$;
        \item The total weight $\wts\brac{F^*}$ of $F^*$ satisfies $\wts\brac{F^*}\geq \frac{\alpha}{6}\cdot \wts\brac{\rho^*}$;
        \item   For each edge $e\in F^*$, let $(a_e, b_e)$ be the hanging points of $e$ on $\rho^*$; then all the sub-path intervals $\{\rho^*[a_e, b_e]: e\in F^*\}$ form a laminar family.
    \end{enumerate}
\end{lemma}
\begin{proof}
   The main part of this proof is to construct a set $\Gamma$ of shortest paths of the form $\gamma_{e}$ for some edges $e\in E(H)\setminus (F^\nw\cup F^\old)$,   such that 
  (1) each edge $e\in E(H)\setminus (F^\nw\cup F^\old)$ is $\frac{2}{3}$-hanging on some path $\gamma\in \Gamma$, and 
  (2) the total length $\wts(\Gamma)$ of all paths in $\Gamma$ is at most $6\wts\brac{G_{\opt, \epsilon}}$.
  Having constructed $\Gamma$, the proof will be concluded by applying the pigeon-hold principle.

\paragraph{Construction of $\Gamma$.}
 Initialize $\Gamma\leftarrow \emptyset$. 
 Traverse the nodes of the tree $\mathcal{T}$ in a depth-first-search order. For each node $R_e\in \mathcal{T}$ such that $e\in E(H)\setminus (F^\nw\cup F^\old)$, let $R_{e_0}\supseteq R_e$ be the lowest ancestor of $R_e$ in $\mathcal{T}$ such that $\gamma_{e_0}\in \Gamma$ (if $\gamma_{e}\notin \Gamma$ for any ancestor $R_e$ in $\mathcal{T}$, then $\gamma_{e_0}\leftarrow \emptyset$). Let $\eta = \gamma_e\cap \gamma_{e_0}$; note that $\gamma_e\cap \gamma_{e_0}$ is empty or a path because both $\gamma_e$ and $\gamma_{e_0}$ are unique shortest paths in $G_{\opt, \epsilon}$. Then, we add $\gamma_e$ to $\Gamma$ iff $\wts\brac{\eta} < \frac{2}{3}\cdot \wts(e)$. See \Cref{fig:dfs} for an illustration.

 \begin{figure}[!htb]
     \centering
     \begin{tikzpicture}[scale=2]
  % Main curve ρ (rho)

  %\draw[thick] plot[smooth] coordinates {(0,0) (0.4, 0.7) (1,1)};

  \draw[thick, color=orange] plot[smooth] coordinates {(1,1) (2, 1.1) (3,1)};

  %\draw[thick] plot[smooth] coordinates {(3,1) (3.6, 0.7) (4, 0)};

  \node at (2, 0.8) {$\eta$};
  
  % Points v_i and v_j on ρ
  \node[circle, fill=black, inner sep=1.5pt] (x) at (1,1) {};
  \node[circle, fill=black, inner sep=1.5pt] (y) at (3,1) {};
  
  % Points a and b
  \node[circle, fill=black, inner sep=1.5pt, label=below:$a_1$] (a1) at (1,0) {};
  \node[circle, fill=black, inner sep=1.5pt, label=below:$b_1$] (b1) at (3,0) {};

  \node[circle, fill=gray, inner sep=1.5pt, label=below:{\color{gray} $a_2$}] (a2) at (0.5,-0.5) {};
  \node[circle, fill=gray, inner sep=1.5pt, label=below:{\color{gray} $b_2$}] (b2) at (3.5,-0.5) {};

  \node[circle, fill=gray, inner sep=1.5pt, label=below:{\color{gray} $a_3$}] (a3) at (0,-1) {};
  \node[circle, fill=gray, inner sep=1.5pt, label=below:{\color{gray} $b_3$}] (b3) at (4,-1) {};
  
  \node[circle, fill=black, inner sep=1.5pt, label=below:$a_4$] (a4) at (-0.5,-1.5) {};
  \node[circle, fill=black, inner sep=1.5pt, label=below:$b_4$] (b4) at (4.5,-1.5) {};
  
  % Straight edge (a,b)
  \draw[thick] (a1) -- (b1);
  \draw[thick, color=gray] (a2) -- (b2);
  \draw[thick, color=gray] (a3) -- (b3);
  \draw[thick] (a4) -- (b4);
  
  \draw[thick, color=orange] (a1) to[out=100, in=250] (x);
  \draw[thick, color=orange] (b1) to[out=80, in=-80] (y);
  \draw[thick, color=gray] (a2) to[out=100, in=240] (x);
  \draw[thick, color=gray] (b2) to[out=80, in=-70] (y);
  \draw[thick, color=gray] (a3) to[out=100, in=230] (x);
  \draw[thick, color=gray] (b3) to[out=80, in=-60] (y);
  \draw[thick, color=orange] (a4) to[out=100, in=220] (x);
  \draw[thick, color=orange] (b4) to[out=80, in=-50] (y);
\end{tikzpicture}
     \caption{In this picture, $e = (a_1, b_1)$, and $e_0 = (a_4, b_4)$. If $\wts(\eta) < \frac{2}{3}\wts(e)$, we would add $\gamma_e$ to $\Gamma$ which decreases $\Phi$. There are two intermediate regions $R_{(a_2, b_2)}$ and $R_{(a_3, b_3)}$ whose paths $\gamma_{(a_2, b_2)}, \gamma_{(a_3, b_3)}$ were not added to $\Gamma$.}
     \label{fig:dfs}
 \end{figure}
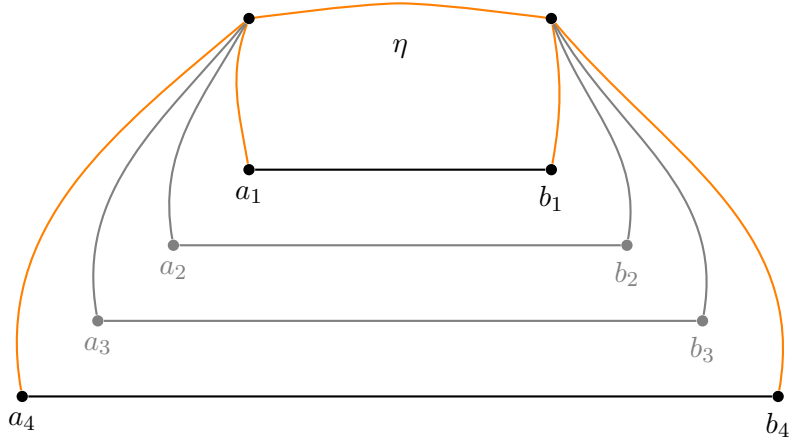

 \paragraph{Analysing the properties of $\Gamma$.}
	\begin{claim}\label{23hang}
		Each edge $e\in E(H)\setminus (F^\nw\cup F^\old)$ is $\frac{2}{3}$-hanging on a path in $\Gamma$.
	\end{claim}
	\begin{proof}
		Consider the step in which node $R_e$ is visited by the depth-first-search, and let $R_{e_0}$ be its lowest ancestor, with $\gamma_{e_0} \in \Gamma$, as defined above. If $\gamma_e\in \Gamma$, then by definition, $e$ is $1$-hanging on $\gamma_e \in \Gamma$. Otherwise $\gamma_e\notin \Gamma$, and by the construction we have $\wts(\eta)\geq\frac{2}{3}\cdot \wts(e)$, where $\eta = \gamma_e \cap \gamma_{e_0}$, which means $e$ is $\frac{2}{3}$-hanging on $\gamma_{e_0} \in \Gamma$.
	\end{proof}
	
	\begin{claim}\label{totalweight}
		The total weight $\wts\brac{\Gamma}$ is at most $6\cdot\wts\brac{G_{\opt, \epsilon}}$.
	\end{claim}
	\begin{proof}
		For each edge $f\in E_{\opt, \epsilon}$, let $P_f^1$ and $P_f^2$ be the two faces containing $f$ in the planar embedding of $G$ (possibly $P_f^1=P_f^2$ if $f$ is a cut edge). Define the following potential function $\Phi(f)$ with respect to $\Gamma$:
 		$$\Phi(f) = \begin{cases}
			0	&	\text{if both }P_f^1 \text{ and } P_f^2 \text{ are contained in some }R_{e} \text{ where }f\in\gamma_{e}\in \Gamma ,\\
			3\wts(f)		&	\text{if only one of }P_f^1 \text{ and } P_f^2 \text{ is contained in some }R_{e} \text{ where }f\in\gamma_{e}\in \Gamma ,\\
			6\wts(f)	&	\text{if neither }P_f^1 \text{ nor } P_f^2 \text{ is contained in  any }R_{e} \text{ where }f\in\gamma_{e}\in \Gamma .
		\end{cases}$$
Define the following sum as the overall potential (in the sum below, we want to count each $f$ at most once):
$$\Phi ~=~ \wts\brac{\Gamma} + \sum_{e\in E(H)\setminus (F^\nw\cup F^\old), f\in \gamma_e}\Phi(f).$$	
  Clearly, we have $\Phi\leq 6\wts\brac{G_{\opt, \epsilon}}$ at the beginning when $\Gamma = \emptyset$. We show that $\Phi$ never increases in the course of the algorithm. Consider any step when we add a certain path $\gamma_e$ to $\Gamma$. 
        Since $R_e$ is enclosed by a simple curve, it contains exactly one face $P_f\in \{P_f^1, P_f^2\}$ for any edge $f\in \gamma_e\setminus \eta$. More importantly, since the algorithm visits all the nodes on $\mathcal{T}$ in a depth-first-search order, $P_f$ was not included in any region $R_{e'}$ before where $f\in \gamma_{e'}\in \Gamma$.
		
		Therefore, $\Phi(f)$ decreases by $3\wts(f)$, and so the sum $\sum_{e\in E(H)\setminus (F^\nw\cup F^\old), f\in \gamma_e}\Phi(f)$ decreases by at least $3\cdot \brac{\wts(\gamma_e) - \wts(\eta)}> \wts(\gamma_e)$. On the other hand, $\wts(\Gamma)$ increases by at most $\wts(\gamma_e)$, and so overall $\Phi$ deceases.
	\end{proof}

 \paragraph{Conclusion of the proof.}
	By \Cref{23hang}, for every edge $e\in E(H)\setminus (F^\nw\cup F^\old)$, there exists a path in $\Gamma$ path on which $e$ can $\frac{2}{3}$-hang. For any path $\rho\in \Gamma$, let $C_\rho\subseteq E(H)\setminus (F^\nw\cup F^\old)$ be the set of all edges $e$ which can $\frac{2}{3}$-hang on $\rho$. By the averaging argument and \Cref{totalweight}, there exists $\rho^*\in \Gamma$ such that
    $$\frac{\wts(C_{\rho^*})}{\wts(\rho^*)}
    \geq \frac{\sum_{\rho\in \Gamma}\wts(C_\rho)}{\sum_{\rho\in\Gamma}\wts(\rho)} 
    \geq \frac{\sum_{e\in E(H)\setminus (F^\nw\cup F^\old)}\wts(e)}{\sum_{\rho\in\Gamma}\wts(\rho)}
    \geq \frac{\sum_{e\in E(H)\setminus (F^\nw\cup F^\old)}\wts(e)}{6\cdot \wts(G_{\opt,\eps})} 
    = \alpha / 6 .$$
    Set $F^* = C(\rho^*)$. To verify the requirements, let $s, t$ be the two endpoints of $\rho^*$. By construction, we know that for any $e\in F^*$, the region $R_e$ is contained within region $R_{s, t}$. Also, since all edges in $F^*$ are on the same side of $\rho^*$, all the sub-path intervals $\{\rho^*[a_e, b_e]: e\in F^*\}$ should form a laminar structure. Check \Cref{fig:struct} for an illustration \qedhere.

    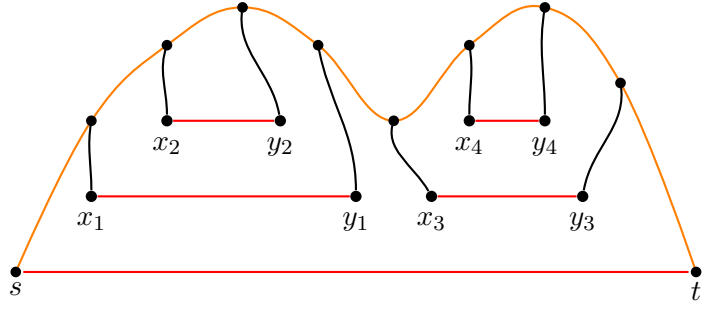
\begin{figure}
        \centering
        \begin{tikzpicture}
    % Define the 10 points for the concave curve
    \coordinate (P1) at (0,0);
    \coordinate (P2) at (1,2);
    \coordinate (P3) at (2,3);
    \coordinate (P4) at (3,3.5);
    \coordinate (P5) at (4,3);   % Start of concavity
    \coordinate (P6) at (5,2);   % Concave point
    \coordinate (P7) at (6,3);   % End of concavity
    \coordinate (P8) at (7,3.5);
    \coordinate (P9) at (8,2.5);
    \coordinate (P10) at (9,0);

    % Draw the curve through the points
    \draw[thick, smooth, tension=0.7, color=orange] plot coordinates {
        (P1) (P2) (P3) (P4) (P5) (P6) (P7) (P8) (P9) (P10)
    };

    % Optionally, draw the points for clarity
    \foreach \i in {1,...,10} {
        \fill (P\i) circle (2pt);
    }
    \node[below] at (P1) {$s$};
    \node[below] at (P10) {$t$};

    \draw[thick, red] (0.1, 0) -- (8.9, 0);
    
    % Label the points if needed
    % \foreach \i in {1,...,10} {
    %     \node[below] at (P\i) {$P_\i$};
    % }

      % Points a and b
    \node[circle, fill=black, inner sep=1.5pt, label=below:$x_1$] (x1) at (1,1) {};
    \node[circle, fill=black, inner sep=1.5pt, label=below:$y_1$] (y1) at (4.5, 1) {};
    \draw[thick, color=red] (x1) -- (y1);
    \draw[thick] (x1) to[out=90, in=-100] (P2);
    \draw[thick] (y1) to[out=90, in=-80] (P5);

    \node[circle, fill=black, inner sep=1.5pt, label=below:$x_2$] (x2) at (2,2) {};
    \node[circle, fill=black, inner sep=1.5pt, label=below:$y_2$] (y2) at (3.5,2) {};
    \draw[thick, color=red] (x2) -- (y2);
    \draw[thick] (x2) to[out=90, in=-110] (P3);
    \draw[thick] (y2) to[out=100, in=-100] (P4);

    \node[circle, fill=black, inner sep=1.5pt, label=below:$x_3$] (x3) at (5.5,1) {};
    \node[circle, fill=black, inner sep=1.5pt, label=below:$y_3$] (y3) at (7.5,1) {};
    \draw[thick, color=red] (x3) -- (y3);
    \draw[thick] (x3) to[out=120, in=-110] (P6);
    \draw[thick] (y3) to[out=80, in=-80] (P9);

    \node[circle, fill=black, inner sep=1.5pt, label=below:$x_4$] (x4) at (6,2) {};
    \node[circle, fill=black, inner sep=1.5pt, label=below:$y_4$] (y4) at (7,2) {};
    \draw[thick, color=red] (x4) -- (y4);
    \draw[thick] (x4) to[out=90, in=-80] (P7);
    \draw[thick] (y4) to[out=90, in=-100] (P8);
\end{tikzpicture}
        \caption{The path $\rho^*$ is drawn as the orange curve, and $C(\rho^*)$ are the red edges which are $\frac{2}{3}$-hanging on $\rho^*$, and the black curves are shortest paths in $G_{\opt,\eps}$.}
        \label{fig:struct}
    \end{figure}
\end{proof}

Next, our main goal is to show that our dynamic programming finds a path $\rho[s^*, t^*, L^*]$ which would be a good approximation to the path $\rho^*$ in $G_{\opt, \epsilon}$, and the multiset $P[s^*, t^*, L^*]$ would also be a good approximation to the set $F^*$. The analysis consists of two steps: first, we will show that the multi-set $P[s^*, t^*, L^*]$ contains a large number of edges in terms of edge weights; secondly, since $P[s^*, t^*, L^*]$ is a multi-set, we need to upper bound the amount of duplications so as to prove a true lower bound on $P[s^*, t^*, L^*]$.

\subsubsection{Lower Bounding $\DP[s^*, t^*, L^*]$}
In this subsection, we show that the multi-set $P[s^*, t^*, L^*]$ is relatively heavy (in terms of $\alpha$, see \Cref{eq:wtratio}) compared to $\rho[s^*, t^*, L^*]$.

\begin{lemma}\label{lowerbd}
	$\DP[s^*, t^*, L^*]\geq \frac{\alpha}{6}\cdot \wts\brac{\rho[s^*, t^*, L^*]}$.
\end{lemma}
\begin{proof}
	Let $\rho^*, F^*$ be the path and edge sets provided by \Cref{struct}. Let $s$ and $t$ be the two endpoints of the path $\rho^*$ which is unknown to our algorithm, and set $L = \wts\brac{\rho^*}$. Since the ratio $$\DP[s^*, t^*, L^*] / \wts\brac{\rho[s^*, t^*, L^*]}$$ is maximized by the triple $(s^*, t^*, L^*)$, it suffices to show that:
    $$\DP[s, t, L]\geq \frac{\alpha}{6}\cdot \wts\brac{\rho[s, t, L]} .$$
    
    For every edge $e\in F^*$, let $(a_e, b_e)$ be the hanging points of $e$ on $\rho^*$. According to the proof of \Cref{struct}, all the sub-path intervals $\{\rho^*[a_e, b_e]: e\in F^*\}$ form a laminar family $\mathcal{L}$. For notational convenience, we could naturally rewrite any sub-path $\rho^*[a, b]$ in an interval manner $[a, b]$. We can view $\mathcal{L}$ as a tree where each node corresponds to an interval $[a, b]$ together with a weight: \begin{equation}\label{eq:p}
	p_{[a, b]} = \sum_{e\in F^*, \{a_e, b_e\} = \{a, b\}}\wts(e).
    \end{equation}
    Note that multiple edges $e\in F^*$ might be hanging at the same vertex pair $(a, b)$, which is why we define $p_{[a, b]}$ by a summation. According to \Cref{struct}, we have: $$\sum_{[a, b]\in \mathcal{L}}p_{[a, b]} = \wts(F^*)\geq \frac{\alpha}{6}\cdot \wts(\rho^*) .$$
    Thus, our goal is to show that $\DP[s, t, L]$ is at least the total weight over all nodes in the tree $\mathcal{L}$.
	
	For technical convenience, we will modify $\mathcal{L}$ to make it a binary tree. First, if the root node of $\mathcal{L}$ is not $[s, t]$, then add a root node corresponding to the interval $[s, t]$ with weight $p_{[s, t]} = 0$. In general, while there is a node $[a, b]$ with more than two children $[a_1, b_1], [a_2, b_2], \ldots, [a_k, b_k]$, where $k>2$, insert an intermediate node $[b_1, b]$ as a child of $[a, b]$ of zero weight $p_{[b_1, b]} = 0$; and move the children $[a_2, b_2], \ldots, [a_k, b_k]$ below $[b_1, b]$.
	
	We will derive a sequence of lower bounds for $\DP[s, t, L]$ using the tree $\mathcal{L}$ and the dynamic programming rules. Intuitively, we will keep unpacking the term $\DP[s, t, L]$ recursively using the dynamic programming rule, and collect all the additive terms of the form $\wts(B[\cdot, \cdot])$ that show up during the unpacking procedure, and argue that the total sum of these additive terms is lower bounded by $\wts(F^*)$.
	
	To do this, to each node $[a, b]$ of $\mathcal{L}$, we will associate with a value $q_{[a, b]}\geq 0$ such that:
	\begin{equation}\label{eq:q}
    \DP\left[a, b, \wts\brac{\rho^*[a, b]}\right]\geq \sum_{[c, d]\in \mathcal{L}[a, b]}q_{[c, d]},
    \end{equation}
	where $\mathcal{L}[a, b]$ denotes the subtree of $\mathcal{L}$ rooted at $[a, b]$.
	We will define the values in $\{q_{[a, b]}: [a, b]\in \mathcal{L}\}$ as following. 
    \begin{itemize}
        \item $[a, b]$ is a leaf node of $\mathcal{L}$. In this case, assign $q_{[a,b]}=\wts(B[a,b])$.

        \item $[a, b]$ has exactly one child $[a_1, b_1]$. In this case, we assign
	$$\begin{aligned}
		q_{[a, b]} =& \wts(B[a, a_1]) + \wts(B[b_1, b])\\
		&+\mathbf{1}\big[\max\{\wts(\rho^*[a, b_1]), \wts(\rho^*[b_1, b])\} < \floor{\wts(\rho^*[a, b])}_2\big]\cdot \wts\brac{B[a, b]}\\
		&+ \mathbf{1}\big[\max\{\wts(\rho^*[a_1, b_1]), \wts(\rho^*[a, a_1])\} < \floor{\wts(\rho^*[a, b_1])}_2\big]\cdot \wts\brac{B[a, b_1]}.
	\end{aligned}$$

        \item $[a, b]$ has two children $[a_1, b_1]$ and $[a_2, b_2]$ (where $a_1$ lies between $a$ and $a_2$). In this case, we assign
	$$\begin{aligned}
		q_{[a, b]} &= \wts(B[a, a_1]) + \wts(B[b_1, a_2]) + \wts(B[b_2, b])\\
		&+\mathbf{1}\left[\max\{\wts(\rho^*[a, b_1]), \wts(\rho^*[b_1, b])\} < \floor{\wts(\rho^*[a, b])}_2\right]\cdot \wts\brac{B[a, b]}\\
		&+ \mathbf{1}\left[\max\{\wts(\rho^*[a, a_1]), \wts(\rho^*[a_1, b_1])\} < \floor{\wts(\rho^*[a, b_1])}_2\right]\cdot \wts\brac{B[a, b_1]}\\
		&+ \mathbf{1}\left[\max\{\wts(\rho^*[b_1, b_2]), \wts(\rho^*[b_2, b])\} < \floor{\wts(\rho^*[b_1, b])}_2\right]\cdot \wts\brac{B[b_1, b]}\\
		&+ \mathbf{1}\left[\max\{\wts(\rho^*[b_1, a_2]), \wts(\rho^*[a_2, b_2])\} < \floor{\wts(\rho^*[b_1, b_2])}_2\right]\cdot \wts\brac{B[b_1, b_2]}.
	\end{aligned}$$
    \end{itemize}

    Next, we verify \Cref{eq:q} based on definitions of $q_{[\cdot, \cdot]}$'s. This is done in a bottom up manner on the tree $\mathcal{L}$. The inequality holds trivially for leaf nodes by definition of $q_{[\cdot, \cdot]}$. For a non-leaf node $[a, b]$, if it has only one child $[a_1, b_1]$, then according our dynamic programming rule, we have
	$$\begin{aligned}
		\DP[a, b, \wts(\rho^*[a, b])] \geq& \DP[a, b_1, \wts(\rho^*[a, b_1])] + \DP[b_1, b, \wts(\rho^*[b_1, b])]\\
		&+ \mathbf{1}\big[\max\{\wts(\rho^*[a, b_1]), \wts(\rho^*[b_1, b])\} < \floor{\wts(\rho^*[a, b])}_2\big]\cdot \wts\brac{B[a, b]}\\
		\geq& \DP[a, a_1, \wts(\rho^*[a, a_1])] + \DP[a_1, b_1, \wts(\rho^*[a_1, b_1])] +  \DP[b_1, b, \wts(\rho^*[b_1, b])]\\
		&+\mathbf{1}\big[\max\{\wts(\rho^*[a, b_1]), \wts(\rho^*[b_1, b])\} < \floor{\wts(\rho^*[a, b])}_2\big]\cdot \wts\brac{B[a, b]}\\
		&+\mathbf{1}\big[\max\{\wts(\rho^*[a_1, b_1]), \wts(\rho^*[a, a_1])\} < \floor{\wts(\rho^*[a, b_1])}_2\big]\cdot \wts\brac{B[a, b_1]}\\
		\geq& \wts(B[a, a_1]) + \DP[a_1, b_1, \wts(\rho^*[a_1, b_1])]  + \wts(B[b, b_1])\\
		&+\mathbf{1}\big[\max\{\wts(\rho^*[a, b_1]), \wts(\rho^*[b_1, b])\} < \floor{\wts(\rho^*[a, b])}_2\big]\cdot \wts\brac{B[a, b]}\\
		&+\mathbf{1}\big[\max\{\wts(\rho^*[a_1, b_1]), \wts(\rho^*[a, a_1])\} < \floor{\wts(\rho^*[a, b_1])}_2\big]\cdot \wts\brac{B[a, b_1]}\\
    \geq & \DP[a_1, b_1, \wts(\rho^*[a_1, b_1])] + q_{[a, b]}
	\end{aligned} .$$
    Then by induction, we can conclude \Cref{eq:q}; here, the third inequality is due to \Cref{dyn-basic}, and recall the definition that $\floor{\wts(\rho^*[a, b])}_2=2^k$ such that $\wts(\rho^*[a, b])\in [2^k,2^{k+1})$.

    Next, assume $[a, b]$ has two children $[a_1, b_1]$ and $[a_2, b_2]$ (where $a_1$ lies between $a$ and $a_2$). According to our dynamic programming rule, we have
	$$\begin{aligned}
		\DP[a, b, \wts(\rho^*[a, b])] 
    \geq & \DP[a, b_1, \wts(\rho^*[a, b_1])] + \DP[b_1, b, \wts(\rho^*[b_1, b])]\\
		&+ \mathbf{1}\big[\max\{\wts(\rho^*[a, b_1]), \wts(\rho^*[b_1, b])\} < \floor{\wts(\rho^*[a, b])}_2\big]\cdot \wts\brac{B[a, b]}
	\end{aligned} .$$

    To further expand the two terms $\DP[a_1, b_1, \wts(\rho^*[a_1, b_1])], \DP[a_2, b_2, \wts(\rho^*[a_2, b_2])]$, we apply again the dynamic programming rules together with \Cref{dyn-basic}, we have
    $$\begin{aligned}
	\DP[a, b_1, \wts(\rho^*[a, b_1])] \geq & \wts(B[a, a_1]) + \DP[a_1, b_1, \wts(\rho^*[a_1, b_1])] \\
		&+ \mathbf{1}\left[\max\{\wts(\rho^*[a, a_1]), \wts(\rho^*[a_1, b_1])\} < \floor{\wts(\rho^*[a, b_1])}_2\right]\cdot \wts\brac{B[a, b_1]}.
    \end{aligned}$$
    $$\begin{aligned}
	\DP[b_1, b, \wts(\rho^*[b)1, b])] \geq & \wts(B[b_1, a_2]) + \DP[a_2, b_2, \wts(\rho^*[a_2, b_2])] + \wts(B[b_2, b])\\
    &+ \mathbf{1}\left[\max\{\wts(\rho^*[b_1, b_2]), \wts(\rho^*[b_2, b])\} < \floor{\wts(\rho^*[b_1, b])}_2\right]\cdot \wts\brac{B[b_1, b]}\\
    &+ \mathbf{1}\left[\max\{\wts(\rho^*[b_1, a_2]), \wts(\rho^*[a_2, b_2])\} < \floor{\wts(\rho^*[b_1, b_2])}_2\right]\cdot \wts\brac{B[b_1, b_2]}.
    \end{aligned}$$
    Summing up the above three inequalities, we have
    $$\begin{aligned}
        \DP[a, b, \wts(\rho^*[a, b])] 
    \geq & \DP[a_1, b_1, \wts(\rho^*[a_1, b_1])] + \DP[a_2, b_2, \wts(\rho^*[a_2, b_2])]\\
		&+ \wts(B[a, a_1]) + \wts(B[b_1, a_2]) + \wts(B[b_2, b])\\
		&+ \mathbf{1}\left[\max\{\wts(\rho^*[a, b_1]), \wts(\rho^*[b_1, b])\} < \floor{\wts(\rho^*[a, b])}_2\right]\cdot \wts\brac{B[a, b]}\\
		&+ \mathbf{1}\left[\max\{\wts(\rho^*[a, a_1]), \wts(\rho^*[a_1, b_1])\} < \floor{\wts(\rho^*[a, b_1])}_2\right]\cdot \wts\brac{B[a, b_1]}\\
		&+ \mathbf{1}\left[\max\{\wts(\rho^*[b_1, b_2]), \wts(\rho^*[b_2, b])\} < \floor{\wts(\rho^*[b_1, b])}_2\right]\cdot \wts\brac{B[b_1, b]}\\
		&+ \mathbf{1}\left[\max\{\wts(\rho^*[b_1, a_2]), \wts(\rho^*[a_2, b_2])\} < \floor{\wts(\rho^*[b_1, b_2])}_2\right]\cdot \wts\brac{B[b_1, b_2]}\\
    \geq&\DP[a_1, b_1, \wts(\rho^*[a_1, b_1])] + \DP[a_2, b_2, \wts(\rho^*[a_2, b_2])] + q_{[a, b]}.
    \end{aligned}$$
    By induction, we can show that
 $$\DP\left[s, t, L\right]\geq \sum_{[a, b]\in \mathcal{L}}q_{[a, b]}.$$
	Finally, let us we prove that: $$\sum_{[a, b]\in \mathcal{L}}q_{[a, b]}\geq \sum_{[a, b]\in \mathcal{L}}p_{[a, b]} .$$
    
    By definition~\eqref{eq:p}, the term $p_{[a, b]}$ is the total weight of edges $e\in F^*$ which are $\frac{2}{3}$-hanging on $\rho^*$ at $(a, b)$, and each $q_{[a, b]}$ is a sum of several terms of the form $\wts(B[\cdot, \cdot])$. So the proof strategy is to assign each edge $e\in F^*$ to a set $B[\cdot, \cdot]$ which contains $e$ and contributes to a value $q_{[a, b]}$.

    Consider any edge $e\in F^*$ which is $\frac{2}{3}$-hanging at $(a, b)$. Let $[c, d]$ be the lowest descendant of $[a, b]$ in $\mathcal{L}$ such that $\wts(\rho^*[c, d])\geq \floor{\wts(\rho^*[a, b])}_2$. There are several cases to consider.

    \begin{itemize}
        \item $[c, d]$ is a leaf node of $\mathcal{L}$.

        According to the assumption, we know that $\dist_G(c, d)\geq \frac{1}{1+\eps}\cdot \wts(\rho^*[c, d]) > \frac{1}{2(1+\eps)}\cdot \wts(\rho^*[a, b])$ and $e$ is $\frac{2}{3}$-hanging at $(a, b)$ on $\rho^*$. Since $c, d$ both are lying on the sub-path $\rho^*[a, b]$, by \Cref{hang} we know that $e$ is $\frac{1}{3(1+\eps)}$-hanging at $(c, d)$ on the shortest path $\pi_{c,d}$, and so $e\in B[c, d]$. By definition, $q_{[c, d]} = \wts(B[c, d])$, so we can assign $\wts(e)$ to $q_{[c, d]}$.

        \item $[c, d]$ has one child $[c_1, d_1]$ in $\mathcal{L}$.

        If $\wts(B[c, c_1])$ or $\wts(B[d, d_1])$ is at least $\floor{\wts[\rho^*[c, d]]}_2$, then by the same calculation as above, we can argue that $e$ is $\frac{1}{3(1+\eps)}$-hanging at one of $(c, c_1)$ or $(d_1, d)$ on shortest path $\pi_{a, a_1}$ or $\pi_{d_1, d}$, respectively.
        
        So, let us assume that $\wts(B[c, c_1]), \wts(B[d_1, d]) < \floor{\wts(\rho^*[c, d])}_2$. There are two sub-cases to verify.
        \begin{itemize}
            \item $\wts(\rho^*[c, d_1]) < \floor{\wts(\rho^*[c, d])}_2$.
            
            In this case, by definition of $q_{[c, d]}$, it concludes the term $\wts(B[c, d])$. As $\dist_G(c, d)\geq \frac{1}{1+\eps}\cdot \wts(\rho^*[c, d]) > \frac{1}{2(1+\eps)}\cdot \wts(\rho^*[a, b])$, we know that $e$ is $\frac{1}{3(1+\eps)}$-hanging at $(c, d)$ on shortest path $\pi_{s, t}$. So we can assign $e$ to $q_{[c, d]}$.

            \item $\wts(\rho^*[c, d_1]) \geq \floor{\wts(\rho^*[c, d])}_2$.

            In this case, since $[c, d]$ is the lowest descendant of $[a, b]$ such that $\wts(\rho^*[c, d])\geq \floor{\wts(\rho^*[a, b])}_2$, we have $\wts(\rho^*[c_1, d_1])< \floor{\wts(\rho^*[a, b])}_2\leq \floor{\wts(\rho^*[c, d])}_2$. As we have already assumed $\wts(\rho^*[c, c_1]) <  \floor{\wts(\rho^*[c, d])}_2$, $q_{[c, d]}$ contains the term $\wts(B[c, d_1])$. By the same calculation, we can show that $e$ is $\frac{1}{3(1+\eps)}$-hanging at $(c, d_1)$ and can be assigned to $q_{[c, d]}$.
        \end{itemize}

        \item $[c, d]$ has two children,  $[c_1, d_1]$ and $[c_2, d_2]$, in $\mathcal{L}$ and $c_1$ lies between $c$ and $c_2$ on $\rho^*$.
        
        By the choice of descendant $[c, d]$, we know that 
        $$\max\left\{\wts(\rho^*[c_1, d_1]), \wts(\rho^*[c_2, d_2])\right\} < \floor{\wts(\rho^*[a, b])}\leq \floor{\wts(\rho^*[c, d])}_2 .$$
        We can first assume that 
        $$\max\left\{\wts(\rho^*[c, c_1]), \wts(\rho^*[d_1, c_2]), \wts(\rho^*[d_2, d])\right\} < \floor{\wts(\rho^*[c, d])}_2 .$$        
        since otherwise we could assign $e$ to one of the three edge sets which contribute to $q_{[c, d]}$. Next, we need to discuss several cases.
	\begin{itemize}
		\item $\max\left\{\wts(\rho^*[c, d_1], \wts(\rho^*[d_1, d]))\right\} < \floor{\wts(\rho^*[c, d])}_2$.

        In this case, the term $\wts(B[c, d])$ would be included in the definition of $q_{[c, d]}$. As before, we can show that $e$ is $\frac{1}{3(1+\eps)}$-hanging at $(c, d)$, so we can assign $e\in B[c, d]$ to $q_{[c, d]}$.

        \item $\wts(\rho^*[c, d_1]) \geq \floor{\wts(\rho^*[c, d])}_2$.

        In this case, since we already know 
        $$\wts(\rho^*[c, c_1]), \wts(\rho^*[c_1, d_1]) < \floor{\wts(\rho^*[c, d])}_2 \leq \floor{\wts(\rho^*[c, d_1])}_2 ,$$
        the term $\wts(B[c, d_1])$ should be included in $q_{[c, d]}$. As $\wts(\rho^*[c, d_1]) \geq \floor{\wts(\rho^*[c, d])}_2$, for the same reason as before we can show $e\in B[c, d_1]$, and so we are able to assign $e$ to $q_{[c, d]}$.

        \item $\wts(\rho^*[d_1, d]) \geq \floor{\wts(\rho^*[c, d])}_2$.

        In this case, if $\wts(\rho^*[d_1, d_2]) < \floor{\wts(\rho^*[d_1, d])}_2$, then the term $\wts(B[d_1, d])$ would be included in the definition of $q_{[c, d]}$. As $\wts(\rho^*[d_1, d]) \geq \floor{\wts(\rho^*[c, d])}_2 > \frac{1}{2}\wts(\rho^*[a, b])$, we know that $e\in B[d_1, d]$, and so we can assign $e$ to $q_{[c, d]}$.

        Otherwise, we have $\wts(\rho^*[d_1, d_2]) \geq \floor{\wts(\rho^*[d_1, d])}_2\geq \floor{\wts(\rho^*[c, d])}_2$. Hence, by the choice of $[c, d]$, we must have
        $$\max\left\{\wts(\rho^*[d_1, c_2]), \wts(\rho^*[c_2, d_2]) \right\} < \floor{\wts(\rho^*[a, b])}_2\leq \floor{\wts(\rho^*[d_1, d_2])}_2 .$$
        Consequently, $q_{[c, d]}$ contains the term $\wts(B[d_1, d_2])$. As $\wts(\rho^*[d_1, d_2])\geq\floor{\wts(\rho^*[a, b])}_2$, we know that $e\in B[d_1, d_2]$, and so we can assign $e$ to $q_{[c, d]}$.
	\end{itemize}
    \end{itemize}
    
	In this way, we can show that any term $p_{[a. b]} > 0$ can be assigned to some term $q_{[c, d]}$, and so we have
	$$\DP\left[s, t, L\right]\geq \sum_{[a, b]\in \mathcal{L}}q_{[a, b]}\geq \sum_{[a, b]\in \mathcal{L}}p_{[a, b]}\geq \frac{\alpha}{6}\cdot \wts(\rho^*) .
    \qedhere $$
\end{proof}

\subsubsection{Upper Bounding the Multiplicity of $P[s^*, t^*, L^*]$} \Cref{lowerbd} implies that the weight of the multi-set $P[s^*, t^*, L^*]$ of edges that are $\frac{1}{3(1+\eps)}$-hanging on $\rho[s^*, t^*, L^*]$ is at least $\frac{\alpha}{6}\cdot \wts(\rho[s^*, t^*, L^*])$. However, it does not mean we are pruning a lot of weight by updating $F^\old\leftarrow F^\old\cup P[s^*, t^*, L^*]$ since the multi-set $P[s^*, t^*, L^*]$ might include many edges in $E(H)\setminus (F^\nw\cup F^\old)$ with high multiplicity. Next, our main goal is to upper bound the total amount of \emph{over-counting} (that is, the difference between the weight of the multi-set $P[s^*, t^*, L^*]$ and the corresponding (non-multi) set).

Suppose $\rho[s^*, t^*, L^*]$ is a (not necessarily simple) path $\rho = \langle (s^*=)u_1, u_2, \ldots, u_m(=t^*) \rangle$. Build a tree $\DPtree$ according to the maximizers of the dynamic programming table as follows. The tree is built in a top-down manner. Each node of $\DPtree$ is associated with an interval $[i, j]\subseteq [1, m]$.  The root node corresponds to interval $[1, m]$. For an arbitrary node $[i, j]$, consider the maximizer for computing the entry $\DP[u_i, u_j, \wts(\rho[u_i, u_j])]$. Since $\rho[u_i, u_j]$ is the walk $\langle u_i, u_{i+1}, \ldots, u_{j}\rangle$, the maximizer for entry $\DP[u_i, u_j, \wts(\rho[u_i, u_j])]$ is either  a vertex $u_k$ for $i<k<j$, or $\DP[u_i, u_j, \wts(\rho[u_i, u_j])]$ is simply equal to $\wts(B[u_i, u_j])$. In the former case, we leave $[i, j]$ as a leaf node on $\DPtree$; and in the latter case, we create two children of $[i, j]$ associated with $[i, k]$ and $[k, j]$.

\begin{definition}
	According to the dynamic programming rules, for any copy of edge $e$ in the multi-set $P[s^*, t^*, L^*]$, it should appear at some tree node $[i, j]$ by set $B[u_i, u_j]$. For convenience, we say that $e$ is \emph{hanging} at the interval $[i, j]$ on $\rho$.
\end{definition}

It is clear that any two different copies of the same edge $e$ in $P[s^*, t^*, L^*]$ should be hanging at distinct intervals since none of the sets $B[x, y]$ is a multi-set. Therefore,  any edge $e\in \left(E(H)\setminus (F^\nw\cup F^\old)\right)\cap P[s^*, t^*, L^*]$ is mapped to a set $I_e$ of distinct intervals in $\DPtree$.

\begin{lemma}\label{nest-depth}
	For any edge $e\in E(H)\setminus (F^\nw\cup F^\old)$ and any interval $[i, j]\in I_e$, the interval $[i, j]$ has at most $3$ ancestors in $I_e$ on $\DPtree$.
\end{lemma}
\begin{proof}
	Suppose, for the sake of contradiction, that $[i, j]$ has four different ancestors $[i_1, j_1]\subset [i_2, j_2]\subset [i_3, j_3]\subset [i_4, j_4]$ in the tree $\DPtree$ which are all in the set $I_e$. On the one hand, $e$ should belong to $B[i, j]\cap B[i_1, j_1]\cap B[i_2, j_2]\cap B[i_3, j_3]\cap B[i_4, j_4]$, which implies that $e$ is $\frac{1}{3(1+\eps)}$-hanging at all of the pairs $(i, j), (i_1, i_1), (i_2, j_2), (i_3, j_3), (i_4, j_4)$, and therefore we have
	$$\begin{aligned}
	    \wts(e) &\geq \dist_G(u_{i_4}, u_{j_4}) \geq \dist_G(u_{i_3}, u_{j_3}) \geq\dist_G(u_{i_2}, u_{j_2})\\
        &\geq\dist_G(u_{i_1}, u_{j_1}) \geq\dist_G(u_{i}, u_{j}) \geq \frac{1}{3(1+\eps)}\wts(e).
	\end{aligned}$$
	However, by our rule of dynamic programming, we know that $$\begin{aligned}
		(1+\epsilon)\wts(e) &\geq \floor{\wts(\rho[u_{i_4}, u_{j_4}])}_2 > \floor{\wts(\rho[u_{i_3}, u_{j_3}])}_2 > \floor{\wts(\rho[u_{i_2}, u_{j_2}])}_2\\
		&> \floor{\wts(\rho[u_{i_1}, u_{j_1}])}_2 > \floor{\wts(\rho[u_{i}, u_{j}])}_2 \geq \frac{1}{6(1+\eps)}\wts(e).
	\end{aligned}$$
	This is impossible as there cannot be $5$ different integral powers of $2$ in the interval $\left[\frac{\wts(e)}{6(1+\eps)}, (1+\eps)\wts(e)\right]$ when $\eps < 0.1$.
\end{proof}

For each $e\in E(H)\setminus (F^\nw\cup F^\old)$, let $J_e\subseteq I_e$ be the set of lowest nodes on $\DPtree$. According to \Cref{nest-depth}, we know that $|J_e|\geq \frac{1}{4}|I_e|$, and therefore $$\sum_{e\in E(H)\setminus (F^\nw\cup F^\old)}|J_e|\cdot\wts(e)\geq \frac{1}{4}\sum_{e\in E(H)\setminus (F^\nw\cup F^\old)}|I_e|\cdot\wts(e) = \frac{\beta}{4}\cdot\wts(\rho)\geq  \frac{\alpha}{48}\cdot\wts(\rho).$$
Recall that $\beta = \DP[s^*, t^*, L^*] / \wts(\rho[s^*, t^*, L^*])$. To lower bound the total weight of edges in $P[s^*, t^*, L^*]$ without any double-counting, it suffices to lower bound the quantity $\sum_{e\in E(H)\setminus F^*}\mathbf{1}[J_e\neq\emptyset]\cdot \wts(e)$.

\paragraph{Lower bounding $\sum_{e\in E(H)\setminus (F^\old\cup F^\nw)}\mathbf{1}[J_e\neq\emptyset]\cdot \wts(e)$ via shortcuts.}
Intuitively speaking, if a set $J_e$ is very large, say $|J_e| > 20$, then $e$ is hanging at many different positions simultaneously $[c_1, d_1], [c_2, d_2], \ldots, [c_{20}, d_{20}]$ on $\rho$. Then, $\wts(\rho[u_{c_1}, u_{d_{20}}])$ would be much larger than $\dist_G(u_{c_1}, u_{d_{20}})$, and so we could shortcut the path $\rho$ by replacing the sub-path $\rho[u_{c_1}, u_{d_{20}}]$ with the shortest path between $u_{c_1}, u_{d_{20}}$ and reduce $\wts(\rho)$ significantly. Since we knew that $\rho$ was already a $(1+\eps)$-approximate shortest path between $s$ and $t$ at the beginning, we could shortcut $\rho$ by at most $\eps\cdot \dist_G(s, t)$ in total length, which would upper bound the total amount of multiplicities of $J_e$'s.

Let us now formalize the above idea. We will design a procedure which repeatedly finds shortcuts along $\rho$ and maintains some invariants. 
\begin{framed}
    \noindent A shortcut structure on $\rho$.
    \begin{itemize}[leftmargin=*]
	\item A subset $K_e\subseteq J_e$ for all $e\in E(H)\setminus (F^\nw\cup F^\old)$.
	\item A sequence of disjoint intervals $\interval = \{[a_1, b_1], [a_2, b_2], \ldots, [a_k, b_k]\}$ such that $b_i\leq a_{i+1}$ for $1\leq i<k$.
	\item A sequence of shortcut (not necessarily simple) paths $\paths = \{\eta_1, \eta_2, \ldots, \eta_{k+1}\}$, where $\eta_i$ is a path in $G$ connecting $u_{b_{i-1}}$ and $u_{a_i}$ such that $\wts(\eta_i)< \wts(\rho[u_{b_{i-1}}, u_{a_i}])$; conventionally, set $s = u_{b_0}$ and $t = u_{a_{k+1}}$. The paths $\eta_i$ are the so-called shortcuts.
    
    Denote $\eta = \eta_1\circ\rho[u_{a_1}, u_{b_1}]\circ \eta_2\circ \rho[u_{a_2}, u_{b_2}]\circ\cdots\circ \rho[u_{a_k}, u_{b_k}]\circ \eta_{k+1}$ which is a (not necessarily simple) path between $s$ and $t$. 
\end{itemize}
\end{framed}

We formulate two properties in terms of the above notation.
 \begin{invariant}\label{inv}
    Our shortcut algorithm will preserve the following properties.
	\begin{enumerate}[(1)]
	\item For any $e\in E(H)\setminus (F^\nw\cup F^\old)$, if $|K_e| > 20$, then there exists an interval $[a, b]\in \interval$ that contains all $[c, d]\in K_e$. 
	\item $\sum_{e\in E(H)\setminus (F^\nw\cup F^\old)} |J_e\setminus K_e|\cdot \wts(e) \leq 10\beta\cdot\brac{\wts(\rho) - \wts(\eta)}$.
	\end{enumerate}
\end{invariant}

Let us now describe the shortcut algorithm which will only refine the intervals in $\interval$. At the beginning, initialize $K_e \leftarrow J_e$ for all $e\in E(H)\setminus (F^\nw\cup F^\old)$, $\interval \leftarrow \{[1, m]\}$, and $\eta \leftarrow \rho$. Note that \Cref{inv} holds initially. Our strategy is to successively decrease the length of $\eta$ as long as some nonempty set $K_e$ has size larger than $20$. We need to define the following notion of \emph{span} for each edge $e\in E(H)\setminus (F^\nw\cup F^\old)$ such that $|K_e| > 20$.
\begin{definition}[span]
	For each edge $e\in E(H)\setminus (F^\nw\cup F^\old)$ such that $|K_e| > 20$, by \Cref{inv}(1), there exists an interval $[a, b]\in\interval$ which contains all intervals in $K_e$. Assume $[c_1, d_1], [c_2, d_2], \ldots, [c_l, d_l]$ are all elements of $K_e$, then define the \emph{span} of $e$ (with respect to the current $\paths, \interval$) to be $\spn(e) = \wts(\rho[u_{c_1}, u_{d_l}])$.
\end{definition}
\begin{claim}
	For any $e\in E(H)\setminus (F^\nw\cup F^\old)$ such that $|K_e| > 20$, we have $\spn(e)> \frac{20}{3(1+\eps)}\wts(e)$.
\end{claim}
\begin{proof}
	This is straightforward since $e$ is $\frac{1}{3(1+\eps)}$-hanging on $\rho$ at $(u_{c_i}, u_{d_i})$ and $l>20$.
\end{proof}

In each iteration of the shortcut procedure, as long as there exists $e\in E(H)\setminus (F^\nw\cup F^\old)$  such that $|K_e| > 20$, let $f\in E(H)\setminus (F^\nw\cup F^\old)$ be the edge such that $|K_f| > 20$ and $\spn(f)$ is \emph{maximized}. Let $[c_1, d_1], [c_2, d_2], \ldots, [c_l, d_l]$ be all the elements in $K_f$, and we already know $l > 20$. Next, we show how to update $\interval$, $\paths$ and sets $\{K_e: e\in E(H)\setminus (F^\nw\cup F^\old)\}$.
    \begin{itemize}
	\item Updating $\paths, \interval$.
	
	Suppose $f = (x, y)$. By definition of $I_f$, each $f$ is $\frac{1}{3(1+\eps)}$-hanging at vertex pair $(u_{c_1}, u_{d_1})$ and $(u_{c_l}, u_{d_l})$. Therefore, we have
    $$\dist_G(x, u_{c_1})\leq \brac{1 -\frac{1}{3(1+\eps)}} \wts(f) < \frac{3}{4}\wts(f) ,$$
    $$\dist_G(y, u_{d_l}) < \brac{1 -\frac{1}{3(1+\eps)}} \wts(f)<\frac{3}{4}\wts(f) .$$ 
    Hence, by triangle inequality, we have
    $$\dist_G(u_{c_1}, u_{d_l}) < \dist_G(u_{c_1}, x) + \wts(f) + \dist_G(y, u_{d_l}) < \frac{5}{2}\wts(f) .$$
    Let $\lambda$ be the shortest path in $G$ between $u_{c_1}$ and $u_{d_l}$. Then, replace $[a, b]$ with $[a, c_1]$ and $[d_l, b]$ in $\interval$, and add $\lambda$ to $\paths$ as a new shortcut between $c_1$ and $d_l$.
	
	\item Updating $\{K_e: e\in E(H)\setminus (F^\nw\cup F^\old)\}$.
	
	Let $[i_1, j_1]$ be the maximal set corresponding to a tree node in $\DPtree$ such that $i_1 < c_1 < j_1$, and there exists an edge $e_1, |K_{e_1}| > 20$ and $[i_1, j_1]\in K_{e_1}$; if no such edge $e_1$ exists, then simply set $i_1 = j_1 = c_1$. Symmetrically, define $[i_2, j_2]$ to be the maximal interval such that $i_2 < d_l < j_2$, and there exists $e_2$ such that $|K_{e_2}| > 20$ and $[i_2, j_2]\in K_{e_2}$. Note that since both $[i_1, j_1], [i_2, j_2]$ are tree nodes in $\DPtree$, these two intervals are disjoint internally.
	
	To update the sets $\{K_e: e\in E(H)\setminus (F^\nw\cup F^\old)\}$, for each $K_e$ such that $|K_e| > 20$, if any $[i, j]\in K_e$ satisfies $[i, j]\subseteq [i_1, j_2]$, then remove $[i, j]$ from $K_e$.
\end{itemize}
See \Cref{fig:shortcut} for an illustration. Next, we show that these updates preserve \Cref{inv}.

\begin{figure}
    \centering
    \begin{tikzpicture}[scale=1.5]
    % Define the style for the nodes

    % Place the nodes
    \node[circle, fill=black, inner sep=1.5pt, label=below:$s$] (s) at (0,0) {};
    \node[circle, fill=black, inner sep=1.5pt, label=below:$u_{c_1}$] (c1) at (2,0) {};
    \node[circle, fill=black, inner sep=1.5pt, label=below:$u_{d_1}$] (d1) at (2.5,0) {};
    \node[circle, fill=black, inner sep=1.5pt, label=below:$u_{c_2}$] (c2) at (4,0) {};
    \node[circle, fill=black, inner sep=1.5pt, label=below:$u_{d_2}$] (d2) at (4.5,0) {};
    \node[circle, fill=black, inner sep=1.5pt, label=below:$u_{c_3}$] (c3) at (6,0) {};
    \node[circle, fill=black, inner sep=1.5pt, label=below:$u_{d_3}$] (d3) at (6.5,0) {};
    \node[circle, fill=black, inner sep=1.5pt, label=below:$t$] (t) at (8.5,0) {};

    \node[circle, fill=black, inner sep=1.5pt] (n1) at (3,0) {};
    \node[circle, fill=black, inner sep=1.5pt] (n2) at (3.5,0) {};
    \node[circle, fill=black, inner sep=1.5pt] (n3) at (5,0) {};
    \node[circle, fill=black, inner sep=1.5pt] (n4) at (5.5,0) {};

    \node[circle, fill=black, inner sep=1.5pt] (x) at (3.5,1) {};
    \node[circle, fill=black, inner sep=1.5pt] (y) at (5,1) {};

    \node[circle, fill=black, inner sep=1.5pt] (x1) at (2.75, -1) {};
    \node[circle, fill=black, inner sep=1.5pt] (y1) at (3.75, -1) {};

    \node[circle, fill=black, inner sep=1.5pt] (x2) at (4.75, -1) {};
    \node[circle, fill=black, inner sep=1.5pt] (y2) at (5.75, -1) {};

    \node[color=orange] at (4.25, 1.3) {$f$};
    \node[color=red] at (3.25, -1.3) {$e_1$};
    \node[color=red] at (5.25, -1.3) {$e_2$};

    % Connect the nodes with lines
    \draw[thick] (s) -- (c1) -- (d1) -- (c2) -- (d2) -- (c3) -- (d3) -- (t);

    \draw[thick, color=orange] (x) -- (y);
    \draw[thick, color=red] (x1) -- (y1);
    \draw[thick, color=red] (x2) -- (y2);    
    
    \draw[thick] (x) to[out=200, in=45] (c1);
    \draw[thick] (y) to[out=-20, in=135] (d3);
    \draw[thick] (y) to[out=220, in=30] (d1);
    \draw[thick] (x) to[out=-40, in=150] (c3);
    \draw[thick] (x) to[out=-70, in=120] (c2);
    \draw[thick] (y) to[out=-110, in=60] (d2);

    \draw[thick] (x1) to[out=80, in=240] (n1);
    \draw[thick] (y1) to[out=100, in=-60] (n2);
    \draw[thick] (x2) to[out=80, in=240] (n3);
    \draw[thick] (y2) to[out=100, in=-60] (n4);

\end{tikzpicture}
    \caption{In this example, we find the best edge $f$ maximizing $\spn(f)$ such that $|K_f| > 20$. Then, we make a shortcut between $u_{c_1}$ and $u_{d_1}$, and remove some intervals from $K_{e_1}$ and $K_{e_2}$.}
    \label{fig:shortcut}
\end{figure}
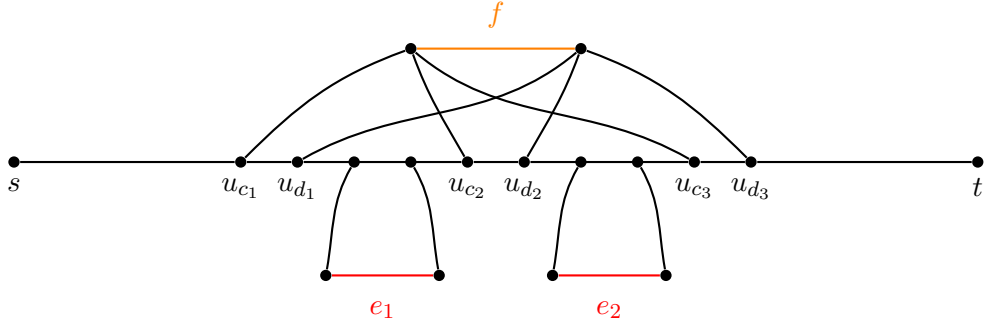

\begin{lemma}\label{shortcut-inv}
	After each iteration of updating $\paths$, $\interval$, and $\{K_e: e\in E(H)\setminus (F^\nw\cup F^\old)\}$, \Cref{inv} is preserved.
\end{lemma}
\begin{proof}
	Let us first verify \Cref{inv}(1). Consider any edge $e \in E(H)\setminus (F^\nw\cup F^\old)$ such that $|K_e| > 20$ at the beginning of the iteration. Since \Cref{inv}(1) held before this iteration, there exists $[a', b']\in \interval$ such that $K_e\subseteq [a', b']$. If $[a', b']\neq [a, b]$, then \Cref{inv}(1) continues to hold for $e$. 
	
	When $[a', b'] = [a, b]$, according to our update rules, for each interval $[i, j]\in K_e$, $|K_e| > 20$, which strictly contains $c_1$ or $d_l$, this interval must be contained entirely within $[i_1, j_2]$. Therefore, after the updates, all elements in $K_e$ are contained either in $[a, c_1]$ or $[d_l, b]$. To complement the argument for \Cref{inv}(1), it suffices to show that there cannot be two different elements $[i, j], [i', j']\in K_e$ such that $[i, j]\subseteq [a, c_1], [i', j']\subseteq [d_l, b]$. This is because we chose $f$ to be the maximizer of $\spn(f)$.
	
	Next, we mainly focus on \Cref{inv}(2). It suffices to upper bound the total amount of edge weight we remove from all the sets $K_e$ by $10\beta \brac{\wts(\rho[u_{c_1}, u_{d_l}]) - \wts(\lambda)}$. 
	
	To limit the total amount of edge weights that we remove when updating $\{K_e: e\in E(H)\setminus (F^\nw\cup F^\old)\}$, we distinguish between three cases for an interval $[i, j]$ that belonged to $K_e$, $|K_e| > 20$, before the update.
	\begin{itemize}[leftmargin=*]
		\item \emph{Case~1: $[i, j]\subseteq [i_1, j_1]$.} 		
		Notice that each such interval $[i, j]$ corresponds to an appearance of $e$ in the 
        multi-set of edges $P\left[u_{i_1}, u_{j_1}, \wts(\rho[u_{i_1}, u_{j_1}])\right]$ with total weight $\DP\left[u_{i_1}, u_{j_1}, \wts(\rho[u_{i_1}, u_{j_1}])\right]$. Since $(s^*, t^*, L^*)$ is the ratio maximizer, we know that
		$$\begin{aligned}
			\DP\left[u_{i_1}, u_{j_1}, \wts(\rho[u_{i_1}, u_{j_1}])\right] &\leq \beta\cdot \wts(\rho[u_{i_1}, u_{j_1}])\leq (1+\eps)\beta\cdot\wts(e_1)\\
			&\leq (1+\eps)\beta\cdot \frac{3(1+\eps)}{20}\spn(e_1)\\
			&\leq \frac{1}{5}\beta\cdot\spn(f) = \frac{1}{5}\beta\cdot\wts(\rho[u_{c_1}, u_{d_l}]) .
		\end{aligned}$$
        The last inequality is due to the selection of $f$, as $f$ was the edge such that $|K_f| > 20$ with maximum $\spn(f)$.
		
		\item \emph{Case~2: $[i, j]\subseteq [i_2, j_2]$.}		
		Symmetrically, we can show that the total weight of all such copies of edge $e$ is at most 
		$$\begin{aligned}
			\DP\left[u_{i_2}, u_{j_2}, \wts(\rho[u_{i_2}, u_{j_2}])\right] &\leq \beta\cdot \wts(\rho[u_{i_2}, u_{j_2}])\leq (1+\eps)\beta\cdot\wts(e_2)\\
			&\leq (1+\eps)\beta\cdot \frac{3(1+\eps)}{20}\spn(e_2)\\
			&\leq \frac{1}{5}\beta\cdot\spn(f) = \frac{1}{5}\beta\cdot\wts(\rho[u_{c_1}, u_{d_l}]) .
		\end{aligned}$$
		
		\item \emph{Case~3: $[i, j]\subseteq [j_1, i_2]$.} 		
		Let $\interval_1$ be the set of all such intervals $[i, j]$. Then, since all such intervals form a laminar family, we can find the set $\interval_2\subseteq \interval_1$ of maximal intervals. Assume $\interval_2 = \{[p_1, q_1], [p_2, q_2], \ldots, [p_z, q_z]\}$ with $q_i\leq p_{i+1}, 1\leq i<z$. Then any $[i, j]\subseteq [j_1, i_2]$ which belongs to some set $K_e, |K_e| > 20$ must be contained in an interval $[p, q]\in \interval_2$, which corresponds to one copy of $e$ in the multi-set $P\left[u_p, u_q, \wts(\rho[u_p, u_q])\right]$. Therefore, the total weight of edges $e$ corresponding to such intervals $[i, j]$ is bounded by
		$$\begin{aligned}
		    \sum_{o = 1}^z \DP\left[u_{p_o}, u_{q_o}, \wts(\rho[u_{p_o}, u_{q_o}])\right] &\leq \beta \cdot \sum_{o = 1}^z \wts(\rho[u_{p_o}, u_{q_o}])\\
            &\leq \beta\cdot\wts(\rho[u_{j_1}, u_{i_2}])\\
            &\leq \beta\cdot \wts(\rho[u_{c_1}, u_{d_l}])\\
            &< 5\beta\cdot \brac{\wts(\rho[u_{c_1}, u_{d_l}]) - \wts(\lambda)} .
		\end{aligned}$$
		
		Here, the first inequality holds because $\beta$ is the maximum ratio. As for the last inequality, we have
		$$\wts(\rho[u_{c_1}, u_{d_l}]) \geq \sum_{o=1}^l\wts(\rho[c_o, d_o]) \geq \frac{20}{3(1+\eps)}\cdot \wts(f) .$$
		On the other hand, since $f = (x, y)$ is $\frac{1}{3(1+\eps)}$-hanging at both $(u_{c_1}, u_{d_1})$ and $(u_{c_l}, u_{d_l})$, by the triangle inequality we get
		$$\wts(\lambda)\leq \dist_G(u_{c_1}, x) + \wts(f) + \dist_G(y, u_{d_l})\leq (3+2\eps)\wts(f) .$$
		Hence, we have $\wts(\rho[u_{c_1}, u_{d_l}]) \geq \frac{20}{3(1+\eps)(3+2\eps)}\wts(\lambda) > 2\wts(\lambda)$, and consequently we have:
		$$\wts(\rho[u_{c_1}, u_{d_l}]) - \wts(\lambda)\geq \frac{1}{5}\wts(\rho[u_{c_1}, u_{d_l}]) .$$
	\end{itemize}
        
	By the above case analysis, the total amount of edge weight we remove from all the sets $K_e$ is at most
    \begin{align*}
    	&\frac{2}{5}\beta\cdot \wts(\rho[u_{c_1}, u_{d_l}]) + 5\beta\cdot \brac{\wts(\rho[u_{c_1}, u_{d_l}]) - \wts(\lambda)}\\
    	&\leq \frac{2}{5}\beta\cdot\brac{\wts(\rho[u_{c_1}, u_{d_l}]}- \wts(\lambda)) + 5\beta\cdot \brac{\wts(\rho[u_{c_1}, u_{d_l}]) - \wts(\lambda)}\\
    	&<10\beta \brac{\wts(\rho[u_{c_1}, u_{d_l}]) - \wts(\lambda)}.
    \end{align*}
    Hence, \Cref{inv}(2) still holds.
\end{proof}

By \Cref{shortcut-inv}, we can repeatedly update the sets $\interval$, $\paths$, and $\{K_e : e\in E(H)\setminus (F^\nw\cup F^\old)\}$ until all the sets $K_e$ have size at most $20$. Then, we can derive a lower bound on $\sum_{e\in E(H)\setminus F^*}\mathbf{1}[J_e\neq\emptyset]\cdot \wts(e)$ as follows:
$$\begin{aligned}
	\sum_{e\in E(H)\setminus F^*}\mathbf{1}[J_e\neq\emptyset]\cdot \wts(e)&\geq \frac{1}{20}\sum_{e\in E(H)\setminus F^*}|K_e|\cdot \wts(e) \\
        &\geq \frac{1}{20}\brac{\sum_{e\in E(H)\setminus (F^\nw\cup F^\old)}|J_e|\cdot\wts(e) - \sum_{e\in E(H)\setminus F^*}|J_e\setminus K_e|\cdot \wts(e)}\\
	&\geq \frac{1}{20}\cdot \brac{\frac{\beta}{4}\wts(\rho) - \sum_{e\in E(H)\setminus F^*}|J_e\setminus K_e|\cdot \wts(e)}\\
	&\geq \frac{1}{20}\cdot \brac{\frac{\beta}{4}\wts(\rho) - 10\beta (\wts(\rho) - \wts(\eta))}\\
	&\geq \frac{1}{20}\cdot \brac{\frac{\beta}{4}\wts(\rho) - 10\beta\cdot \epsilon \cdot \wts(\rho)}\\
	& > \frac{\beta}{100}\wts(\rho) \geq \frac{\alpha}{600}\wts(\rho).
\end{aligned}$$
Recall that $\wts(\rho) - \wts(\eta)\leq (1+\eps)\dist_G(s, t) - \dist_G(s, t)\leq \epsilon \cdot\wts(\rho)$ and $\eps \leq 10^{-2}$. Therefore, when we update $F^\nw\leftarrow F^\nw\cup E(\rho[s^*, t^*, L^*])$ and $F^\old\leftarrow F^\old\cup P[s^*, t^*, L^*]$, then $\wts\brac{F^\nw}$ increases by $\wts(\rho)$ and $F^\old$ increases by at least $\frac{\alpha}{600}\wts(\rho)$, where $\alpha = \wts\brac{E(H)\setminus (F^\nw\cup F^\old)} / \wts\brac{G_{\opt, \eps}}$. Given this, the following statement helps to verify property (2) of \Cref{prune}.

\begin{lemma}
    Assume $\theta = \wts(H) / \wts(G_{\opt, \eps})$. Then, when \Cref{alg-prune} terminates, we have $\wts(H_1)\leq O(\log\theta)\cdot \wts(G_{\opt,\eps})$.
\end{lemma}
\begin{proof}
    During the course of \Cref{alg-prune}, let $\alpha = \wts(H \setminus F^\old) / \wts(G_{\opt,\eps})$. For any fixed integer $k$, whenever $\alpha\in [2^k, 2^{k+1})$, each iteration of the while-loop decreases $\wts(H\setminus F^\old)$ by $\Delta$ and increases $\wts(F^\nw)$ by at most $\frac{600}{2^k}\cdot\Delta$. Hence, while $\alpha\in [2^k, 2^{k+1})$, the weight $\wts(F^\nw)$ could increase by at most $\frac{600}{2^k}\cdot 2^{k+1}\cdot \wts(G_{\opt,\eps}) = O(1)\cdot \wts(G_{\opt,\eps})$. Therefore, in the end when $\beta<1$, we have $\wts(H\setminus (F^\nw\cup F^\old)) = \alpha\cdot \wts(G_{\opt,\eps})\leq 600\wts(G_{\opt,\eps})$ and $\wts(F^\nw)\leq O(\log\theta)\cdot \wts(G_{\opt,\eps})$, which finishes the proof.
\end{proof}

\subsubsection{Stretch Analysis}
Let us begin with a basic property of \Cref{alg-prune}.
\begin{lemma}\label{monotone}
    During \Cref{alg-prune}, when an edge $e\in E$ joins $F^\nw$, then it stays in $F^\nw$ until the end.
\end{lemma}
\begin{proof}
    This is evident because all the multi-sets $P[s, t, L]$ are contained in $E(H)\setminus (F^\nw\cup F^\old)$.
\end{proof}

Finally, let us analyze the stretch of graph $H_1 = F^\nw\cup (H\setminus (F^\nw\cup F^\old))$, proving property (1) of \Cref{prune}.
\begin{lemma}
	The stretch of graph $H_1 = F^\nw\cup (H\setminus F^\old)$ is at most $1+O(1)\cdot\delta$.
\end{lemma}
\begin{proof}
	Consider any pair of vertices $s, t\in V$. By assumption, we have $\dist_H(s, t)\leq (1+\delta)\cdot \dist_G(s, t)$. To bound the distance between $s$ and $t$ in $H_1$, let us conceptually maintain a short path $\rho$ in $H_1\cup H$ which is originally the shortest path between $s, t$ in $H$ and then gradually transform it to an $st$-path in $H_1$. To analyze the length of $\rho$, we use a potential function $\Phi(\rho)$ which is the total length of the edges in not in $H_1$. Initially, $\Phi(\rho)$ is at most $\wts(\rho)\leq (1+\delta)\cdot \dist_G(s, t)$.

    Let us iteratively update $\rho$ so that $\rho$ eventually belongs to $H_1$. While there is an edge $(u, v)\in E(\rho)$ not in $H_1$, by definition, $(u, v)$ must belong to $F^\old$. Consider the moment when $(u, v)$ was added to $F^\old$ by \Cref{alg-prune}. According to the algorithm description, at the moment there must exist a path $\gamma_{u, v}\subseteq E(F^\nw)$ (between vertices $x, y\in V$) such that $(u, v)$ is $\frac{1}{3(1+\eps)}$-hanging at $\gamma_{u, v}$. Let $\rho_1$ and $\rho_2$ be the shortest paths between $u, x$ and $v, y$ in graph $H$. By \Cref{hang}, we have
    $$\begin{aligned}
        \wts(\rho_1) + \wts(\gamma_{u, v}) + \wts(\rho_2) &\leq (1+\delta)\cdot \dist_G(u, x) + \wts(\gamma_{u, v}) + (1+\delta)\cdot\dist_G(y, v)\\
        &\leq (1+\delta)\cdot \brac{(1+\eps)\wts(u, v) - \wts(\gamma_{u,v})} + \wts(\gamma_{u, v})\\
        &\leq (1+\delta)(1+\eps)\cdot \wts(u, v) .
    \end{aligned}$$
    
    Update $\rho\leftarrow\rho[s, u]\circ \rho_1\circ \gamma_{u, v}\circ \rho_2\circ \rho[v, t]$, and so $\wts(\rho)$ would increase by at most
    $$(1+\delta)(1+\eps)\cdot \wts(u, v) - \wts(u, v) \leq (\delta+\eps + \delta\epsilon)\cdot\wts(u, v) < 3\delta\cdot\wts(u, v) .$$
    The key point is that all edges on $\gamma_{u, v}$ are in $F^\nw$ at the moment when $(u, v)$ joined $F^\old$, and by \Cref{monotone}, these edges will stay in $F^\nw$ til the end. Therefore, by replacing $\rho$ with $\rho[s, u]\circ \rho_1\circ \gamma_{u, v}\circ \rho_2\circ \rho[v, t]$, the value of $\Phi(\rho)$ has decreased by at least $\wts(\gamma_{u, v})\geq \frac{1}{3(1+\eps)}\wts(u, v)$, according to \Cref{hang}.
    
    As $\Phi(\rho)$ was originally at most $(1+\delta)\cdot \dist_G(s, t)$, the total amount of error increase would be bounded by
    $$3\delta\cdot 3(1+\eps)\cdot (1+\delta)\cdot\dist_G(s, t)\leq 10\delta\cdot \dist_G(s, t) .$$
    If follows that $\dist_{H_1}(s, t)\leq (1+11\delta)\cdot\dist_G(s, t)$.
\end{proof}

\subsection{Extension to Large Edge Weights}
In this subsection, let us discuss how to deal with general edge weights when $W \geq n^2/\epsilon$.
Recall that $G= (V, E, \wts)$ is an undirected weighted planar graph, where $\wts: E\rightarrow \{1, 2,\ldots, W\}$. We may assume that $W=\max_{e\in E}\wts(e)$, and for any edge $(u, v)\in E$, $\dist_G(u, v) = \wts(u, v)$, since otherwise we could remove $(u, v)$ from $E$. Under these  assumptions, we know that $\wts(G_{\opt, \eps})\geq W$, so we could always include all edges with weight less than $W/n$ in a spanner of weight $O(\wts(G_{\opt, \eps}))$.

To compute a spanner, contract all the connected components spanned by edges of weights less than $\epsilon W/n^2$, and denote the contracted graph by $G' = (V', E', \wts')$, where we round the edge weights as $\wts'(u, v) = \left\lfloor\wts(u, v)\cdot \frac{n^2}{W\epsilon}\right\rfloor$. Intuitively, the optimal spanner on $G'$ will be the same as the optimal spanner on the original graph $G$ since the total weight change is small. The advantage of the rounding procedure is that the maximum weight is now polynomial in $n$, and so the runtime would also be polynomial, according to the main algorithm.

Technically speaking, by definition of $G'$, edge weights $\wts'$ take integer values in $[1, n^2/\eps]$. Then, apply the main algorithm on graph $G'$ to obtain a $(1+\epsilon_0)$-spanner $H'$ in time $\poly(n, \eps^{-1})$ such that 
$$\wts'(H')\leq C\cdot \wts'(G'_{\opt, 2\eps}),$$
where $C = O(1)$, $\eps_0 = \eps\cdot 2^{O(\log^*1/\eps)}$ and $G'_{\opt,2\eps}$ is the optimal $(1+2\eps)$-spanner of $G'$. In the end, define 
$$E_0 = \{e\in E : \wts(e) \leq W/n\}$$ 
and return 
$$H = H'\cup E_0$$ 
as the approximate spanner of $G$.
Let us verify the stretch and weight of $H$ below.

\begin{lemma}
    For any $(s, t)\in E$, we have $\dist_H(s, t)\leq (1+\eps_0 + 2\eps)\cdot \wts(s, t)$.
\end{lemma}
\begin{proof}
    If $\wts(s, t)\leq W/n$, then by construction $(s, t)\in E(H)$, so $\dist_H(s, t) = \wts(s, t)$. Otherwise, we may assume $\wts(s, t) > W/n$. Since $H'$ is a $(1+\eps_0)$-spanner of $G'$, we have
    $$\dist_{H'}(s, t)\leq (1+\eps_0)\wts'(s, t)\leq (1+\eps_0)\wts(s, t)\cdot \frac{n^2}{W\eps}.$$

    Let $\pi'$ be the shortest path between $s$ and $t$ in the contracted graph $H'$. Unpack all the contracted nodes in $G'$, then $\pi'$ expands to a sequence of edges $(u_1, v_1), (u_2, v_2), \ldots, (u_k, v_k)\in E(H)$, where $v_{i-1}, u_{i}$ are in the same contracted node for all $1\leq i\leq k+1$ ($s = v_0, t = u_{k+1}$). Since $H$ includes all edges whose weights are at most $W/n$ under $\wts$, the unpacking increases the distance between $s$ and $t$ by at most $\sum_{i=1}^{k+1}\dist_H(v_{i-1}, u_i)<n\cdot \frac{W\epsilon}{n^2} = \frac{W\eps}{n}$. Therefore, overall we have
    $$\begin{aligned}
        \dist_H(s, t) &\leq \frac{W\eps}{n} + \sum_{i=1}^k\wts(u_i, v_i) \leq \frac{W\eps}{n} + \sum_{i=1}^k\frac{W\eps}{n^2}\cdot (\wts'(u_i, v_i)+1)\\
        &< (1+\eps_0)\wts(s, t) + \frac{2W\eps}{n} < (1+\eps_0 + 2\eps)\wts(s, t) . \qedhere
    \end{aligned}$$
\end{proof}

\begin{lemma}
    $\wts(H)\leq (4C+4)\cdot \wts(G_{\opt, \eps})$.
\end{lemma}
\begin{proof}
    Define $F = E(G_{\opt,\eps})\cup E_0$, and let $\widehat{G}$ be the graph obtained from $(V, F)$ by contracting all edges in $F$ whose weights are less than $\eps W/n^2$.

    We claim that $\widehat{G}$ is a $(1+2\eps)$-spanner of $G'$. In fact, for any edge $(s, t)\in E$ such that $\wts(s, t) > W/n$, we have
    $$\begin{aligned}
        \dist_{\widehat{G}}(s, t) &\leq \dist_{G_{\opt,\eps}}(s, t)\cdot \frac{n^2}{W\eps} \leq (1+\eps)\wts(s, t)\cdot \frac{n^2}{W\eps}\\
        &< (1+\eps)\cdot (\wts'(s, t)+1) \leq (1+2\eps)\cdot \wts'(s, t) .
    \end{aligned}$$
    The last inequality holds because $\wts(s, t) > W/n$ and $\wts'(s, t) = \left\lfloor \wts(s, t)\cdot \frac{n^2}{W\eps}\right\rfloor \geq \floor{n/\eps} \geq 1 + \eps^{-1}$.

    As $\widehat{G}$ is a $(1+2\eps)$-spanner of $G'$, we have
    $$\begin{aligned}
        \frac{n^2}{W\eps}\cdot\wts(G_{\opt,\eps}) + \wts'(E_0)&\geq\wts'(\widehat{G}) \geq \wts'(G'_{\opt,2\eps})\geq \frac{1}{C}\cdot \wts'(H')\\
        &\geq \frac{1}{C}\cdot \brac{\wts'(H\setminus E_0)}\\
        &\geq \frac{1}{C}\brac{\wts'(H) - \wts'(E_0)}\\
        &\geq \frac{1}{C}\brac{\frac{n^2}{W\eps}\brac{\wts(H)-|E(H)|}} - \frac{1}{C}\wts'(E_0).
    \end{aligned}$$
    Rearranging the terms and using $|E(H)|\leq |E| < 3n$ and $3n\leq W\leq \wts(G_{\opt,\eps})$, we obtain
    $$\begin{aligned}
        \wts(H) &< C\cdot \wts(G_{\opt,\eps}) + (C+1)\frac{W\eps}{n^2}\cdot\wts'(E_0) + 3n\\
        &\leq C\cdot \wts(G_{\opt,\eps}) + (C+1)\cdot \frac{W}{n}\cdot |E_0| + 3n\\
        &< (4C+4)\cdot\wts(G_{\opt,\eps}) .\qedhere
    \end{aligned}$$
\end{proof}
\section{Hardness for Planar Spanners}
\label{sec:hardness}

In this section, we prove \Cref{thm:lower}. For simplicity of notation, we assume that $\epsilon>0$ is rational.
The integer edge weight condition is achieved through proper scaling.

\paragraph{3SAT.} Given a set $X$ of $n$ Boolean variables $x_1, x_2, \ldots ,x_n$ and a Boolean formula $\phi$ in conjunctive normal form, where each clause has at most $3$ literals, the 3SAT problem is to determine whether there is an assignment of \texttt{true} or \texttt{false} values to the variables such that $\phi$ is satisfied (i.e., evaluates true). 

\paragraph{Incidence graph.} Given a 3SAT instance $\mathcal{I}$ with variable set $X$ and a Boolean formula $\phi = c_1 \land c_2 \land \ldots \land c_m$, the \emph{incidence graph} $G=G(\mathcal{I})$ corresponding to $\mathcal{I}$ is a bipartite graph with partite sets corresponding to all variables in $X$ and all clauses in $\phi$; there is an undirected edge $(x_i, c_j)$ in $G$ if and only if the clause $c_j$ contains $x_i$ or $\neg x_i$. The edge $(x_i, c_j)$ is a \emph{positive edge} if $c_j$ contains $x_i$ and is a \emph{negative edge} if $c_j$ contains $\neg x_i$. For each clause $c_j$, let $|c_j|$ be the number of variables in $c_j$. For each variable $x_i$, let $C^+_i$ be the set of clauses containing $x_i$ and $C^-_i$ be the set of clauses containing $\neg x_i$.

\paragraph{Planar 3SAT.} A 3SAT instance $\mathcal{I}$ is \emph{planar} if its incidence graph $G$ is planar.
Furthermore, $\mathcal{I}$ is \emph{planar rectilinear} if there exists a planar representation of $G$ where the vertices are represented by horizontal line segments (specifically, the vertices representing variables are on the $x$-axis while the segments representing clauses are above or below the $x$-axis); and the edges are represented by vertical segments. An instance $\mathcal{I}$ is \emph{planar rectilinear monotone} if all edges above the $x$-axis are positive and all edges below are negative (see an example in \Cref{fig:example}).

\begin{figure}[h!]
	\centering
	\begin{tikzpicture}[
  segment/.style={thick},
  varlabel/.style={font=\footnotesize, above},
  clauselabel/.style={font=\footnotesize},
  every node/.style={font=\footnotesize}
]

% Parameters
\def\varwidth{1.2}
\def\offset{0.3}
\def\varheight{0}

% Variable positions (x1 to x5)
\foreach \i/\x in {1/0, 2/2, 3/4, 4/6, 5/8} {
  \draw[segment] (\x - \varwidth/2, \varheight) -- (\x + \varwidth/2, \varheight);
  \node[varlabel, below] at (\x, \varheight) {$x_\i$};
}

% Clause C1 = x1 ∨ x2 ∨ x3
\def\ypos{2}
\draw[segment] (0, \ypos) -- (4.3, \ypos);
\node[clauselabel, above] at (2.0, \ypos) {$c_1 = x_1 \lor x_2 \lor x_3$};
\draw[segment] (0, \varheight) -- (0, \ypos); % x1 - offset
\draw[segment] (2.0, \varheight) -- (2.0, \ypos);   % x2 center
\draw[segment] (4.3, \varheight) -- (4.3, \ypos);   % x3 + offset

% Clause C2 = x1 ∨ x4 ∨ x5 (move x1 connection to the left to avoid C1)
\def\ypos{3.2}
\draw[segment] (-0.4, \ypos) -- (8.3, \ypos);
\node[clauselabel, above] at (4.0, \ypos) {$c_2 = x_1 \lor x_4 \lor x_5$};
\draw[segment] (-0.4, \varheight) -- (-0.4, \ypos);   % moved left of x1
\draw[segment] (6.0, \varheight) -- (6.0, \ypos);     % x4 center
\draw[segment] (8.3, \varheight) -- (8.3, \ypos);     % x5 + offset

% Clause C3 = ¬x1 ∨ ¬x3 ∨ ¬x4 (negative clause, below)
\def\ypos{-2}
\draw[segment] (0.3, \ypos) -- (5.7, \ypos);
\node[clauselabel, below] at (3.3, \ypos) {$c_3 = \lnot x_1 \lor \lnot x_3 \lor \lnot x_4$};
\draw[segment] (0.3, \varheight) -- (0.3, \ypos);   % x1 + offset
\draw[segment] (3.7, \varheight) -- (3.7, \ypos);   % x3 + offset
\draw[segment] (5.7, \varheight) -- (5.7, \ypos);   % x4 - offset

% Optional axis
%\draw[very thick, gray, ->] (-1,0) -- (9,0) node[right] {Variables};

\end{tikzpicture} 
	\caption{A planar representation of a planar rectilinear monotone 3SAT instance.}
	\label{fig:example}
\end{figure}
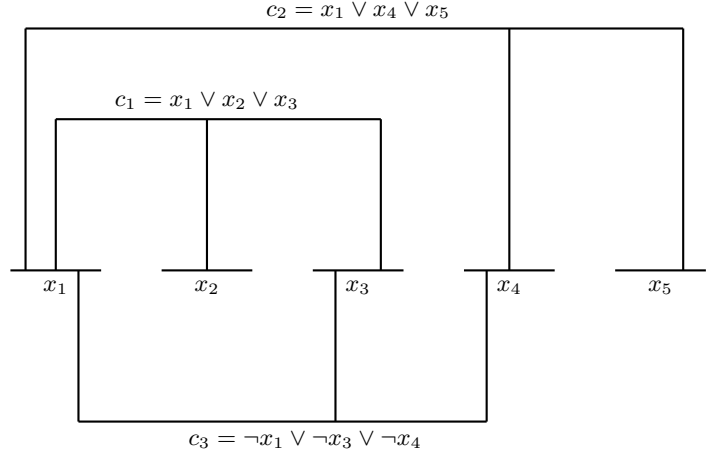

\begin{theorem}[\cite{BergK12}]
	Planar rectilinear monotone 3SAT is NP-hard.	
\end{theorem}

\paragraph{Weight-$k$ planar $(1 + \eps)$-spanner.} Given a planar graph $G = (V, E)$, is there a $(1 + \eps)$-spanner of $G$ of weight $k$?
In this section, 
we reduce the weight-$k$ planar $(1 + \epsilon)$-spanner problem to the planar rectilinear monotone 3SAT. 

\paragraph{Reduction.} Given an instance $\mathcal{I}$ of planar rectilinear monotone 3SAT.  Let $G'$ be the planar rectilinear drawing of the incidence graph of $\mathcal{I}$.  We construct an edge-weighted planar graph $G$ such that $G$ has a $(1 + \epsilon)$-spanner of weight at most $k$ (for some value of $k$ defined later) if and only if $\mathcal{I}$ is satisfiable. We can assume that each variable $x$ appears in both positive and negative clauses, since otherwise we can assign a \texttt{true}/\texttt{false} value to $x$ according to a clause it appears in and delete all clauses containing $x$ or $\neg x$.  

\paragraph{Clause gadget.} If clause $c_j$ has $3$ literals, we assume w.l.o.g.\ that $c_j = x_1 \land x_2 \land x_3$. 
In the drawing of $G'$, we assume that the edges corresponding to $(x_{1}, c_j)$, $(x_{2}, c_j)$ and $(x_{3}, c_j)$ appear from left to right. 
We then replace each clause segment with the gadget in \Cref{fig:clause_graph}.
Formally, we create a cycle $(e_j, l_{j, 1}, r_{j, 1}, l_{j, 2}, r_{j, 2}, l_{j, 3}, r_{j, 3}, f_j)$. For $h \in \{1, 2, 3\}$, the vertices $l_{j, h}$ and $r_{j, h}$ are connected to a gadget of the literal $x_{h}$, which will be discussed later, via two edges $(l_{j, h}, a_{j, h})$ and $(r_{j, h}, b_{j, h})$. The order of those edges is shown in \Cref{fig:clause_graph}. The weight of each edge $(l_{j, h}, r_{j, h})$ is $2 + 2\epsilon$. We set the weight of each edge $(l_{j, h}, a_{j, h})$ and $(r_{j, h}, b_{j, h})$ to be $\epsilon$. 
Set the weight of $(e_j, f_j)$ to $\frac{6 + 10\epsilon}{1 + \epsilon}$. All other edges have weight $0$. 
If clause $c_j$ contains two literals, say $x_{1}$ and $x_{2}$, then we construct the graph of the clause similar to the case with three literals, but without $l_{j, 3}$ and $r_{j, 3}$. The weight of $(l_{j, 1}, r_{j, 1})$ and $(l_{j, 2}, r_{j, 2})$ remains unchanged, while the weight of $(e_j, f_j)$ is $\frac{4 + 6\epsilon}{1 + \epsilon}$.
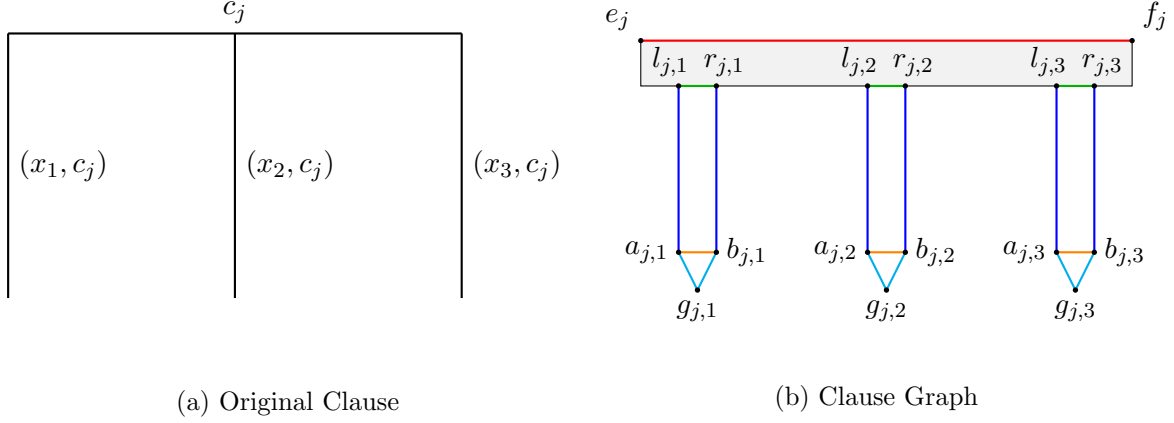
\begin{figure}[h]
  \centering
  \begin{subfigure}[t]{0.45\textwidth}
    \centering
    \vspace{0pt}
    \begin{tikzpicture}
   % Horizontal segment from (2,0) to (8,0)
  \draw[thick] (0,0) -- (6,0) node[midway, above] {$c_j$};

  % Perpendicular downward segments at both ends and the middle
  \draw[thick] (0,0) -- (0,-3.5) node[midway, right] {$(x_{1}, c_j)$};  % Left end
  \draw[thick] (3,0) -- (3,-3.5) node[midway, right] {$(x_{2}, c_j)$};  % Middle
  \draw[thick] (6,0) -- (6,-3.5) node[midway, right] {$(x_{3}, c_j)$};  % Right end
  \path (0,-3.5) -- (0,-4);
\end{tikzpicture}
    \caption{Original Clause}
    \label{fig:org}
  \end{subfigure}
  \hspace{0.1cm}
  \begin{subfigure}[t]{0.45\textwidth}
    \centering
    \vspace{0pt}
    \begin{tikzpicture}[scale=1]

  % Parameters
  \def\rectHeight{0.6}
  \def\rectWidth{6.5}
  \def\rectX{0}
  \def\rectY{0}
  \def\greenLen{0.5}       % length of l_i to r_i
  \def\blueDrop{2.2}       % length of blue vertical segment
  \def\triangleHeight{0.5} % vertical distance from midpoint to c_i

  % Rectangle coordinates
  \coordinate (e) at (\rectX, \rectY + \rectHeight);
  \coordinate (f) at (\rectX+\rectWidth, \rectY + \rectHeight);
  \draw[fill=gray!10, draw=black, thin] (\rectX,\rectY) rectangle ++(\rectWidth,\rectHeight);

  % Draw red top edge from e to f and mark vertices
  
  \draw[red, thick] (e) -- (f);
  \fill (e) circle (1pt) node[above left] {$e_j$};
  \fill (f) circle (1pt) node[above right] {$f_j$};

  % Three positions for green segments
  \foreach \i/\x in {1/\rectX + 0.5, 2/\rectX + 3, 3/\rectX + 5.5} {
    % Green segment on lower edge
    \coordinate (l\i) at (\x, \rectY);
    \coordinate (r\i) at ({\x + \greenLen}, \rectY);
    \draw[green!70!black, thick] (l\i) -- (r\i);

    % Mark l_i and r_i as vertices

    % Blue vertical lines to a_i and b_i
    \coordinate (a\i) at ($(l\i)-(0,\blueDrop)$);
    \coordinate (b\i) at ($(r\i)-(0,\blueDrop)$);
    \draw[blue, thick] (l\i) -- (a\i);
    \draw[blue, thick] (r\i) -- (b\i);

    % Orange horizontal segment
    \draw[orange, thick] (a\i) -- (b\i);

    % Isosceles right triangle
    \coordinate (m\i) at ($(a\i)!0.5!(b\i)$);
    \coordinate (c\i) at ($(m\i)-(0,\triangleHeight)$);

    \draw[cyan, thick] (a\i) -- (c\i);
    \draw[cyan, thick] (b\i) -- (c\i);
    \fill (l\i) circle (1pt) node[above, xshift = -3] {$l_{j, \i}$};
    \fill (r\i) circle (1pt) node[above, xshift = 3] {$r_{j, \i}$};
    % Mark a_i and b_i
    \fill (a\i) circle (1pt) node[left] {$a_{j, \i}$};
    \fill (b\i) circle (1pt) node[right] {$b_{j, \i}$};
    % Mark c_i
    \fill (c\i) circle (1pt) node[below] {$g_{j, \i}$};
  }

\end{tikzpicture}
    \caption{Clause Graph}
    \label{fig:clause_graph}
  \end{subfigure}
  \label{fig:clause}
  \caption{The clause graph of $c_j = x_1 \lor x_2 \lor x_3$}
\end{figure}

\paragraph{Literal gadget.} We now construct the gadget for each literal $x_i$.
Let $h_i = \max\{|C^+_i|, |C^-_i|\}$. Assume w.l.o.g.\ that $|C^+_i| \leq |C^-_i|$. 
The gadget contains two disjoint paths from a vertex $s_i$ to a vertex $t_i$, each path has length $4h_i$. The two paths enclose  a region. 
Inside the region, there is an edge $(s_i, t_i)$ of weight $4h_i/(1 + \epsilon)$. 
For each clause $c_j$, recall that $l_{j, i}, r_{j, i}$ are connected to the gadget of $x_i$ via two edges, and let $a_{j, i}, b_{j, i}$ be the other endpoints of those edges. We set $\wts(a_{j, i}, b_{j, i}) = 2$ and create a vertex $g_{j, i}$ inside the enclosed region of the two paths from $s_i$ to $t_i$ that is adjacent to $a_{j, i}$ and $b_{j, i}$. We set the weight of two edges $(g_{j, i}, a_{j, i})$ and $(g_{j, i}, b_{j, i})$ to be $1 + \epsilon$. 
We arrange all edges $(a_{j, i}, b_{j, i})$ corresponding to positive clauses on the upper path from $s_i$ to $t_i$ and all edges corresponding to negative clauses on the lower path as in \Cref{fig:literal}.
From left to right, the 
order of each $(a_{j, i}, b_{j, i})$ added to the path from $s_i$ to $t_i$ is the same as the order of their corresponding edges in the drawing of $G'$. 
If $c_j$ contains $x_i$, we call the edge $(a_{j, i}, b_{j, i})$ a true edge. Otherwise, $(a_{j, i}, b_{j, i})$ is a false edge.  
We append an edge of weight $2h_i$ as the last non-zero weight edge in both the upper and lower paths from  $s_i$ to $t_i$.
Call these edges $(u_i, u'_i)$ and $(v_i, v'_i)$.

\begin{figure}[h!]
	\centering
	\begin{tikzpicture}[scale=1.2]

  % === Shared parameters ===
  \def\segLength{0.8}
  \def\triangleInset{0.4}

  % === Literal Gadget ===
  \def\rectX{0}
  \def\rectY{0}
  \def\rectWidth{7}
  \def\rectHeight{2}
  \def\gap{1.6}
  \pgfmathsetmacro\midY{\rectY + \rectHeight/2}
  \pgfmathsetmacro\leftX{\rectX}
  \pgfmathsetmacro\rightX{\rectX + \rectWidth}
  \pgfmathsetmacro\yTop{\rectY + \rectHeight}
  \pgfmathsetmacro\yBottom{\rectY}

  % === Rectangle and red edge ===
  \draw[thin] (\rectX,\rectY) rectangle ++(\rectWidth,\rectHeight);
  \draw[red, very thick] (\leftX, \midY) -- (\rightX, \midY);
  \fill (\leftX, \midY) circle (1pt) node[left] {$s_i$};
  \fill (\rightX, \midY) circle (1pt) node[right] {$t_i$};

  % === First Clause Gadget (top left) ===
  \def\clauseRectHeight{0.6}
  \def\clauseRectWidth{4.5}
  \pgfmathsetmacro\sharedXStartAone{\rectX + 0.7}
  \pgfmathsetmacro\sharedXEndAone{\sharedXStartAone + \segLength}
  \pgfmathsetmacro\clauseRectYAone{\yTop + 0.8}
  \pgfmathsetmacro\clauseRectXAone{\sharedXEndAone - \clauseRectWidth + 0.05}

  \coordinate (e1) at (\clauseRectXAone, \clauseRectYAone + \clauseRectHeight);
  \coordinate (f1) at (\clauseRectXAone+\clauseRectWidth, \clauseRectYAone + \clauseRectHeight);
  \draw[fill=gray!10, draw=black, thin] (\clauseRectXAone,\clauseRectYAone) rectangle ++(\clauseRectWidth,\clauseRectHeight);
  \draw[very thick, red] (e1) -- node[midway, below, yshift=-5pt] {\textcolor{black}{$c_j$}} (f1);
  %\draw[very thick, red] (e1) -- (f1) ;
  % Top left clause gadget (was e', f')

  \coordinate (l1) at (\sharedXStartAone, \clauseRectYAone);
  \coordinate (r1) at (\sharedXEndAone, \clauseRectYAone);
  \draw[very thick, green!70!black] (l1) -- (r1);

  \coordinate (a1) at (\sharedXStartAone, \yTop);
  \coordinate (b1) at (\sharedXEndAone, \yTop);
  \draw[very thick, blue] (l1) -- (a1);
  \draw[very thick, blue] (r1) -- (b1);
  \fill (e1) circle (1pt);
  \fill (f1) circle (1pt);
  
  % === Second Clause Gadget (top right) ===
  \def\clauseRectHeightTwo{0.6}
  \def\clauseRectWidthTwo{6.5}
  \pgfmathsetmacro\sharedXStartAtwo{\rectX + 0.7 + 1 * \gap}
  \pgfmathsetmacro\sharedXEndAtwo{\sharedXStartAtwo + \segLength}
  \pgfmathsetmacro\clauseRectYAtwo{\yTop + 2.1}
  \pgfmathsetmacro\clauseRectXAtwo{(\sharedXStartAtwo + \sharedXEndAtwo)/2 - \clauseRectWidthTwo/2}

  \coordinate (e2) at (\clauseRectXAtwo, \clauseRectYAtwo + \clauseRectHeightTwo);
  \coordinate (f2) at (\clauseRectXAtwo+\clauseRectWidthTwo, \clauseRectYAtwo + \clauseRectHeightTwo);
  \draw[fill=gray!10, draw=black, thin] (\clauseRectXAtwo,\clauseRectYAtwo) rectangle ++(\clauseRectWidthTwo,\clauseRectHeightTwo);
  \draw[very thick, red] (e2) -- (f2);
  % Top right clause gadget (was e'', f'')
\fill (e2) circle (1pt);
\fill (f2) circle (1pt);

  \coordinate (l2) at (\sharedXStartAtwo, \clauseRectYAtwo);
  \coordinate (r2) at (\sharedXEndAtwo, \clauseRectYAtwo);
  \draw[very thick, green!70!black] (l2) -- (r2);

  \coordinate (a2) at (\sharedXStartAtwo, \yTop);
  \coordinate (b2) at (\sharedXEndAtwo, \yTop);
  \draw[very thick, blue] (l2) -- (a2);
  \draw[very thick, blue] (r2) -- (b2);

  % === Reflected Clause Gadget of First (bottom left) ===
  \pgfmathsetmacro\clauseReflectYOne{2*\midY - \clauseRectYAone - \clauseRectHeight}
  \coordinate (e1r) at (\clauseRectXAone, \clauseReflectYOne); % red edge is bottom
  \coordinate (f1r) at (\clauseRectXAone+\clauseRectWidth, \clauseReflectYOne);
  \draw[fill=gray!10, draw=black, thin] (\clauseRectXAone,\clauseReflectYOne) rectangle ++(\clauseRectWidth,\clauseRectHeight);
  \draw[very thick, red] (e1r) -- (f1r);
  % Bottom left reflected clause gadget (was e''', f''')
\fill (e1r) circle (1pt);
\fill (f1r) circle (1pt);

  \coordinate (l1r) at (\sharedXStartAone, \clauseReflectYOne + \clauseRectHeight);
  \coordinate (r1r) at (\sharedXEndAone, \clauseReflectYOne + \clauseRectHeight);
  \draw[very thick, green!70!black] (l1r) -- (r1r);

  \coordinate (a1r) at (\sharedXStartAone, \yBottom);
  \coordinate (b1r) at (\sharedXEndAone, \yBottom);
  \draw[very thick, blue] (l1r) -- (a1r);
  \draw[very thick, blue] (r1r) -- (b1r);

 % === Updated Bottom Right Clause Gadget (connected to a'_2 b'_2), lower position ===
\def\clauseRectHeightTwoR{0.6}
\def\clauseRectWidthTwoR{3.6}

\pgfmathsetmacro\sharedXStartAtwoR{\rectX + 0.7 + 1 * \gap}
\pgfmathsetmacro\sharedXEndAtwoR{\sharedXStartAtwoR + \segLength}

% Move rectangle further down
\pgfmathsetmacro\clauseRectYAtwoR{\yBottom - 2.9}
\pgfmathsetmacro\clauseRectXAtwoR{\sharedXStartAtwoR - 0.1}

% Red edge and rectangle
\coordinate (e2r) at (\clauseRectXAtwoR, \clauseRectYAtwoR);
\coordinate (f2r) at (\clauseRectXAtwoR+\clauseRectWidthTwoR, \clauseRectYAtwoR);
\draw[fill=gray!10, draw=black, thin] (\clauseRectXAtwoR,\clauseRectYAtwoR) rectangle ++(\clauseRectWidthTwoR,\clauseRectHeightTwoR);
\draw[very thick, red] (e2r) -- (f2r);
\fill (e2r) circle (1pt);
\fill (f2r) circle (1pt);

% Green edge on top of the rectangle
\coordinate (l2r) at (\sharedXStartAtwoR, \clauseRectYAtwoR + \clauseRectHeightTwoR);
\coordinate (r2r) at (\sharedXEndAtwoR, \clauseRectYAtwoR + \clauseRectHeightTwoR);
\draw[very thick, green!70!black] (l2r) -- (r2r);

% Extended blue connections to a'_2 and b'_2
\coordinate (a2r) at (\sharedXStartAtwoR, \yBottom);
\coordinate (b2r) at (\sharedXEndAtwoR, \yBottom);
\draw[very thick, blue] (l2r) -- (a2r);
\draw[very thick, blue] (r2r) -- (b2r);

% === Bottom Clause Gadget (connected to a'_3 b'_3) ===
\def\clauseRectHeightThreeR{0.6}
\def\clauseRectWidthThreeR{3.375} % reduced width

\pgfmathsetmacro\sharedXStartAthreeR{\rectX + 0.7 + 2 * \gap}
\pgfmathsetmacro\sharedXEndAthreeR{\sharedXStartAthreeR + \segLength}

% Keep l3r/r3r fixed, move rectangle RIGHT so they are near the left short side
\pgfmathsetmacro\clauseRectXAthreeR{\sharedXStartAthreeR - 0.1}
\pgfmathsetmacro\clauseRectYAthreeR{\yBottom - 1.5}

% Rectangle and red edge
\coordinate (e6) at (\clauseRectXAthreeR, \clauseRectYAthreeR);
\coordinate (f6) at (\clauseRectXAthreeR+\clauseRectWidthThreeR, \clauseRectYAthreeR);
\draw[fill=gray!10, draw=black, thin] (\clauseRectXAthreeR,\clauseRectYAthreeR) rectangle ++(\clauseRectWidthThreeR,\clauseRectHeightThreeR);
\draw[very thick, red] (e6) -- (f6);
\fill (e6) circle (1pt);
\fill (f6) circle (1pt);

% Green edge on top (fixed)
\coordinate (l3r) at (\sharedXStartAthreeR, \clauseRectYAthreeR + \clauseRectHeightThreeR);
\coordinate (r3r) at (\sharedXEndAthreeR, \clauseRectYAthreeR + \clauseRectHeightThreeR);
\draw[very thick, green!70!black] (l3r) -- (r3r);

% Blue edges to literal gadget (a'_3, b'_3)
\coordinate (a3r) at (\sharedXStartAthreeR, \yBottom);
\coordinate (b3r) at (\sharedXEndAthreeR, \yBottom);
\draw[very thick, blue] (l3r) -- (a3r);
\draw[very thick, blue] (r3r) -- (b3r);

% === Purple segments on top and bottom of main rectangle ===
\pgfmathsetmacro\purpleSegLength{1.0}
\pgfmathsetmacro\xPurpleStart{\rectX + \rectWidth - 1.5}
\pgfmathsetmacro\xPurpleEnd{\xPurpleStart + \purpleSegLength}

% Coordinates for top purple segment
\coordinate (u) at (\xPurpleStart, \yTop);
\coordinate (v) at (\xPurpleEnd, \yTop);
\draw[very thick, purple] (u) -- (v);
\fill (u) circle (1pt) node[above left] {$u_i$};
\fill (v) circle (1pt) node[above right] {$u_i'$};

% Coordinates for bottom purple segment
\coordinate (u') at (\xPurpleStart, \yBottom);
\coordinate (v') at (\xPurpleEnd, \yBottom);
\draw[very thick, purple] (u') -- (v');
\fill (u') circle (1pt) node[below left] {$v_i$};
\fill (v') circle (1pt) node[below right] {$v_i'$};

 % === Top and bottom segments and triangles ===
  \foreach \i in {0,1,2} {
    \pgfmathsetmacro\xstart{\rectX + 0.7 + \i * \gap}
    \pgfmathsetmacro\xend{\xstart + \segLength}
    \pgfmathsetmacro\xmid{(\xstart + \xend)/2}

    % Top triangle
    \draw[very thick, orange] (\xstart,\yTop) -- (\xend,\yTop);
    \pgfmathtruncatemacro\labeli{\i + 1}
    %\node[above left] at (\xstart,\yTop) {$a_{\labeli}$};
    %\node[above right] at (\xend,\yTop) {$b_{\labeli}$};
    \pgfmathsetmacro\cyTop{\yTop - \triangleInset}
    
    \draw[very thick, cyan] (\xstart,\yTop) -- (\xmid,\cyTop) -- (\xend,\yTop);
    %\node[below] at (\xmid,\cyTop) {$c_{\labeli}$};

    % Bottom triangle
    \draw[very thick, orange] (\xstart,\yBottom) -- (\xend,\yBottom);
    
    %\node[below left] at (\xstart,\yBottom) {$a'_{\labeli}$};
    %\node[below right] at (\xend,\yBottom) {$b'_{\labeli}$};
    \pgfmathsetmacro\cyBottom{\yBottom + \triangleInset}
    
    \draw[very thick, cyan] (\xstart,\yBottom) -- (\xmid,\cyBottom) -- (\xend,\yBottom);
    %\node[above] at (\xmid,\cyBottom) {$c'_{\labeli}$};
    \ifthenelse{\i = 0}{
        \fill (\xstart,\yTop) circle (1pt) node[yshift = 6pt, left]{$a_{j, i}$};
        \fill (\xend,\yTop) circle (1pt) node[yshift = 7pt, right]{$b_{j, i}$};
        \fill (\xmid, \cyTop) circle (1pt) node[below]{$g_{j, i}$};
    }{
        \fill (\xstart,\yTop) circle (1pt);
        \fill (\xend,\yTop) circle (1pt);
        \fill (\xmid, \cyTop) circle (1pt);
    }
        
    \fill (\xstart,\yTop) circle (1pt);
    \fill (\xend,\yTop) circle (1pt);
    \fill (\xmid, \cyBottom) circle (1pt);
    \fill (\xstart,\yBottom) circle (1pt);
    \fill (\xend,\yBottom) circle (1pt);
    \fill (\xmid, \cyTop) circle (1pt);
  }
\end{tikzpicture} 
	\caption{Literal Gadget}
	\label{fig:literal}
\end{figure}
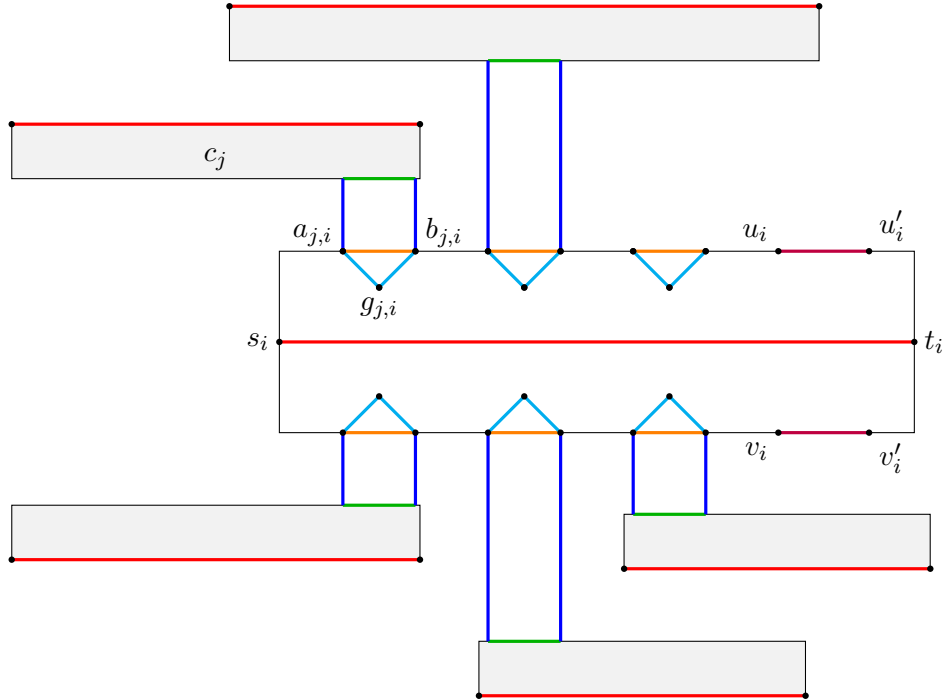

If the number of positive clauses is strictly less than the number of negative clauses, we add $|C^-_i| - |C^+_i|$ true edges, each of which does not connect to any clause gadgets in the upper path (see \Cref{fig:literal}). We connect true edges to form a path from $s_i$ to $t_i$ of length $4h_i$, that is, the edges connecting positive/negative edges have weight $0$.
We use $G[x_i]$ to denote the gadget corresponding to $x_i$ and $G[c_j]$ for the gadget corresponding to a clause $c_j$. 

The construction of $G$ takes polynomial time of the total number of clauses and literals in $\mathcal{I}$. The graph $G$ is planar by the planarity of $G'$. 
Let $W = 2\epsilon \sum_j |c_j| + 2(5 + 2\epsilon) \cdot \sum_i h_i$. We prove the following lemma:

\begin{lemma}
	$\mathcal{I}$ is satisfiable if and only if $G$ has a spanner of total weight at most $W$. 
\end{lemma}

For the necessity, assume that $\mathcal{I}$ admits a satisfying truth assignment for the  Boolean variables $x_1, x_2, \ldots ,x_n$. We then construct a $(1 + \epsilon)$-spanner of $G$ of weight at most $W$. 
\begin{lemma}
	If $\mathcal{I}$ is satisfiable, then $G$ has a $(1 + \epsilon)$-spanner of total weight at most $W$. 
\end{lemma} 

\begin{proof}
	Consider a truth assignment satisfying $\phi$. We initialize $H$ to be a subgraph containing all weight-$0$ edges in $G$. For each variable $x_i$, if $x_i =\texttt{true}$, we add all true edges in the literal gadget corresponding to $x_i$ to $H$. Otherwise, we add all false edges in $G[x_i]$ to $H$. For each true or false edge $(a_{j, i}, b_{j, i})$, we also add $(a_{j, i}, g_{j, i})$ and $(b_{j, i}, g_{j, i})$ to $H$.
	 
	 Then, for each clause $c_j$ containing a variable $x_i$, we add the edges $(l_{j, i}, a_{j, i})$ and $(r_{j, i}, b_{j, i})$ (blue and green edges in \Cref{fig:literal}). 
	 
	 We show that $H$ is a $(1 + \epsilon)$-spanner of $G$, that is, for every edge $(u, v)$ in $G$, we have $\dist_H(u, v) \leq (1 + \epsilon)\wts(u, v)$. Since $H$ contains all edges of weight $0$, we only need to consider the case when $\wts(u, v) > 0$.
	 
	 By our construction, we add all edges connecting clause gadgets to literal gadgets. Thus, we focus on the case that $(u, v)$ is either in a clause gadget or in a literal gadget. 
	 
	 Let $c_j = x_{1} \lor x_{2} \lor x_{3}$ be a clause with the clause gadget of $c_j$ as in \Cref{fig:clause_graph}. The case when $c_j$ contains only two literals can be solved similarly. Consider the edge $(l_{j, 1}, r_{j, 1})$. If $x_{1} = \texttt{true}$, we have
	 \begin{equation}
	 	\label{eq:true-literal}
	 	\begin{split}
	 		\dist_H(l_{j, 1}, r_{j, 1}) = \wts(l_{j, 1}, a_{j, 1}) + \wts(a_{j, 1}, b_{j, 1}) + \wts(b_{j, 1}, r_{j, 1}) = 2 + 2\epsilon.
	 	\end{split}
	 \end{equation}
	Else if $x_{1} = \texttt{false}$, then 
	\begin{equation}
		\label{eq:false-literal}
		\dist_H(l_{j, 1}, r_{j, 1}) = \wts(l_{j, 1}, a_{j, 1}) + \wts(a_{j, 1}, g_{j, 1}) + \wts(g_{j, 1}, b_{j, 1}) + \wts(b_{j, 1}, r_{j, 1}) = 2 + 4\epsilon \leq (1+\eps)\wts(l_{j, 1}, r_{j, 1}).
	\end{equation}
	
	In both cases, the distances between $(l_{j, 1}, r_{j, 1})$ is preserved up to $(1+\eps)$-factor in $H$. Using a similar argument, the distance between $(l_{j, 2}, r_{j, 2})$ and $(l_{j, 3}, r_{j, 3})$ is also preserved. We now prove that $\dist_H(e_j, f_j) \leq (1 + \epsilon)\wts(e_j, f_j)$. Recall that we do not add the edge $(e_j, f_j)$ to $H$. We have
	 \begin{equation*}
	 	\begin{split}
	 		\dist_H(e_j, f_j) =& \wts(e_j, l_{j, 1}) + \dist_H(l_{j, 1}, r_{j, 1}) + \wts(r_{j, 1}, l_{j, 2}) + \dist_H(l_{j, 2}, r_{j, 2})\\
            &+ \wts(r_{j, 2}, l_{j, 3}) + \dist_H(l_{j, 3}, r_{j, 3}) + \wts(r_{j, 3}, f_j)\\
	 		=& \dist_H(l_{j, 1}, r_{j, 1}) + \dist_H(l_{j, 2}, r_{j, 2}) + \dist_H(l_{j, 3}, r_{j, 3}).
	 	\end{split}
	 \end{equation*}
	 Observe that \Cref{eq:true-literal,eq:false-literal} yields $\dist_H(l_{j, 1}, r_{j, 1}), \dist_H(l_{j, 2}, r_{j, 2}), \dist_H(l_{j, 3}, r_{j, 3}) \in \{2 + 2\epsilon, 2 + 4\epsilon\}$. Since there is at least one literal in $\{x_{1}, x_{2}, x_{3}\}$ is true, there is at least one value, say $\dist_H(l_{j, 1}, r_{j, 1})$, equal to $2 + 2\epsilon$. Thus, 
	 \begin{equation*}
	 	\begin{split}
	 		\dist_H(e_j, f_j) &= \dist_H(l_{j, 1}, r_{j, 1}) + \dist_H(l_{j, 2}, r_{j, 2}) + \dist_H(l_{j, 3}, r_{j, 3})\\
	 		&\leq 2 + 2\epsilon + 2 \cdot (2 + 4\epsilon) = 6 + 10\epsilon = (1 + \epsilon)\wts(e_j, f_j). 
	 	\end{split}
	 \end{equation*}
	
	For each literal gadget, since there is a path from $s_i$ to $t_i$ of total weight $4h_i$ containing only true or false edges, the distance between $s_i$ and $t_i$ in $H$ is $4h_i$, and hence is $(1 + \epsilon)$ times the distance in $G$ which is $\wts(s_i, t_i) = 4h_i / (1+\eps)$. For every edge $(a_{j, i}, b_{j, i})$ in $G$, since we add $(a_{j, i}, c_{j, i})$ and $(b_{j, i}, g_{j, i})$ to $H$, $\dist_H(a_{j, i}, b_{j, i}) \leq \wts(a_{j, i}, g_{j, i}) + \wts(b_{j, i}, g_{j, i}) = 2 + 2\epsilon = (1 + \epsilon)\dist_G(a_{j, i}, b_{j, i})$.
	
	We now analyze the total weight of $H$. We distinguish between four types of edges in $H$:
	\begin{itemize}
		\item $(l_{j, i}, a_{j, i})$ and $(r_{j, i}, b_{j, i})$: Each of these edges has weight $\epsilon$. The total weight of those edges is $2\epsilon \cdot \sum_j |c_j|$.
		\item $(a_{j, i}, c_{j, i})$ and $(b_{j, i}, c_{j, i})$: Each literal gadget of $x_i$ has $4\max\{|C^+_i|, |C^-_i|\}$ edges of this type. The total weight of these edges is $\sum_i (1 + \epsilon) \cdot 4\max\{|C^+_i|, |C^-_i|\}$. 
		\item $(u_i, u'_i)$ and $(v_i, v'_i)$: Each literal gadget for $x_i$ has $2$ edges of this type. The total weight is then $4\cdot \max\{|C^+_i|, |C^-_i|\}$.
		\item $(a_{j, i}, b_{j, i})$: $H$ contains exactly $\max\{|C^+_i|, |C^-_i|\}$ edges of this type in each literal gadget. The total weight of this type is $\sum_i 2\max\{|C^+_i|, |C^-_i|\}$. 
	\end{itemize}
	The total weight of $H$ is
	$$2\epsilon \cdot \sum_j |c_j| + \sum_i (1 + \epsilon) \cdot 4\max\{|C^+_i|, |C^-_i|\} +  4\cdot \max\{|C^+_i|, |C^-_i|\} + \sum_i 2 \cdot \max\{|C^+_i|, |C^-_i|\}= W ,$$
	as desired. 
\end{proof}

We now show the converse direction.

\begin{lemma}
	\label{lm:inv}
	If $G$ has a $(1 + \epsilon)$-spanner of total weight at most $W$, then $\mathcal{I}$ is satisfiable. 
\end{lemma}

Observe that if $G$ has a spanner of total weight $W$, then there exists a $(1 + \epsilon)$-spanner of $G$ with the same total weight and containing all weight-$0$ edges. Hence, we can assume that any spanner contains all weight-$0$ edges.

Let $H$ be a $(1 + \epsilon)$-spanner of $G$ containing all weight-$0$ edges. We show that $H$ has weight at most $W$ only if $\mathcal{I}$ is satisfiable. 

\begin{observation}
	\label{obs:trivial-edge}
	$H$ contains all edges $(l_{j, i}, a_{j, i})$, $(r_{j, i}, b_{j, i})$, $(a_{j, i}, g_{j, i})$, $(b_{j, i}, g_{j, i})$, $(u_i, u'_i)$, and $(v_i, v'_i)$.
\end{observation}

Recall that $h_i = \max\{|C^+_j|, |C^-_j|\}$. We then show that each literal gadget has bounded weight. 
\begin{lemma}
	\label{lm:literal-weight}
	For every literal $x_i$, every $(1 + \epsilon)$-spanner $H$ of $G$ must contain a subgraph of $G[x_i]$ of weight at least $2h_i \cdot (5 + 2\epsilon).$ 
\end{lemma}

\begin{proof}
	Consider the gadget for $x_i$ as in \Cref{fig:literal}. Let $H[x_i]$ be the intersection of $H$ with $G[x_i]$. By \Cref{obs:trivial-edge}, for every clause $c_j$ containing either $x_i$ or $\neg x_i$, $H[x_i]$ must contain both edges $(a_{j, i}, g_{j, i})$ and $(b_{j, i}, g_{j, i})$. Furthermore, $H[x_i]$ also contains $(u_i, u'_i)$ and $(v_i, v'_i)$. The total weight of these edges is $4h_i + 2h_i \cdot (2 + 2\epsilon) = 4h_i \cdot(2 + \epsilon)$.  
	
	If $H[x_i]$ contains $(s_i, t_i)$, then the total weight of $H[x_i]$ is at least $4h_i/(1 + \epsilon) + 4h_i \cdot(2 + \epsilon) = 4h_i\left(2 + \epsilon + \frac{1}{1 + \epsilon}\right) > 2h_i \cdot (5 + 2\epsilon)$ given $\epsilon < 1$.
	
	If $H[x_i]$ does not contain $(s_i, t_i)$, then the only $(1 + \epsilon)$ path between $s_i$ and $t_i$ is either the path containing all true edges or the path containing all false edges in $G[x_i]$. Note that the total weight of true (resp., false) edges is $2h_i$. Hence, the total weight of $H[x_i]$ is at least $4h_i \cdot(2 + \epsilon) + 2h_i = 2h_i(5 + 2\epsilon)$, as claimed.  
\end{proof}

We can now prove \Cref{lm:inv}. 

\begin{proof}[Proof of \Cref{lm:inv}]
	Since $H$ must contain both edges $(l_{j, i}, a_{j, i})$ and $(r_{j, i}, b_{j, i})$ for all pairs $(j, i )$ when those edges exist, we have: 
	\begin{equation}
		\begin{split}
			\wts(H) &= 2\epsilon \cdot \sum_j |c_j| + \sum_i(\wts(H[x_i]))\\
			&\geq 2\epsilon \cdot \sum_j |c_j| + 2(5 + 2\epsilon)\cdot \sum_i h_i \qquad \text{(by \Cref{lm:literal-weight})}\\
			&= W.
		\end{split}
	\end{equation}
	Equality holds if and only if $H$ does not contain any edge of positive weight in $G[c_j]$ for all $j$ and the weight of each $H[x_i]$ achieves the minimum value $2h_i(5 + 2\epsilon)$ for all $i$. For the weight of each $H[x_i]$ to be $2h_i(5 + 2\epsilon)$, $H[x_i]$ does not contain the edge $(s_i, t_i)$, but rather a path of all true or false edges from $s_i$ to $t_i$. Observe that $H[x_i]$ cannot contain both true and false edges at the same time (otherwise the total weight is not minimum). Thus, if $H[x_i]$ contains a path of true edges from $s_i$ to $t_i$, we set $x_i = \texttt{true}$, otherwise, we set $x_i =  \texttt{false}$. We show that by this assignment, every clause is satisfied. 
	
	Consider a clause $c_j = x_{1} \lor x_{2} \lor x_{3}$ with $G[c_j]$ in \Cref{fig:clause_graph}. Since $H$ is a spanner of $G$, $\dist_H(e_j, f_j) \leq (1 + \epsilon)\wts(e_i, f_i) = 6 + 10\epsilon$.
	 Given that $H$ does not contain any edge in $G[c_j]$ of positive weight, we prove that $\dist_H(l_{j, 1}, r_{j, 1}) \in \{2 + 2\epsilon, 2+ 4\epsilon\}$. Observe that there are only two cases for the shortest path $P$ from $l_{j, 1}$ to $r_{j, 1}$ in $H$:
	\begin{itemize}
		\item $P$ contains $(a_{j, 1}, b_{j, 1})$. In this case, $\dist_H(l_{j, 1}, r_{j, 1}) = 2\epsilon + \wts(a_{j, 1}, b_{j, 1}) = 2\epsilon + 2$. 
		\item $P$ contains $(a_{j, 1}, c_{j, 1})$ and $(b_{j, 1}, c_{j, 1})$. In this case, $\dist_H(l_{j, 1}, r_{j, 1}) = 2\epsilon + \wts(a_{j, 1}, c_{j, 1}) + \wts(b_{j, 1}, c_{j, 1}) = 2 + 4\epsilon$. 
	\end{itemize}
	Similarly, we obtain the same result for $\dist_H(l_{j, 2}, r_{j, 2})$ and $\dist_H(l_{j, 3}, r_{j, 3})$. On the other hand, observe that $\dist_H(e_j, f_j) = \dist_H(l_{j, 1}, r_{j, 1}) + \dist_H(l_{j, 2}, r_{j, 2}) + \dist_H(l_{j, 3}, r_{j, 3})$ as all other edges in the path from $e_j$ to $f_j$ have weight $0$. If $\dist_H(l_{j, 1}, r_{j, 1}) = \dist_H(l_{j, 2}, r_{j, 2}) = \dist_H(l_{j, 3}, r_{j, 3}) = 2 + 4\epsilon$, $\dist_H(e_i, f_i) = 6 + 12\epsilon > (1 + \epsilon) \wts(e_i, f_i) = 6 + 10\epsilon$. Then,  at least one of the edges among $(a_{j, 1}, b_{j, 1}), (a_{j, 2}, b_{j, 2}), (a_{j, 3}, b_{j, 3})$ is present in $H$. Since if one edge, say $(a_{j, 1}, b_{j, 1})$, appears in $H$, we set $x_{1}$ to be \texttt{true} if $(a_{j, 1}, b_{j, 1})$  is a true edge and false otherwise. In both cases, $c_j$ is satisfied. 
\end{proof}

\section{A Hard Instance for the Greedy Algorithm}\label{sec:greedy}

In this section, we adapt the hard instance for Euclidean spaces of  \cite{le2024towards}  to give a hard instance for planar graphs where greedy spanners have heavy total weight compared with optimal $(1+\eps)$-spanners, even when relaxing the stretch of the greedy spanner to be $1+x\eps$ for some large $x\gg 1$. Recall that the greedy $t$-spanner $H$ for an edge-weighted graph $G=(V,E,\wts)$ is constructed as follows~\cite{althofer1993sparse}: 
Initialize an empty graph $H=(V,\emptyset)$, and then consider the edges $(u,v)\in E$ sorted by nondecreasing weight (ties are broken arbitrarily), and if $\dist_H(u,v)> t\cdot\dist_G(u,v)$ in the current spanner $H$, then add $(u,v)$ to $H$.

\begin{theorem}
For every $\epsilon \in (0, \frac14 )$ and every parameter $1\leq x \leq \frac{1}{2}\sqrt{1/\epsilon}$, there exists an undirected planar graph $G = (V, E, \wts)$ with positive edge weights such that $\wts(G_{\gr(x)})\geq \Omega\brac{\frac{\epsilon^{-1}}{x^2}}\cdot \wts(G_{\opt, \eps})$, where $G_{\gr(x)}$ is the greedy $(1+x\epsilon)$-spanner of graph $G$.
\end{theorem}
\begin{proof}
	This construction is adapted from a Euclidean example in \cite{le2024towards}. The graph $G$ is built as following.
	\begin{itemize}
		\item \textbf{Vertices.} Let $V = \{w\}\cup \{x_1, x_2, \ldots, x_n\}\cup \{y_1, y_2, \ldots, y_n\}$ be a vertex set with $2n+1$ vertices, where $n = 1+\ceil{\frac{1}{x} + \frac{\epsilon^{-1}}{2x^2}}$.
		\item \textbf{Edges.} For each $1\leq i\leq n$, add edge $(x_i, y_i)$ with weight $1$. Connect all vertices via a path $\pi = \langle x_1, x_2, \ldots, x_n, w, y_1, y_2, \ldots, y_n\rangle$. For each $1\leq i < n$, set the weight of $(x_i, x_{i+1})$ and $(y_i, y_{i+1})$ to be $x\epsilon$, and add edge $(x_n, y_1)$ with weight
        $$\wts(x_n, y_1) = 1+\epsilon - (n-1)x\epsilon \in \left(1 - \frac{1}{2x} -x\epsilon, 1-\frac{1}{2x}\right] .$$
        Finally, set the weights of $(x_n, w)$ and $(w, y_1)$ to be $\frac{1}{2}\cdot(1+x\epsilon)\cdot \wts(x_n, y_1)$.
        
	\end{itemize}
        One can verify that $G$ is a planar graph; check \Cref{fig:greedy-drawing} for a planar drawing.
	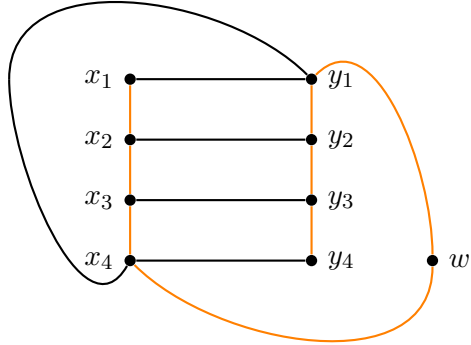
\begin{figure}[hpt]
		\centering
		\begin{tikzpicture}[scale=0.8]
    % Define the style for the nodes

    % Place the nodes
    \node[circle, fill=black, inner sep=1.5pt, label=left:$x_1$] (x1) at (0,0) {};
    \node[circle, fill=black, inner sep=1.5pt, label=left:$x_2$] (x2) at (0,-1) {};
    \node[circle, fill=black, inner sep=1.5pt, label=left:$x_3$] (x3) at (0,-2) {};
    \node[circle, fill=black, inner sep=1.5pt, label=left:$x_4$] (x4) at (0,-3) {};

    \node[circle, fill=black, inner sep=1.5pt, label=right:$y_1$] (y1) at (3,0) {};
    \node[circle, fill=black, inner sep=1.5pt, label=right:$y_2$] (y2) at (3,-1) {};
    \node[circle, fill=black, inner sep=1.5pt, label=right:$y_3$] (y3) at (3,-2) {};
    \node[circle, fill=black, inner sep=1.5pt, label=right:$y_4$] (y4) at (3,-3) {};

    \node[circle, fill=black, inner sep=1.5pt, label=right:$w$] (w) at (5,-3) {};
    \coordinate (n) at (-2, 0);
    
    \draw[thick] (x1) to (y1);
    \draw[thick] (x2) to (y2);
    \draw[thick] (x3) to (y3);
    \draw[thick] (x4) to (y4);
    
    \draw[thick, color=orange] (x1) to (x2);
    \draw[thick, color=orange] (x2) to (x3);
    \draw[thick, color=orange] (x3) to (x4);
    
    \draw[thick, color=orange] (y1) to (y2);
    \draw[thick, color=orange] (y2) to (y3);
    \draw[thick, color=orange] (y3) to (y4);
    
    \draw[thick, color=orange] (y1) to[out=45, in=90] (w);
    \draw[thick, color=orange] (w) to[out=-90, in=-45] (x4);
    \draw[thick] (x4) to[out=-120, in=-90] (n);
    \draw[thick] (n) to[out=90, in=135] (y1);    
\end{tikzpicture}
		\caption{A planar drawing of the hard example against the greedy algorithm when $n =4$. The path $\pi$ is highlighted as the orange path.}
		\label{fig:greedy-drawing}
	\end{figure}
    
	First, we study the total weight of $G_{\opt, \eps}$. Consider the graph $H = \pi\cup \{(x_n, y_1)\}$. Then, for any edge $(x_i, y_i)$, by definition of $n$, we have
	$$\dist_H(x_i, y_i)\leq (n-1)x\epsilon + \wts(x_n, y_1)\leq 1+\epsilon = (1+\epsilon)\wts(x_i, y_i) .$$
    Therefore, $H$ is a $(1+\epsilon)$-spanner, which implies: $$\begin{aligned}
        \wts(G_{\opt,\epsilon})\leq \wts(H) &= 2(n-1)\cdot x\epsilon + (1+x\epsilon)\wts(x_n, y_1) +\wts(x_n, y_1)\\
        &\leq 2\brac{\epsilon + \frac{1}{2x}} + (2+x\epsilon)\cdot \brac{1 - \frac{1}{2x}}\\
		&< 2\epsilon + 1/x + 2+x\epsilon < 4 .
	\end{aligned}$$
	Next, we study the total weight of $G_{\gr(x)}$. The greedy algorithm would first go over all edges on path $\pi$ and add them to $G_{\gr(x)}$. After that, when it comes to the edge $(x_n, y_1)$, since $\wts(x_n, w) + \wts(w, y_1) = (1+x\epsilon)\wts(x_n, y_1)$, it would not include edge $(x_n, y_1)$.
	
	Finally, it makes a pass over all edges $(x_i, y_i)$ for all $1\leq i\leq n$. We argue that the greedy algorithm must add edge $(x_i, y_i)$ in $G_{\gr(x)}$, namely $\dist_{G_{\gr(x)}}(x_i, y_i) > 1+x\epsilon$ for the current $G_{\gr(x)}$. Consider any path $\rho$ between $x_i$ and $y_i$ in the current version of $G_{\gr(x)}$. If $\rho = \pi[x_i, y_i]$, then we have
	$$\begin{aligned}
		\wts(\pi[x_i, y_i]) &= (n-1)x\epsilon + (1+x\epsilon)\wts(x_n, y_1)\\
		& = 1+\epsilon + x\epsilon\cdot\wts(x_n, y_1)\\
        &> 1+\epsilon + x\epsilon - (x^2\epsilon^2 + \epsilon/2)\\
        &>1+x\epsilon + (\epsilon / 2 - x^2\epsilon^2) \geq 1+x\epsilon .
	\end{aligned}$$
	Otherwise if $\rho\neq \pi[x_i, y_i]$, $\rho$ would contain the edge $(x_k, y_k)$ for some $1\leq k<i$ along with the edges $(x_k, x_{k+1})$ and $(y_k, y_{k+1})$. Since all edges $(x_j, x_{j+1}), (y_j, y_{j+1})$ have weight $x\epsilon$, the total weight of $\rho$ would be at least $1 + 2x\epsilon > 1+x\epsilon$.
	
	In the end, $G_{\gr(x)}$ would include the edges $(x_i, y_i)$ for all $1\leq i \leq n$, so $\wts(G_{\gr(x)}) > n$. Hence, the approximation ratio is at least
	$$\frac{\wts\brac{G_{\gr(x)}}}{\wts(G_{\opt,\epsilon})} > \frac{n}{4} > \frac{\epsilon^{-1}}{8x^2} . \qedhere$$
\end{proof}

\paragraph{Acknowledgements.~}
Hung Le and Cuong Than are supported by the NSF CAREER award CCF-2237288, the NSF grants CCF-2517033 and CCF-2121952, and a Google Research Scholar Award. Research by Csaba D.\ T\'oth was supported by the NSF award DMS-2154347. Shay Solomon is funded by the European Union (ERC, DynOpt, 101043159). Views and opinions expressed are however those of the author(s) only and do not necessarily reflect those of the European Union or the European Research Council. Neither the European Union nor the granting authority can be held responsible for them. Shay Solomon is also funded by a grant from the United States-Israel Binational Science Foundation (BSF), Jerusalem, Israel, and the United States National Science Foundation (NSF). Work of Tianyi Zhang was done while at ETH Z\"urich when supported by funding from the starting grant ``A New Paradigm for Flow and Cut Algorithms'' (no. TMSGI2\_218022) of the Swiss National Science Foundation. 
\vspace{5mm}

\bibliographystyle{alphaurl}
\bibliography{ref,ref2}

\end{document}